\documentclass[11pt]{article}

\usepackage{amsmath,amsthm,amssymb}
\usepackage{mathtools}
\usepackage{fullpage}
\usepackage{hyperref}

\DeclarePairedDelimiter\abs{\lvert}{\rvert}
\DeclarePairedDelimiter\len{\lvert}{\rvert}
\DeclarePairedDelimiter\set{\{}{\}}
\DeclarePairedDelimiter\parens{(}{)}

\DeclarePairedDelimiter\symnorm{\|}{\|_{*}}
\DeclarePairedDelimiter\symnormmax{\|}{\|_{*,\infty}}

\newtheorem{theorem}{Theorem}[section]
\newtheorem*{theorem*}{Theorem}

\newtheorem{proposition}[theorem]{Proposition}
\newtheorem*{proposition*}{Proposition}
\newtheorem{lemma}[theorem]{Lemma}
\newtheorem*{lemma*}{Lemma}
\newtheorem{corollary}[theorem]{Corollary}
\newtheorem*{conjecture*}{Conjecture}
\newtheorem{fact}[theorem]{Fact}
\newtheorem*{fact*}{Fact}

\newtheorem*{hypothesis*}{Hypothesis}

\newtheorem{itheorem}[theorem]{Informal Theorem}
\newtheorem{claim}[theorem]{Claim}
\newtheorem*{claim*}{Claim}

\theoremstyle{definition}
\newtheorem{definition}[theorem]{Definition}

\newtheorem{question}[theorem]{Question}
\newtheorem{openquestion}[theorem]{Open Question}

\theoremstyle{remark}
\newtheorem{remark}[theorem]{Remark}
\newtheorem*{remark*}{Remark}

\newtheorem*{observation*}{Observation}

\newcommand{\R}{\mathbb{R}}

\newcommand{\Z}{\mathbb{Z}}

\newcommand{\I}{\mathbb{I}}

\newcommand{\calD}{\mathcal{D}}
\newcommand{\calF}{\mathcal{F}}
\newcommand{\calG}{\mathcal{G}}
\newcommand{\calN}{\mathcal{N}}

\newcommand{\calT}{\mathcal{T}}
\newcommand{\calX}{\mathcal{X}}

\newcommand{\poly}{\mathrm{poly}}

\newcommand{\Paren}[1]{\left(#1\right)}

\newcommand{\Abs}[1]{\left\lvert#1\right\rvert}

\newcommand{\norm}[1]{\lVert #1 \rVert}

\newcommand{\bignorm}[1]{\big\lVert#1\big\rVert}
\newcommand{\Bignorm}[1]{\Big\lVert#1\Big\rVert}

\newcommand{\iprod}[1]{\langle#1\rangle}
\newcommand{\Iprod}[1]{\left\langle#1\right\rangle}

\newcommand{\Vol}{\text{Vol}}

\newcommand{\Esymb}{\mathbb{E}}
\newcommand{\Psymb}{\mathbb{P}}

\DeclareMathOperator*{\E}{\Esymb}

\DeclareMathOperator*{\ProbOp}{\Psymb}

\renewcommand{\Pr}{\ProbOp}

\newcommand{\grad}{\nabla}

\newcommand{\dy}{\, d{y}}
\newcommand{\dr}{\, d{r}}

\newcommand{\tr}{\text{tr}}

\newcommand{\tmu}{\tilde{\mu}}
\newcommand{\tsigma}{\tilde{\sigma}}
\newcommand{\tw}{\tilde{w}}
\newcommand{\tb}{\tilde{b}}

\newcommand{\tA}{\tilde{A}}
\newcommand{\tF}{\tilde{F}}
\newcommand{\tsf}{\tilde{f}}

\newcommand{\tcalG}{\tilde{\calG}}

\newcommand{\vx}{{\bf x}}

\newcommand{\vsig}[1]{{ {\xi}_{#1}{\scriptstyle (\sigma)} }}

\newcommand{\smax}{\sigma_{\max}}
\newcommand{\smin}{\sigma_{\min}}
\newcommand{\rhos}{\rho_{\sigma}}
\newcommand{\rhow}{\rho_w}

\newcommand{\eps}{\varepsilon}
\renewcommand{\epsilon}{\varepsilon}

\newcommand{\dx}{dx}
\newcommand{\im}{\mathrm{i}}
\newcommand{\wmin}{w_{\text{min}}}

\newcommand{\dparam}{\Delta_{\mathrm{param}}}

\newcommand{\ft}[1]{\widehat{#1}}
\newcommand{\phitail}[1]{\tilde{\Phi}_{#1}}
\newcommand{\phit}{\phitail{0,\sigma}}
\newcommand{\iphit}[1]{\tilde{\Phi}^{-1}_{0,\sigma}(#1)}

\newcommand{\cmoments}{c_0}

\newcommand{\basin}{c_0 d^{-5/2}}

\newcommand{\ehat}{\widehat{e}}
\newcommand{\ejl}{\ehat_{j \ell}}
\newcommand{\ejj}{\ehat_{jj}}

\newcommand{\ctails}{}

\newcommand{\sqsigsub}[1]{\vsig{i}^2 \tau^2_i}

\newcommand{\cI}{\mathcal{I}}
\newcommand{\ev}{\widehat{v}}
\newcommand{\er}[1]{\widehat{e}_{r_{#1}}}
\newcommand{\eji}{\widehat{e}_{ji}}

\newif\ifnotes\notesfalse

\ifnotes
\usepackage{color}
\definecolor{mygrey}{gray}{0.50}
\newcommand{\notename}[2]{{\textcolor{mygrey}{\footnotesize{\bf (#1:} {#2}{\bf ) }}}}
\newcommand{\noteswarning}{{\begin{center} {\Large WARNING: NOTES ON}\end{center}}}

\else

\newcommand{\notename}[2]{{}}
\newcommand{\noteswarning}{{}}

\fi

\newcommand{\onote}[1]{{\notename{Oded}{#1}}}
\newcommand{\anote}[1]{{\notename{Aravindan}{#1}}}

\title{On Learning Mixtures of Well-Separated Gaussians}
\author{ 
	Oded Regev\thanks{Courant Institute of Mathematical Sciences, New York
 University. Supported by the Simons Collaboration on Algorithms and Geometry and by the National Science Foundation (NSF) under Grant No.~CCF-1320188.}
\and
Aravindan Vijayaraghavan\thanks{
  Department of Electrical Engineering and Computer Science,
  Northwestern University. Supported by the National Science Foundation (NSF) under Grant No.~CCF-1652491 and CCF-1637585. Any opinions, findings, and conclusions or recommendations expressed in this material are those of the authors and do not necessarily reflect the views of the NSF.} 
 }
\date{}

\begin{document}

\maketitle

\noteswarning

\begin{abstract}
We consider the problem of efficiently learning mixtures of 
a large number of 
spherical Gaussians, when the components of the mixture are well separated.
In the most basic form of this problem, we are given samples from a uniform mixture of $k$ standard spherical Gaussians with 
means $\mu_1,\ldots,\mu_k \in \R^d$,
and the goal is to estimate the means up to accuracy $\delta$ using $\poly(k,d, 1/\delta)$ samples.

In this work, we study the following question: what is the minimum separation needed between the means for 
solving this task?
The best known algorithm due to Vempala and Wang [JCSS 2004] requires a separation of 
roughly $\min\set{k,d}^{1/4}$. 
On the other hand, Moitra and Valiant [FOCS 2010] showed that with separation $o(1)$,
exponentially many samples are required. 
We address the significant gap between these two bounds, by showing the following results.

\begin{itemize}
\item We show that with separation $o(\sqrt{\log k})$, super-polynomially many samples are required. 
In fact, this holds even when the $k$ means of the Gaussians are picked at random in $d=O(\log k)$ dimensions. 
 

\item 
We show that with separation $\Omega(\sqrt{\log k})$,
$\poly(k,d,1/\delta)$ samples suffice. Notice that the bound on the separation is independent of $\delta$. 
This result is based on a new and efficient ``accuracy boosting'' algorithm that 
takes as input coarse estimates of the true means 
and in time (and samples) $\poly(k,d, 1/\delta)$ outputs estimates of the means up to 
arbitrarily good accuracy $\delta$ assuming the separation between the means 
is $\Omega(\min\set{\sqrt{\log k},\sqrt{d}})$ (independently of $\delta$). 
The idea of the algorithm is to iteratively solve a ``diagonally dominant'' system of non-linear equations.
\end{itemize}
We also (1) present a \emph{computationally efficient} algorithm in $d=O(1)$ dimensions with only $\Omega(\sqrt{d})$ separation,
and (2) extend our results to the case that components might have different weights and variances.
These results together essentially characterize the optimal order of separation between components that is needed to learn a mixture of $k$ spherical Gaussians with polynomial samples.
\end{abstract}

\section{Introduction}


Gaussian mixture models are one of the most widely used statistical models for clustering. 
In this model, we are given random samples, where each sample point $x \in \R^d$ is drawn independently from one of $k$ Gaussian components according to mixing weights $w_1, w_2, \dots, w_k$, where each Gaussian component $j \in [k]$ has a mean $\mu_j \in \R^d$ and a 
covariance $\Sigma_j \in \R^{d \times d}$. 
We focus on an important special case of the problem where each of the components is a \emph{spherical} Gaussian, i.e., the covariance matrix of each component is a multiple of the identity. 
If $f$  represents the p.d.f.\ of the Gaussian mixture $\calG$, and $g_j$ represents the p.d.f.\ of the $j$th Gaussian component,
$$g_j = \frac{1}{\sigma_j^d}\exp\left(- \pi \norm{x - \mu_j}_2^2 / \sigma_j^2\right), ~ f(x)=\sum_{j=1}^k w_j g_j(x).$$
The goal is to estimate the parameters $\set{(w_j, \mu_j, \sigma_j): j \in [k]}$ up to required accuracy $\delta>0$ in time and number of samples that is polynomial in $k,d, 1/\delta$.

Learning mixtures of Gaussians has a long and rich history, starting with the work of Pearson \cite{Pea94}.
(See Section~\ref{sec:priorwork} for an overview of prior work.)
Most of the work on this problem, especially in the early years but also recently, is 
under the assumption that there is some minimum \emph{separation} between the means of the components in the mixture. 
Starting with work by Dasgupta~\cite{Das99}, and continuing with a long line of work (including~\cite{AK01, VW04, AM05, KSV08, Srebro06,DS07, BV08,KK10, AS12,BWY,Hsuetal16,DTZ16}),
efficient algorithms were found under mild separation assumptions. 
Considering for simplicity the case of uniform mixtures (i.e., all weights are $1/k$) of standard Gaussians (i.e., spherical with $\sigma=1$),
the best known result due to Vempala and Wang~\cite{VW04} provides an efficient algorithm (both in terms of samples and running time)
under separation of at least $\min\set{k,d}^{1/4} \poly\log(dk/\delta)$  between any two means.

A big open question in the area is whether efficient algorithms exist under weaker separation assumptions. 
It is known that when the separation is $o(1)$, a super-polynomial number of samples is required (e.g.,~\cite{MV10,ABGRV14,HP15}),
but the gap between this lower bound and the above upper bound of $\min\set{k,d}^{1/4}\poly\log(dk/\delta)$ is quite wide. 
Can it be that efficient algorithms exist under only $\Omega(1)$ separation? In fact,
prior to this work, this was open even in the case of $d=1$.

\begin{question}\label{qn:gaussians}
What is the minimum order of separation that is needed to learn the parameters 
of a mixture of $k$ spherical Gaussians up to accuracy $\delta$ using $\poly(d,k,1/\delta)$ samples? 
\end{question}

\subsection{Our Results} \label{sec:results}

By improving both the lower bounds and the upper bounds mentioned above, we characterize (up to constants) the minimum 
separation needed to learn the mixture from polynomially many samples. 
Our first result shows super-polynomial lower bounds when the separation is of the order $o(\sqrt{\log k})$. 
In what follows, $\dparam(\calG, \tilde\calG)$ represents the ``distance'' between the parameters of the two mixtures of Gaussians $\calG, \tilde\calG$ (see Definition~\ref{def:paramdist} for the precise definition). 
\begin{itheorem}[Lower Bounds]\label{thm:informal:lowerbound}
For any $\gamma(k) = o(\sqrt{\log k})$, there are two uniform mixtures of 
standard spherical Gaussians $\calG, \tilde\calG$ in $d=O(\log k)$ dimensions 
with 
means $\set{\mu_1,\dots,\mu_k}, \set{\tmu_1, \tmu_2, \dots, \tmu_k}$ respectively, that are well separated 
$$\forall i \ne j \in [k]: \norm{\mu_i - \mu_j}_2 \ge \gamma(k), ~\text{ and } \norm{\tmu_i - \tmu_j}_2 \ge \gamma(k),$$ 
and whose parameter distance is large $\dparam\left(\set{\mu_1,\dots,\mu_k},\set{\tmu_1,\dots,\tmu_k} \right) = \Omega(1)$, but have very small statistical distance $\norm{\calG - \tilde\calG}_{TV} \le k^{-\omega(1)}$.
\end{itheorem}
The above statement implies that we need at least $k^{\omega(1)}$ many samples to distinguish between $\calG, \tilde\calG$, and identify $\calG$. 
See Theorem~\ref{thm:lowerbound:random} for a formal statement of the result. 
%
In fact, these sample complexity lower bounds hold even when the means of the Gaussians are picked randomly in a ball of radius $\sqrt{d}$ in $d=o(\log k)$ dimensions.   
This rules out obtaining smoothed analysis guarantees for small dimensions (as opposed to \cite{BCMV, ABGRV14} which give polytime algorithms for smoothed mixtures of Gaussians in $k^{\Omega(1)}$ dimensions). 

Our next result shows that the separation of $\Omega(\sqrt{\log k})$ is tight -- this separation suffices to learn the parameters of the mixture with polynomial samples. We state the theorem for the special case of uniform mixtures of spherical Gaussians. (See Theorem~\ref{thm:betterupperbounds} for the formal statement.)
\begin{itheorem}[Tight Upper Bound in terms of $k$]\label{thm:informal:betterupperbounds}
  There exists a universal constant $c>0$, such that given samples from a uniform mixture of standard spherical Gaussians in $\R^d$ with 
  well-separated means, i.e.,
\begin{equation}
\forall i, j \in [k], i \ne j:~ \norm{\mu_i - \mu_j}_2 \ge  c \sqrt{\log k} 
\end{equation}
there is an algorithm that for any $\delta>0$ uses only $\poly(k,d,1/\delta)$ samples and with high probability finds $\set{\tmu_1, \tmu_2, \dots, \tmu_k}$ satisfying $\dparam\left(\set{\mu_1,\dots,\mu_k},\set{\tmu_1,\dots,\tmu_k} \right) \le \delta$. 
\end{itheorem}

While the above algorithm uses only $\poly(k,d,1/\delta)$ samples, it is computationally inefficient.

Our next result shows that in constant dimensions, one can obtain a \emph{computationally efficient} algorithm.
In fact, in such low dimensions a separation of order $\Omega(1)$ suffices.

\begin{itheorem}[Efficient algorithm in low dimensions]\label{thm:informal:lowd}
  There exists a universal constant $c>0$, such that given samples from a uniform mixture of standard spherical Gaussians in $\R^d$ 
  with well-separated means, i.e.,
\begin{equation}
\forall i, j \in [k], i \ne j:~ \norm{\mu_i - \mu_j}_2 \ge  c \sqrt{d} 
\end{equation}
there is an algorithm that for any $\delta>0$ uses only $\poly_d(k,1/\delta)$ time (and samples) and with high probability finds $\set{\tmu_1, \tmu_2, \dots, \tmu_k}$ satisfying $\dparam\left(\set{\mu_1,\dots,\mu_k},\set{\tmu_1,\dots,\tmu_k} \right) \le \delta$. 
\end{itheorem}

See Theorem~\ref{thm:lowdims} for a formal statement. An important feature of the above two algorithmic results is that the separation is independent of the accuracy $\delta$ that we desire in parameter estimation ($\delta$ can be arbitrarily small compared to $k$ and $d$).
\onote{I'm confused about the following sentence: how does the low-dimensional result help in getting the tight characterization? does the ``above two results'' refer to the
``above two algorithmic results'' referred to in the previous sentence?}
These results together almost give a {\em tight characterization} (up to constants) for the amount of separation needed to learn with $\poly(k,d,1/\delta)$ samples.

\paragraph{Iterative Algorithm.} 
The core technical portion of Theorem~\ref{thm:informal:betterupperbounds} and Theorem~\ref{thm:informal:lowd} is a new iterative algorithm, which is the main algorithmic contribution of the paper. 
This algorithm takes coarse estimates of the means, and iteratively refines them to get arbitrarily good accuracy $\delta$. 
 We now present an informal statement of the guarantees of the iterative algorithm. 
\begin{itheorem}[Iterative Algorithm Guarantees]\label{thm:informal:iterative}
 There exists a universal constant $c>0$, such that given samples from a uniform mixture of standard spherical Gaussians in $\R^d$ with 
 well-separated means, i.e.
\begin{equation}
\forall i, j \in [k], i \ne j:~ \norm{\mu_i - \mu_j}_2 \ge  c \min\set{\sqrt{\log k},\sqrt{d}} 
\end{equation}
and suppose we are given initializers $\tmu_1, \dots, \tmu_k$ for the means $\mu_1, \dots, \mu_k$ satisfying
$$\forall j \in [k], ~\frac{1}{\sigma_j}\norm{\mu_j - \tmu_j}_2 \le 1/\poly\big(\min\set{d,k}\big).
$$
There exists an iterative algorithm that for any $\delta>0$ that runs in $\poly(k,d,1/\delta)$ time (and samples), and after $T=O(\log\log(k/\delta))$ iterations, finds with high probability $\mu^{(T)}_1, \dots, \mu^{(T)}_k$ such that $\dparam(\set{\mu_1, \dots, \mu_k}, \set{\mu^{(T)}_1, \dots, \mu^{(T)}_k}) \le \delta$.
\end{itheorem}

The above theorem also holds when the weights and variances are unequal. See Theorem~\ref{thm:boxiterative} for a formal statement. Note that in the above result,  the desired accuracy $\delta$ can be arbitrarily small compared to $k$, and the separation required does not depend on $\delta$. To prove the polynomial identifiability results (Theorems~\ref{thm:informal:betterupperbounds} and~\ref{thm:informal:lowd}), we first find coarse estimates of the means that serve as initializers to this iterative algorithm, which then recovers the means up to arbitrarily fine accuracy independent of the separation. 

The algorithm works by solving a system of non-linear equations that is obtained by estimating simple statistics (e.g., means) of the distribution restricted to certain carefully chosen regions.
We prove that the system of non-linear equations satisfies a notion of ``diagonal dominance'' that allows us to leverage iterative algorithms like Newton's method
and achieve rapid (quadratic) convergence.

The techniques developed here can find such initializers using only $\poly(k,d)$ many samples, but use time that is exponential in $k$. This leads to the following natural open question:

\begin{openquestion}
Given a mixture of spherical Gaussians with equal weights and variances, and with separation 
$$\forall i \ne j \in [k], \norm{\mu_i - \mu_j}_2 \ge c \sqrt{\log k}  $$
for some sufficiently large absolute constant $c>0$, is there an algorithm that recovers the parameters up to $\delta$ accuracy in time $\poly(k,d,1/\delta)$?  
\end{openquestion}     
Our iterative algorithm shows that to resolve this open question affirmatively, it is enough to find initializers that are reasonably close to the true parameters. 
In fact, a simple amplification argument shows that initializers that are $c\sqrt{\log k}/8$ close to the true means will suffice for this approach.

Our iterative algorithm is reminiscent of some commonly used iterative heuristics,
such as Lloyd's Algorithm and especially Expectation Maximization (EM). While these iterative methods are the practitioners' method-of-choice for learning probabilistic models, they have been notoriously hard to analyze.
We believe that the techniques developed here may also be useful to prove 
guarantees for these heuristics. 

\subsection{Prior Work and Comparison of Results} \label{sec:priorwork}


Gaussian mixture models are among the most widely used probabilistic models in statistical inference~\cite{Pea94,Tei61,Tei67}. 
Algorithmic results fall into two broad classes --- separation-based results, and moment-based methods that do not assume explicit geometric separation.

\paragraph{Separation-based results.} The body of work that is most relevant to this paper assumes that there is some minimum separation between the means of the components in the mixture. The first polynomial time algorithmic guarantees for mixtures of Gaussians were given by Dasgupta~\cite{Das99}, who showed how to learn mixtures of spherical Gaussians when the separation is of the order of $d^{1/2}$. 
This was later improved by a series of works~\cite{AK01,VW04,AM05,KSV08, DS07, BV08,KK10, AS12} for both spherical Gaussians and general Gaussians. 
The algorithm of Vempala and Wang~\cite{VW04} gives the best known result, and 
uses PCA along with distance-based clustering to learn mixtures of spherical Gaussians with separation 
$$\norm{\mu_i - \mu_j}_2 \ge \parens{\min\set{k,d}^{1/4} \log^{1/4}(dk/\delta)+ \log^{1/2}(dk/\delta)} (\sigma_i+\sigma_j).$$
We note that all these clustering-based 
algorithms require a separation that either implicitly or explicitly depend on the estimation 
accuracy $\delta$.\footnote{Such a dependency on $\delta$ seems necessary for 
clustering-based algorithms that cluster every point accurately with high probability.} 
Finally, although not directly related to our work, we note that 
a similar separation condition was shown to suffice also for \emph{non-spherical} 
Gaussians~\cite{BV08}, where separation is measured based on the variance 
along the direction of the line joining the respective means 
(as opposed, e.g., to the sum of maximum variances $\norm{\Sigma_i}+\norm{\Sigma_j}$ which could be much larger). 
 
Iterative methods like Expectation Maximization (EM) and Lloyd's algorithm (sometimes called the $k$-means heuristic) are commonly used in practice to learn mixtures of spherical Gaussians but, as mentioned above, are notoriously hard to analyze.
Dasgupta and Schulman~\cite{DS07} proved that a variant of the EM algorithm learns mixtures of spherical Gaussians with separation of the order of $d^{1/4}\text{polylog}(dk)$. 
Kumar and Kannan~\cite{KK10} and subsequent work~\cite{AS12} showed
that the spectral clustering heuristic (i.e., PCA followed by Lloyd's algorithm) provably recovers the clusters in a rather wide family of distributions which includes non-spherical Gaussians; in the special case of spherical Gaussians, their analysis requires separation of order $\sqrt{k}$.

Very recently, the EM algorithm was shown to succeed for mixtures of $k=2$ spherical Gaussians with $\Omega(\sigma)$ separation~\cite{BWY, Hsuetal16,DTZ16} (we note that in this setting with $k=O(1)$, polynomial time guarantees are also known using other algorithms like the method-of-moments~\cite{KMV10}, as we will see in the next paragraph). 
SDP-based algorithms have also been studied in the context of learning mixtures of spherical Gaussians with a similar separation requirement of $\Omega(k \max_i \sigma_i)$~\cite{Mixonetal}. The question of how much separation between the components is necessary was also studied 
empirically by Srebro et al.~\cite{Srebro06}, who observed that iterative heuristics successfully learn the parameters under much smaller separation compared to known theoretical bounds.

\paragraph{Moment-based methods.} In a series of influential results, algorithms based on the method-of-moments were developed by \cite{KMV10,MV10,BS10} for efficiently learning mixtures of $k=O(1)$ Gaussians under arbitrarily small separation. To perform parameter estimation up to accuracy $\delta$, the running time of the algorithms is 
$\poly(d,1/\wmin,1/\delta)^{O(k^2)}$ (this holds for mixtures of general Gaussians). This exponential dependence on $k$ is necessary in general, due to statistical lower bound results~\cite{MV10}. The running time dependence on $\delta$ was improved in the case of $k=2$ Gaussians  in \cite{HP15}. 

Recent work~\cite{HK13,BCV,GVX14,BCMV,ABGRV14,GHK} use uniqueness of tensor decompositions (of order $3$ and above) to implement the method of moments and give polynomial time algorithms assuming the means are sufficiently high dimensional, and do not lie in certain degenerate configurations. Hsu and Kakade~\cite{HK13} gave a polynomial time algorithm based on tensor decompositions to learn a mixture of spherical Gaussians, when the means are linearly independent. This was extended by \cite{GVX14,BCMV,ABGRV14} to give smoothed analysis guarantees to learn ``most'' mixtures of spherical Gaussians when the means are in $d=k^{\Omega(1)}$ dimensions. These algorithms do not assume any strict geometric separation conditions and learn the parameters in $\poly(k,d,1/\delta)$ time (and samples), when these non-degeneracy assumptions hold. 
However, there are many settings where the Gaussian mixture consists of many clusters in a low dimensional space, or have their means lying in a low-dimensional subspace or manifold, where these tensor decomposition guarantees do not apply.
Besides, these algorithms based on tensor decompositions seem less robust to noise than clustering-based approaches and iterative algorithms, giving further impetus to the study of the latter algorithms as we do in this paper.


\paragraph{Lower Bounds.} Moitra and Valiant~\cite{MV10} showed that $\exp(k)$ samples are needed to learn the parameters of a mixture of $k$ Gaussians \cite{MV10}. 
In fact, the lower bound instance of \cite{MV10} is one dimensional, 
with separation of order $1/\sqrt{k}$.
Anderson et al.~\cite{ABGRV14} proved a lower bound on sample complexity that is reminiscent
of our Theorem~\ref{thm:informal:lowerbound}. 
Specifically, they obtain a super-polynomial lower bound assuming separation $O(\sigma/\poly\log(k))$ for $d=O(\log k/ \log\log k)$. 
This is in contrast to our lower bound which 
allows separation greater than $\sigma$, or $o(\sigma \sqrt{\log k})$ to be precise.

\paragraph{Other Related work.} While most of the previous work deals with estimating the parameters (e.g., means) of the Gaussians components in the given mixture $\calG$, a recent line of work~\cite{FOS06,AJOS14,DK14} focuses on the task of 
learning a mixture of Gaussians $\calG'$ (with possibly very different parameters) that is close in statistical distance i.e., $\norm{\calG-\calG'}_{TV} < \delta$ (this is called ``proper learning'', since the hypothesis that is output is also a mixture of $k$ Gaussians). 
When identifiability using polynomial samples is known for the family of distributions,
proper learning automatically implies parameter estimation. 
Algorithms for proper learning mixtures of spherical Gaussians~\cite{AJOS14,DK14} give polynomial sample complexity bounds when $d=1$ (note that the lower bounds of \cite{MV10} do not apply here) and have run time dependence that is exponential in $k$; the result of ~\cite{FOS06} has sample complexity that is polynomial in $d$ but exponential in $k$. Algorithms that take time and samples $\poly(k,1/\delta)^d$ are also known for ``improper learning'' and density estimation for mixtures of $k$ Gaussians (the hypothesis class that is output may not be a mixture of $k$ Gaussians)~\cite{CDRS,Bhaskara15}. We note that known algorithms have sample complexity that is either exponential in $d$ or $k$, even though proper learning and improper learning are easier tasks than parameter estimation. To the best of our knowledge, better bounds are not known under additional separation assumptions.  

 A related problem in the context of clustering graphs and detecting communities is the problem of learning a stochastic block model or planted partitioning model~\cite{McSherry}. Here, a sharp threshold phenomenon involving an analogous separation condition (between intra-cluster probability of edges and inter-cluster edge probability) is known under various settings~\cite{MNS13,M14,MNS14colt} (see the recent survey by Abbe~\cite{Abbesurvey} for details). In fact, the algorithm of Kumar and Kannan~\cite{KK10} give a general separation condition that specializes to separation between the means for mixtures of Gaussians, and separation between the intra-cluster and inter-cluster edge probabilities for the stochastic block model.

\subsection{Overview of Techniques} \label{sec:techniques}


\paragraph{Lower bound for $O(\sqrt{\log k})$ separation.}

The sample complexity lower bound proceeds by showing a more general statement: in any large enough collection of uniform mixtures, for all but a small fraction of the mixtures, there is at least one other mixture in the collection that is close in statistical distance (see Theorem~\ref{thm:ddim}). For our lower bounds, we will just produce a large collection of uniform mixtures of well-separated spherical Gaussians in $d= c \log k$ dimensions, whose pairwise parameter distances are reasonably large. In fact, we can even pick the means of these mixtures randomly in a ball of radius $\sqrt{d}$ in $d=c \log k$ dimensions; w.h.p. most of these mixtures will need at least $k^{\omega(1)}$ samples to identify. 

To show the above pigeonhole style statement about large collections of mixtures, we will associate with a uniform mixture having  means $\mu_1, \dots, \mu_k$, the following quantities  
that we call ``mean moments,'' and we will use them as a proxy for the actual moments of the distribution: 
$$(M_1, \dots, M_R)~\text{where }~ \forall 1\le r\le R:  M_r= \frac{1}{k} \sum_{j=1}^k \mu_j^{\otimes r}.$$
The mean moments just correspond to the usual moments of a mixture of delta functions centered at $\mu_1, \dots, \mu_k$. 
Closeness in the first $R=O(1/\eps)$ mean moments (measured in injective tensor norm) implies that the two corresponding distributions are $\eps$ close in statistical distance (see Lemma~\ref{lem:matchingmoments} and Lemma~\ref{lem:ell2toell1}). 
The key step in the proof uses a careful packing argument to show that for most mixtures in a large enough collection, there is a different mixture in the collection that approximately matches in the first $R$ mean moments (see Lemma~\ref{lem:packing}).

\paragraph{Iterative Algorithm.}
Our iterative algorithm will function in both settings of interest: the high-dimensional setting when we have $\Omega(\sqrt{\log k})$ separation, and the low-dimensional setting when $d<\log k$ and we have $\Omega(\sqrt{d})$ separation. 
For the purpose of this description, let us assume $\delta$ is arbitrarily small compared to 
$d$ and $k$.
In our proposed algorithm, we will consider distributions obtained by restricting the support to just certain regions around the initializers $z_1=\tilde\mu_1, \dots,z_k=\tilde\mu_k$ that are somewhat close to the means $\mu_1, \mu_2, \dots, \mu_k$ respectively. Roughly speaking, we first partition the space into a Voronoi partition given by $\set{z_j: j \in [k]}$, and then for each component $j \in [k]$ in $\calG$, let $S_j$ denote the region containing $z_j$ (see Definition~\ref{def:box:Sj} for details). For each $j \in [k]$ we consider only the samples in the set $S_j$ and let $u_j \in \R^d$ be the (sample) mean of these points in $S_j$, after subtracting $z_j$.

The regions are chosen in such a way that $S_j$ has a large fraction of the probability mass from the $j$th component, and the total probability mass from the other components is relatively small (it will be at most $1/\poly(k)$ with $\Omega(\sqrt{\log k})$ separation, and $O_d(1)$ with $\Omega(1)$ separation in constant dimensions). 
However, since $\delta$ can be arbitrarily small functions of $k,d$, there can still be a relatively large contribution from the other components. For instance, in the low-dimensional case with $O(1)$ separation, there can be $\Omega(1)$ mass from a single neighboring component! Hence, $u_j$ does not give a $\delta$-close estimate for $\mu_j$ (even up to scaling), unless the separation is at least of order $\sqrt{\log (1/\delta)}$ -- this is too large when $\delta =k^{-\omega(1)}$ with $\sqrt{\log k}$ separation, or $\delta=o_d(1)$ with $\Omega(1)$ separation in constant dimensions.

Instead we will use these statistics to set up a system of non-linear equations where the unknowns are the true parameters and solve for them using the Newton method. We will use the initializers $z_j=\mu^{(0)}_j$, to define the statistics that give our equations.  Hence the unknown parameters $\set{\mu_i: i \in [k]}$ satisfy the following equation for each $j \in [k]$:
\begin{equation}\label{eq:intro:iterative}
\sum_{i=1}^k w_i \int_{y \in S_j} (y - z_j) \cdot \sigma_j^{-d}\exp\Paren{-\frac{\pi \norm{y - \mu_i}_2^2}{\sigma_i^2}} \dy = u_j.
\end{equation}
Note that in the above equation, the only unknowns or variables are the true means $\set{\mu_i: i \in [k]}$.  
After scaling the equations, and a suitable change of variables $\vx_j=\mu_j/\sigma_j$ to make the system ``dimensionless'' we get a non-linear system of equations denoted by $F(\vx)=b$. 
 For the above system, $\vx^*_{i}=\mu_i /\sigma_i$ represents a solution to the system given by the parameters of $\calG$.  The Newton algorithm uses the iterative update
$$\vx^{(t+1)} = \vx^{(t)} + \parens{F'(\vx^{(t)})}^{-1} ( b- F(\vx^{(t)}) ) . $$

For the Newton method we need access to the estimates for $b$, and the derivative matrix $F'$  (the Jacobian) evaluated at $\vx^{(t)}$. The derivative of the $j$ equation w.r.t. $\vx_i$ corresponds to 
$$\grad_{\vx_i} F_j (\vx) =\frac{2\pi w_i}{w_j \sigma_j \sigma_i}  \int_{y \in S_j} (y - z_j) (y- \sigma_i \vx_i)^T g_{\sigma_i \vx_i, \sigma_i}(y) \dy ~,$$
where $g_{\sigma_i \vx_i, \sigma_i}(y)$ represents the p.d.f. at a point $y$ due to a spherical Gaussian with mean at $\sigma_i \vx_i$ and covariance $\sigma_i^2/ (2 \pi)$ in each direction. Unlike usual applications of the Newton method, we  
do not have closed form expressions for $F'$ (the Jacobian), due to our definition of the set $S_j$. However, we will instead be able to estimate the Jacobian at $\vx^{(t)}$ by calculating the above expression (RHS) by considering a Gaussian with mean $\sigma_i \vx^{(t)}_i $ and variance $\sigma_i^2/(2\pi)$. The Newton method can be shown to be robust to errors in $b, F, F'$ (see Theorem~\ref{thm:newton:errors}).

We want to learn each of the $k$ means up to good accuracy; hence we will measure the error and convergence in $\norm{\cdot}_\infty$ norm. 
This is important in low dimensions since measuring convergence in $\ell_2$ norm will introduce extra $\sqrt{k}$ factors, that are prohibitive for us since the means are separated only by $\Theta_d(1)$. 
The convergence of the Newton's method depends on upper bounding the operator norm of the inverse of the Jacobian $\norm{(F')^{-1}}$ and the second-derivative $\norm{F''}$, with the initializer being chosen $\delta$-close to the true parameters so that $\delta \norm{(F')^{-1}} \norm{F''} < 1/2$. 

The main technical effort for proving convergence is in showing that the inverse $(F')^{-1}$ evaluated at any point in the neighborhood around $\vx^*$ is well-conditioned. 
We will show the convergence of the Newton method by showing ``diagonal dominance'' properties of the $dk \times dk$ matrix $F'$. This uses the separation between the means of the components, and the properties of the region $S_j$ that we have defined. For $\Omega(\sqrt{\log k})$ separation, this uses standard facts about Gaussian concentration to argue that each of the $(k-1)$ off-diagonal blocks (in the $j$th row of $F'$) is at most $1/(2k)$ factor of the corresponding diagonal term.  
With $\Omega(1)$ separation in $d=O(1)$ dimensions, we can not hope to get such a uniform bound on all the off-diagonal blocks (a single off-diagonal block can itself be $\Omega_d(1)$ times the corresponding diagonal entry). We will instead use careful packing arguments to show that the required diagonal dominance condition (see Lemma~\ref{lem:smallleakage} for a statement). 
\anote{Is the above too vague?}

Hence, the initializers are used to both define the regions $S_j$, and as initialization  for the Newton method. Using this diagonal dominance in conjunction with initializers (Theorem~\ref{thm:upperbounds}) and the (robust) guarantees of the Newton method (Theorem~\ref{thm:newton:errors}) gives rapid convergence to the true parameters.

\section{Preliminaries} \label{sec:prelims}


Consider a mixture of $k$ spherical Gaussians $\calG$ in $\R^d$ that has parameters $\set{(w_j, \mu_j,\sigma_j): j \in [k]}$. The $j$th component has mean $\mu_j$ and covariance $\sigma_j^2/2\pi \cdot I_{d \times d}$. 
For $\mu \in \R^d, \sigma \in \R_+$, let $g_{\mu, \sigma}: \R^d \rightarrow \R_+$ represent the p.d.f. of a spherical Gaussian centered at $\mu$ and with covariance $\sigma^2/(2\pi) \cdot I_{d \times d}$. We will use $f$ to represent the p.d.f.\ of the mixture of Gaussians $\calG$, and $g_j$ to represent the p.d.f.\ of the $j$th Gaussian component.

\begin{definition}[Standard mixtures of Gaussians]\label{def:uniform}
A \emph{standard} mixture of $k$ Gaussians with means $\mu_1, \dots, \mu_k \in \R^d$ is a mixture of $k$ spherical Gaussians $\set{(\tfrac{1}{k}, \mu_j, 1): j \in [k]}$.
\end{definition}

A standard mixture is just a uniform mixture of spherical Gaussians with all covariances $\sigma^2=1/(2 \pi)$. 
Before we proceed, we define the following notion of parameter ``distance'' between mixtures of Gaussians:

\anote{This definition may not be the right one -- maybe account for $d$ factor?. Also it is better if it is $\sigma_j^2 - (\sigma'_j)^2$ as opposed to $\sigma_j - \sigma'_j$. }
\begin{definition}[Parameter distance]\label{def:paramdist}
Given two mixtures of Gaussians in $\R^d$, $\calG=\set{(w_j, \mu_j,\sigma_j): j \in [k]}$ and $\calG'=\set{(w'_j, \mu'_j,\sigma_j'): j \in [k]}$, define
\[ 
\dparam\left(\calG, \calG' \right)= 
\min_{\pi \in \text{Perm}_k} \sum_{j=1}^k \frac{\abs{w_j - w_{\pi(j)}}}{\min\set{w_j, w_{\pi(j)}}} + 
\sum_{j=1}^k \frac{\norm{\mu_j - \mu'_{\pi(j)}}_2}{\min\set{\sigma_j, \sigma'_{\pi(j)}}}+
\sum_{j=1}^k \frac{\abs{\sigma_j - {\sigma'_{\pi(j)}}}}{\min\set{\sigma_j, \sigma'_{\pi(j)}}} \; .
\]
For \emph{standard} mixtures, the definition simplifies to 
\[ 
\dparam\left((\mu_1,\dots,\mu_k), (\mu'_1,\dots,\mu'_k) \right)= 
\min_{\pi \in \text{Perm}_k} \sum_{j=1}^k \norm{\mu_j - \mu'_{\pi(j)}}_2 \; .
\]
\end{definition}
Note that this definition is invariant to scaling the variances (for convenience). We note that parameter distance is not a metric, but it is just a convenient way of measure closeness of parameters between two distributions. 

The distance between two individual Gaussian components can also be measured in terms of the total variation distance between the components~\cite{MV10}. 
For instance, in the case of standard spherical Gaussians, a parameter distance of $c\sqrt{\log k}$ corresponds to a total variation distance of $k^{-O(c^2)}$.

\begin{definition}[$\rho$-bounded mixtures]
For $\rho \ge 1$, 
a mixture of spherical Gaussians $\calG=\set{(w_j, \mu_j,\sigma_j)}_{j=1}^k$ in $\R^d$ is called \emph{$\rho$-bounded} if 
for each $j \in [k]$, $\norm{\mu_j}_2 \le \rho$ and $\tfrac{1}{\rho} \le \sigma_j \le \rho$.
In particular, a \emph{standard} mixture is $\rho$-bounded if for each $j \in [k]$, $\norm{\mu_j}_2 \le \rho$. 
\end{definition}
Also, for a given mixture of $k$ spherical gaussians $\calG=\set{(w_j, \mu_j,\sigma_j): j \in [k]}$, we will denote $\wmin= \min_{j \in [k]} w_j$,  $\smax=\max_{j \in [k]} \sigma_j$ and $\smin=\min_{j \in [k]} \sigma_j$.

In the above notation the bound $\rho$ can be thought of as a sufficiently large polynomial in $k$, since we are aiming for bounds that are polynomial in $k$. Since we can always scale the points by an arbitrary factor without affecting the performance of the algorithm, we can think of $\rho$ as the (multiplicative) range of values taken by the parameters $\set{\mu_i, \sigma_i: i \in [k]}$. Since we want separation bounds independent of $k$, we will denote individual aspect ratios for variances and weights given by $\rhos=\max_{i \in [k]} \sigma_i / \min_{i \in [k]} \sigma_i$, and $\rhow = \max_{i \in [k] } w_i / \min_{i \in [k]} w_i$.


Finally, we list some of the conventions used in this paper. We will denote by $N(0,\sigma^2)$ a normal random variable with mean $0$ and variance $\sigma^2$. 
For $x \in \R$ generated according to $N(0,\sigma^2)$, let $\phit(t)$ denote the probability that $x >t$, and let $\iphit{y}$ 
denote the quantile $t$ at which $\phit(t)\le y$. 
For any function $f:\R^d \rightarrow \R$, $f'$ will denote the first derivative (or gradient) of the function, and $f''$ will denote the second derivative (or Hessian).
We define $\norm{f}_{1,S}=\int_S \abs{f(x)} \dx$ to be the $L_1$ norm of $f$ restricted to the set $S$. Typically, we will use indices $i,j$ to represent one of the $k$ components of the mixture, and we will use $r$ (and $s$) for coordinates. For a vector $x \in \R^d$, we will use $x(r)$ denote the $r$th coordinate. 
Finally, we will use {\em w.h.p.} in statements about the success of algorithms to represent probability at least $1-\gamma$ where $\gamma = (d+k)^{-\Omega(1)}$.

\paragraph{Norms.} For any $p\ge 1$, given a matrix $M \in \R^{d \times d}$, we define 
the matrix norm
$$\norm{M}_{p \rightarrow p}= \max_{x \in \R^d: \norm{x}_p=1} \norm{Mx}_p.$$
\onote{8/25: removing:
$$\norm{M}_p =  \Paren{\sum_{i,j=1}^d \abs{M(i,j)}^p}^{1/p}.$$
}

\subsection{Notation and Preliminaries about Newton's method}\label{sec:newtonprelims}
%

Consider a system of $m$ non-linear equations in variables $u_1, u_2, \dots, u_m$:
$$\forall j \in [m], f_j(u_1, \dots, u_m) = b_j. $$

Let $F'=J(u) \in \R^{m \times m}$ be the Jacobian of the system given by the non-linear functional $f:\R^{m} \rightarrow \R^{m}$, where the $(j,i)^{th}$ entry of $J$ is the partial derivative $\frac{\partial f_j (u)}{\partial u_i}$ is evaluated at $u$. Newton's method starts with an initial point $u^{(0)}$, and updates the solution using the iteration:
$$u^{(t+1)} = u^{(t)} + \left(J(u^{(t)})\right)^{-1} \left( b_j- f(u^{(t)})\right) . $$

Standard results shows quadratic convergence of the Newton method for general normed spaces~\cite{AtkinsonHall}. We restrict our attention in the restricted setting where both the range and domain of $f$ is $\R^m$, equipped with an appropriate norm $\norm{\cdot}$ to measure convergence.
\begin{theorem}[Theorem 5.4.1 in \cite{AtkinsonHall}]\label{thm:newton}
Assume $u^* \in \R^m$ is a solution to the equation $f(y)=b$ where $f: \R^m \rightarrow \R^m$ and the inverse Jacobian $J^{-1}$ exists in a neighborhood $N=\set{u: \norm{u- u^*} \le \norm{u^{(0)} - u^*}}$, %
and $F':\R^m \rightarrow \R^{m \times m}$ is locally $L$-Lipschitz continuous in the neighborhood $N$ i.e.,
$ \forall u, v \in N, ~~ \norm{F'(u) - F'(v)} \le L \norm{u - v}$.
Then we have $\norm{u^{(t+1)} - u^* } \le L \cdot \norm{J(u^{(t)})^{-1}} \cdot \norm{u^{(t)} - u^*}^2$.
\end{theorem}

In particular, for Newton's method to work, $\norm{u^{0} - u^*} \le (L \max_{u \in \calN} \norm{J(u)^{-1}})^{-1}$ will guarantee convergence.
A statement of the robust convergence of Newton's method in the presence of estimates is given in Theorem~\ref{thm:newton:errors} and Corollary~\ref{corr:newton:errors}.


We want to learn each of the $k$ sets of parameters up to good accuracy; hence we will measure the error in $\ell_\infty$ norm. 
To upper bound 
$\norm{J^{-1}}_{\infty \to \infty}$, we will use \emph{diagonal dominance} properties of the matrix $J$. Note that $\norm{A}_{\infty \to \infty}$ is just the maximum $\ell_1$ norm of the rows of $A$. The following lemma bound $\norm{A^{-1}}_{\infty \to \infty}$ for a diagonally dominant matrix $A$.

\begin{lemma}[\cite{Varah}]\label{lem:diagdominant}
Consider any square matrix $A$ of size $n \times n$ satisfying
\[
\forall i \in [n]~~ a_{ii} - \sum_{j \ne i} |a_{ij}| \ge \alpha. 
\]
Then, $\norm{A^{-1}}_{\infty \to \infty} \le 1/\alpha.$\onote{8/7: I think the proof is very easy, so perhaps we can include it}
\end{lemma}

%



\section{Lower Bounds with \texorpdfstring{$O(\sqrt{\log k})$}{O(sqrt(log k))} Separation}

Here we show a sample complexity lower bound for learning standard mixtures of $k$ spherical Gaussians even when the separation is of the order of $\sqrt{\log k}$. In fact, this lower bound will also hold for a \emph{random} mixture of Gaussians in $d \le c \cdot \log k$ dimensions (for sufficiently small constant $c$) with high probability.\footnote{In particular, this rules out polynomial-time smoothed analysis guarantees of the kind shown for $d=k^{\Omega(1)}$ in \cite{BCMV,ABGRV14}.} 

\begin{theorem} \label{thm:lowerbound:random}
For any large enough $C$ there exist $c, c_2>0$, such that the following holds for all $k \ge C^8$. 
Let $\calD$ be the distribution over 
standard mixtures of $k$ spherical Gaussians 
obtained by picking each of the $k$ means independently and uniformly from a ball of radius $\sqrt{d}$ around the origin in $d=c \log k$ dimensions.
Let $\set{\mu_1, \mu_2, \dots, \mu_k}$ be a mixture chosen according to $\calD$.
Then with probability at least $1-2/k$ there exists another standard mixture of $k$ spherical Gaussians with means $\set{\tmu_1, \tmu_2, \dots, \tmu_k}$ 
such that both mixtures are $\sqrt{d}$ bounded and well separated, i.e.,
\[
\forall i,j \in [k], i \ne j: \norm{\mu_i - \mu_j}\ge c_2\sqrt{\log k} \text{~~and~~} \norm{\tmu_i - \tmu_j}\ge c_2\sqrt{\log k},
\]
and their p.d.f.s satisfy 
\begin{equation}
\norm{f-\tsf}_1 \le k^{-C}
\end{equation}
even though their parameter distance is at least $c_2 \sqrt{\log k}$. 
Moreover,  we can take $c=1/(4 \log C)$ and $c_2=C^{-24}$.
\end{theorem}

\noindent {\em Remark.} 
In Theorem~\ref{thm:lowerbound:random}, there is a trade-off between getting a smaller statistical distance $\eps=k^{-C}$, and a larger separation between the means in the Gaussian mixture. 
When $C=\omega(1)$, 
with $c_1,c=o(1)$ we see that $\norm{f-\tsf}_1 \le k^{-\omega(1)}$ when the separation is $o(\sqrt{\log k}) \sigma$. 
On the other hand, we can also set $C=k^{\eps'}$ (for some small constant $\eps'>0$) to get lower bounds for mixtures of spherical Gaussians in $d=1$ dimension with $\norm{f - \tsf}_1 = \exp(-k^{\Omega(1)})$ and separation $1/k^{O(1)}$ between the means. 
We note that in this setting of parameters, our lower bound is nearly
identical to that in~\cite{MV10}. Namely, they achieve
statistical distance $\norm{f - \tsf}_1 = \exp(-\Omega(k))$ (which is better than ours) 
with a similar separation of $k^{-O(1)}$ in one dimension. 
One possible advantage of our bound is that it holds with a random choice of means,
unlike their careful choice of means.

\subsection{Proof of Theorem~\ref{thm:lowerbound:random}}

The key to the proof of Theorem~\ref{thm:lowerbound:random} is the following pigeonhole statement, 
which can be viewed as a bound on the covering number (or equivalently, the 
metric entropy) of the set of Gaussian mixtures. 

\anote{9/29: $\log 3d$ to $\log 5d$}
\begin{theorem} \label{thm:ddim}
Suppose we are given a collection $\calF$ of standard mixtures of spherical Gaussians in $d$ dimensions that are $\rho=\sqrt{d}$ bounded, i.e., $\norm{\mu_j} \le \sqrt{d}$ for all $j\in[k]$. There are universal constants $c_0, c_1 \ge 1$, such that for any $\eta>0, \eps \le \exp(-c_1 d)$, if
\begin{equation}
  \abs{\calF} > \frac{1}{\eta} \exp\left( c_0 \Big(\frac{\log (1/\eps)}{d} \Big)^{d} \cdot \log(1/\epsilon) \log (5d) \right), \label{eq:family:condition}
\end{equation}
then for at least $(1-\eta)$ fraction of the mixtures $\set{\mu_1, \mu_2, \dots, \mu_k}$ from $\calF$, there is another mixture $\set{\tmu_1, \tmu_2, \dots, \tmu_k}$ from $\calF$ with p.d.f.\ $\tsf$ such that $\norm{f - \tsf}_1 \le \epsilon$. Moreover, we can take $c_0=8\pi e$ and $c_1=36$.
\end{theorem}
\begin{remark}
Notice that $k$ plays no role in the statement above. In fact, the proof
also holds for mixtures with arbitrary number of components and arbitrary weights. 
\end{remark}

\begin{claim} \label{clm:lbhelper}
Let $x_1,\ldots, x_N$ be chosen independently and uniformly from the ball of radius $r$ in $\R^d$. Then
for any $0<\gamma<1$, with probability at least $1-N^2 \gamma^d$, we have that for all
$i \neq j$, $\|x_i - x_j\| \ge \gamma r$. 
\end{claim}
\begin{proof}
For any fixed $i \neq j$, the probability that $\|x_i - x_j\| \ge \gamma r$ is at most $\gamma^d$,
because the volume of a ball of radius $\gamma r$ is $\gamma^d$ times that of a ball of radius $r$.
The claim now follows by a union bound. 
\end{proof}

\begin{proof}[Proof of Theorem~\ref{thm:lowerbound:random}]
Set $\gamma := 2^{-6/c}$, and consider
the following probabilistic procedure. 
We first let $\calX$ be a set of $(1/\gamma)^{d/3}$ points chosen independently and uniformly from the ball of radius $\sqrt{d}$. 
We then output a mixture chosen uniformly from the collection $\calF$,
defined as the collection of all standard mixtures of spherical 
Gaussians obtained by selecting $k$ distinct means from $\calX$. 
Observe that the output of this procedure is distributed according to $\calD$.
Our goal is therefore to prove that with probability at least $1-2/k$, 
the output of the procedure satisfies the property in the theorem. 

First, by Claim~\ref{clm:lbhelper}, with probability at least $1-\gamma^{d/3} \ge 1-1/k$,
any two points in  $\calX$ are at distance at least $\gamma \sqrt{d}$. 
It follows that in this case, the means in any mixture in $\calF$ are at least $\gamma \sqrt{d}$ apart,
and also that any two distinct mixtures in $\calF$ have a parameter distance 
of at least $\gamma \sqrt{d}$ since they must differ in at least 
one of the means. Note that $\gamma=C^{-24}$ for our choice of $c, \gamma$. 

\anote{8/29: changed $\log(3d)$ to $\log(5d)$ in the following expressions. }
To complete the proof, we notice that by our choice of parameters, and denoting $\eps = k^{-C}$,
\begin{align*}
 \abs{\calF} = 
\binom{\abs{\calX}}{k} \ge 
\parens[\Big]{\frac{1}{\gamma}}^{dk/3} \cdot k^{-k} 
= k^k 
\ge 
k \cdot \exp\left( c_0 \Big(\frac{\log (1/\eps)}{d} \Big)^{d} \cdot \log(1/\epsilon) \log (5d) \right)
\; .
\end{align*}
The last inequality follows since for our choice of $\eps=k^{-C}$, 
$c=\frac{1}{4\log C}$ and $C$ is large enough with $C \ge c_0$, 
so that
$$\left(\frac{\log (1/\eps)}{d}\right)^d= k^{c \log(C/c)} < \sqrt{k}, ~~\text{ and }~ c_0 \log(1/\eps) \log(5d) \le c_0 C \log k \log( 5c \log k) < \sqrt{k}.$$
Hence applying Theorem~\ref{thm:ddim} to $\calF$, for at least $1-1/k$ fraction of the mixtures in $\calF$, there is another mixture in $\calF$ 
that is $\eps$ close in total variation distance. We conclude that with probability at least $1-2/k$, a random mixture in $\calF$ satisfies all the required properties, as desired. 
\end{proof}

\subsection{Proof of Theorem~\ref{thm:ddim}}

It will be convenient to represent the p.d.f.\ $f(x)$ of the standard mixture of spherical Gaussians with means $\mu_1, \mu_2, \dots, \mu_k$ as a convolution 
of a standard mean zero Gaussian with a sum of delta functions centered at $\mu_1, \mu_2, \dots, \mu_k$,
$$f(x)=\left(\frac{1}{k} \sum_{j=1}^k \delta(x-\mu_j) \right) * e^{-\pi \norm{x}_2^2  }.$$

Instead of considering the moments of the mixture of Gaussians, we will consider moments of just the corresponding mixture of 
delta functions at the means. We will call them ``mean moments,'' and we will use them as a proxy for the actual moments of the distribution.
$$(M_1, \dots, M_R)~\text{where }~ \forall 1\le r\le R:  M_r= \frac{1}{k} \sum_{j=1}^k \mu_j^{\otimes r}.$$

To prove Theorem~\ref{thm:ddim} we will use three main steps. 
Lemma~\ref{lem:packing} will show using the pigeonhole principle that for any large enough collection of Gaussian mixtures $\calF$, most Gaussians mixtures in the family have other mixtures which approximately match in their first $R=O(\log (1/\eps))$ mean moments. 
This closeness in
moments will be measured using the \emph{symmetric injective tensor norm} defined 
for an order-$\ell$ tensor $T \in \R^{d^\ell}$ as
\[
\symnorm{T}=\max_{\substack{y \in \R^d \\ \norm{y}=1}} \abs{\iprod{T, y^{\otimes \ell}}}.
\]
Next, Lemma~\ref{lem:matchingmoments} shows that the two distributions that are close in the first $R$ mean moments are also close in the $L_2$ distance.
This translates to small statistical distance between the two distributions using Lemma~\ref{lem:ell2toell1}.

We will use the following standard packing claim. 

\begin{claim}\label{clm:packing}
Let $\|\cdot\|$ be an arbitrary norm on $\R^D$.
If $x_1,\ldots,x_N \in \R^D$ are such that $\|x_i\| \le \Delta$ for all $i$, and for all $i \neq j$,
$\|x_i - x_j\| > \delta$, then $N \le (1+2\Delta/\delta)^D$. In particular, if $x_1,\ldots,x_N \in \R^D$
are such that $\|x_i\| \le \Delta$ for all $i$, then for all but $(1+2\Delta/\delta)^D$ of the indices $i \in [N]$, there exists a $j \neq i$ such that $\|x_i - x_j\| \le \delta$.
\end{claim}
\begin{proof}
Let $K$ be the unit ball of the norm $\|\cdot\|$. Then by assumption, the sets $x_i + \delta K /2$ for $i \in [N]$ are disjoint. But since they are all contained in $(\Delta+\delta/2)K$, 
\[
N \le \frac{\mathrm{Vol}((\Delta+\delta/2)K)}{\mathrm{Vol}(\delta K/2)} = (1+2\Delta/\delta)^D \; .
\]
The ``in particular'' part follows by taking a maximal set of $\delta$-separated points. 
\end{proof}

\begin{lemma}\label{lem:packing}
Suppose we are given a set $\calF$ of standard mixtures of spherical Gaussians in $d$ dimensions with means of length at most $\sqrt{d}$. Then for any integer $R \ge d$,
if $\abs{\calF} > \frac{1}{\eta} \cdot \exp\big( (2eR/d)^d R \log (5d) \big)$, it holds that for at least $(1-\eta)$ fraction of the mixtures $\set{\mu_1, \mu_2, \dots, \mu_k}$  in $\calF$, there is another mixture 
$\set{\tilde\mu_1, \tilde\mu_2, \dots, \tilde\mu_k}$ in $\calF$ satisfying that for $r=1,\ldots,R$,
\anote{9/29: changed to $(d+1)^{-R/4}$. was: $d^{-R/4}$}
\begin{equation}
\symnorm[\Big]{\frac{1}{k} \sum_{j=1}^k \mu_j^{\otimes r}- \frac{1}{k} \sum_{j=1}^k (\tmu_j)^{\otimes r}} \le (d+1)^{-R/4}.
\end{equation}
\end{lemma}


\begin{proof}
With any choice of means $\mu_1,\mu_2, \dots, \mu_k \in \R^d$ we can associate a vector of moments 
$\psi(\mu_1,\mu_2, \dots, \mu_k) = (M_1, \dots, M_R)$ where for $r=1,\ldots,R$, 
\[
M_r= \frac{1}{k} \sum_{j=1}^k \mu_j^{\otimes r}\; .
\]
Notice that the image of $\psi$ lies in a direct sum of symmetric subspaces 
whose dimension is $D=\binom{d+R}{R}$ (i.e., the number of ways of distributing $R$ identical balls into $(d+1)$ different bins).
Since $R\ge d$, 
\begin{equation} 
D = \binom{d+R}{R}\le \left(\frac{e(d+R)}{d}\right)^d \le \left(\frac{2eR}{d}\right)^d.\label{eq:deff}
\end{equation}
We define a norm on these vectors $\psi(\mu_1, \dots, \mu_k)$ in terms of the maximum injective norm of the mean moments,
\begin{align}
\symnormmax{\psi(\mu_1, \dots, \mu_k)} &= \max_{r \in [R]} \symnorm[\Big]{\frac{1}{k} \sum_{j=1}^{k} \mu_j^{\otimes r}} \; .
\end{align}
Since each of our means has length at most $\sqrt{d}$, we have that
\begin{align}
\symnormmax{\psi(\mu_1, \dots, \mu_k)} &\le \max_{r \in [R]} \symnorm[\Big]{\frac{1}{k} \sum_{j=1}^{k} \mu_j^{\otimes r}}
 \le \max_{r\in [R]} \max_{\norm{\mu} \le \sqrt{d} } \symnorm{\mu^{\otimes r}}
 \le d^{R/2}. \label{eq:symnormmax}
\end{align}
Using Claim~\ref{clm:packing}, if $\abs{\calF}> N/ \eta$ where 
\begin{align*}
N &= \Paren{1+ 2d^{R/2}/(d+1)^{-R/4}}^{D} \le \Paren{1+2(d+1)^{3R/4}}^{D} \le \exp\left(\frac{3}{4}(2e R/d)^d \cdot R \log (5d) \right),
\end{align*}
\anote{9/29: changed to $(d+1)^{-R/4}$. was: $d^{-R/4}$}
we have that for at least $1-\eta$ fraction of the Gaussian mixtures in $\calF$, there is another mixture from $\calF$ which is $(d+1)^{-R/4}$ close, as required.\anote{9/29: changed d to d+1}
\end{proof}


Next we show that the closeness in moments implies closeness in the $L_2$ distance. 
This follows from fairly standard Fourier analytic techniques. We will first show that if the mean moments 
are close, then the low-order Fourier coefficients are close. This will then imply that the Fourier spectrum of the corresponding Gaussian mixtures $f$ and $\tsf$ are close. 

\begin{lemma}\label{lem:matchingmoments}
Suppose $f(x),\tsf(x)$ are the p.d.f.\ of $\calG, \tcalG$ which are both standard mixtures of $k$ Gaussians in $d$ dimensions with means $\set{\mu_j : j\in [k]}$ and $\set{\tmu_j: j \in [k]}$ respectively that are both $\rho = \sqrt{d}$ bounded.
There exist universal constants $c_1, \cmoments\ge 1$, such that for every $\eps \le \exp(-c_1 d)$ if the following holds for $R = \cmoments \log(1/\epsilon)$:
\begin{equation}\label{eq:momentserror:ddim}
\forall 1\le r\le R, ~ \frac{1}{k} \symnorm[\Big]{\sum_{j=1}^k \mu_j^{\otimes r}- \sum_{j=1}^k (\tmu_j)^{\otimes r}} \le \eps_r := \eps \parens[\Big]{\frac{r}{ 8 \pi e \sqrt{\log(1/\eps)}}}^r,
\end{equation}
then $\norm{f-\tsf}_2 \le \epsilon$.
\end{lemma}



\begin{proof}
We first show that if the moments are very close, the Fourier transform of the two distributions is very close. This translates to a bound on the $L_2$ distance by Parseval's identity.
%

Let $g,h: \R^d \rightarrow \R$ be defined by
\begin{equation}
h(x)=\frac{1}{k} \parens[\Big]{\sum_{j=1}^k \delta_{\mu_j}(x)-\sum_{j=1}^k \delta_{\tmu_j}(x)}, \quad \;  g(x)=f(x)-\tsf(x)=h(x)*e^{-\pi \norm{x}^2 }.
\end{equation}

Since the Fourier transform of a convolution is the product of the Fourier transforms, we have 
\begin{align}
\forall \zeta \in \R^d,~ \ft{f}(\zeta)& =
\frac{1}{k} \parens[\Big]{\sum_{j=1}^k e^{-2 \pi \im \iprod{\zeta,\mu_j}}} \cdot e^{-\pi \norm{\zeta}^2}&\\
\ft{h}(\zeta)&= \frac{1}{k} \parens[\Big]{\sum_{j=1}^k e^{-2\pi \im \iprod{\mu_j,\zeta}} - \sum_{j=1}^k  e^{-2\pi \im \iprod{\tmu_j, \zeta}}} , ~\qquad~ \ft{g}(\zeta) = \ft{h}(\zeta) \cdot e^{-\pi \norm{\zeta}^2}.
\end{align}

We will now show that $\int_{\R^d} \abs{\ft{g}(\zeta)}^2 d \zeta \le \epsilon^2$. 
We first note that the higher order Fourier coefficients of $g$ do not contribute much to the Fourier mass of $g$ because of the Gaussian tails. 
Let $\tau^2 = 4 \log(1/\eps)$. Since 
$\tau^2 \ge (d+ 2\sqrt{d \log (16/\eps^2))} + 2 \log (16/\eps^2))/(2\pi)$, 
using Lemma~\ref{lem:lengthconc}
\begin{equation}
  \int_{\norm{\zeta} > \tau} \abs{\ft{g}(\zeta)}^2 d \zeta = \int_{\norm{\zeta} > \tau} \abs{\ft{h}(\zeta)}^2 e^{-2 \pi \norm{\zeta}^2 } d \zeta
\le \int_{\norm{\zeta} > \tau }4 e^{-2 \pi \norm{\zeta}^2 } d \zeta
\le \frac{\eps^2}{4} \; . \label{eq:largerthantau}
\end{equation}
Now we upper bound $\abs{\ft{h}(\zeta)}$ for $\norm{\zeta} <\tau$. 
\begin{align*}
\abs{\ft{h}(\zeta)}&= 
    \frac{1}{k} \Abs{\sum_{j=1}^k \parens{e^{-2\pi \im \iprod{\mu_j,\zeta}} - e^{-2\pi \im \iprod{\tmu_j,\zeta}}}} =  
    \frac{1}{k} \Abs{\sum_{j=1}^k \sum_{r=1}^\infty \frac{(-2\pi \im)^r}{r!}\left(\iprod{\mu_j,\zeta}^r-\iprod{\tmu_j,\zeta}^r \right)} \\
& \le \sum_{r=1}^{\infty} \frac{(2\pi)^r \norm{\zeta}^r }{r!}\cdot \frac{1}{k} \cdot \symnorm[\Big]{\sum_{j=1}^k \mu_j^{\otimes r}- \sum_{j=1}^k (\tmu_j)^{\otimes r}}.
\end{align*}
We claim that the injective norm above is at most $k \eps_r$ for all $r \ge 1$. 
For $r \le R$, this follows immediately from the assumption in~\eqref{eq:momentserror:ddim}. For $r \ge R$, we use the fact that the means are $\rho=\sqrt{d}$ bounded,
\begin{align*}
\frac{1}{k} \symnorm[\Big]{\sum_{j=1}^k \mu_j^{\otimes r}- \sum_{j=1}^k(\tmu_j)^{\otimes r}}
& \le 2 \max_{\mu \in \set{\mu_j, \tmu_j}_{j=1}^k} \norm{\mu}_2^r \le 2d^{r/2} \le 2^{-r/2+1} (2d)^{r/2} \\
&\le \eps \cdot \parens[\Big]{\frac{R^2}{(8 \pi e)^2 \log(1/\eps)}}^{r/2} \le \eps \parens[\Big]{\frac{r}{8 \pi e \sqrt{\log(1/\eps)}}}^{r} = \eps_r,
\end{align*}
where the last line follows since $R\ge 16 \pi e \log(1/\eps)$ and $2d \le \log(1/\eps)$.
Hence, for $\norm{\zeta} \le \tau$, we have 
\begin{align*}
\abs{\ft{h}(\zeta)} & \le \sum_{r=1}^{\infty} \frac{(2\pi)^r \norm{\zeta}^r }{r!}\cdot \eps_r \le \sum_r \frac{(2\pi \norm{\zeta})^r}{ \sqrt{2 \pi r} (r/e)^r}\cdot \eps \parens[\Big]{\frac{r}{ 8 \pi e \sqrt{\log(1/\eps)}}}^r \\
&\le \eps \sum_{r \ge 1} \frac{1}{\sqrt{2 \pi r}} \parens[\Big]{\frac{2 \pi e \norm{\zeta}}{r} \cdot \frac{r}{ 8 \pi e \sqrt{\log(1/\eps)}}}^r \le \eps \sum_{r \ge 1} \frac{1}{\sqrt{2 \pi r}} \left(\frac{\norm{\zeta}}{2 \tau} \right)^r \le \frac{\eps}{2},
\end{align*}
since $ \norm{\zeta} \le \tau$. 
%
%
%
Finally, using this bound along with \eqref{eq:largerthantau} we have
\begin{align*}
\int_{\R^d} \Abs{\ft{g}(\zeta)}^2 d \zeta 
& \le \frac{\eps^2}{4} \int_{\norm{\zeta} \le \tau}e^{-2\pi \norm{\zeta}^2 }d \zeta + \frac{\eps^2}{4}
 \le \frac{\eps^2}{4}  + \frac{\eps^2}{4} \le \eps^2. 
\end{align*}
Hence, by Parseval's identity, the lemma follows.

\end{proof}

\newcommand{\mumax}{\left(\max\norm{\mu}\right)}

The following lemma shows how to go from $L_2$ distance to $L_1$ distance using the Cauchy-Schwartz lemma. Here we use the fact that all the means have length at most $\sqrt{d}$. Hence, we can focus on a ball of radius at most $O(\sqrt{\log(1/\eps)})$, since both $f,\tsf$ have negligible mass outside this ball.

\begin{lemma}\label{lem:ell2toell1}
In the notation above, suppose the p.d.f.s $f,\tsf$ of two standard mixtures of Gaussians in $d$ dimensions that are $\sqrt{d}$-bounded (means having length $\le \sqrt{d}$) satisfy $\norm{f-\tsf}_2 \le \eps$, for some $\eps \le  \exp(-6 d)$.
Then,
$\norm{f - \tsf}_1 \le 2\sqrt{\eps}$.
\end{lemma}

\begin{proof}
The means are all length at most $\sqrt{d}$, and $\eps< 2^{-d}$.
Let us define as before
$$\tau^2 = \frac{1}{2\pi}\parens[\big]{d+2 \sqrt{d \log(2/\eps)}+ 2 \log(2/\eps)}, ~\text{ and }~\gamma = \mumax+\tau \le 2 (\sqrt{d}+ \sqrt{\log (2/\eps)}) \le 4 \sqrt{\log(2/\eps)}.$$ 
Let $g=f-\tsf$ and let $S=\set{x: \norm{x} \le \gamma}$.
Using Gaussian concentration in Lemma~\ref{lem:lengthconc}, we see that both $f, \tsf$ have negligible mass of at most $\eps/2$ each outside $S$.
Hence, using Cauchy-Schwartz inequality and $\norm{f - \tsf}_2 \le \eps$,
 $$ \int_{\R^d} \abs{g(x)} \dx \le  \int_S \abs{g(x)}\dx + \eps \le \sqrt{\int_S g(x)^2 \dx} \cdot \sqrt{\int_S \dx} + \eps
 \le \epsilon \sqrt{\Vol(S)} + \eps.$$
Since $S$ is a Euclidean ball of radius $\gamma$, by Stirling's approximation (Fact~\ref{fact:stirling}), the volume is
\begin{align*}
   \Vol(S) = \frac{\pi^{d/2} \gamma^d}{(d/2)!}
	\le \left( \frac{2 \pi e  \gamma^2}{d} \right)^{d/2}  \le \left( \frac{32 \pi e \log(2/\eps)}{d} \right)^{d/2} \le 2^{\log(1/\eps)} = \frac{1}{\eps},
\end{align*}
where the third inequality follows by raising both sides to the power $2/d$ and using the fact that
for $\alpha \ge 6$, $32 \pi e (\alpha+1) \le 2^{2 \alpha}$. 
This concludes the proof.
%

\end{proof}


\anote{8/29: Made change to Lemma 3.6 to make it $(d+1)^{-R/4}$ instead of $d^{-R/4}$, and propagated changes. }
\begin{proof}[Proof of Theorem~\ref{thm:ddim}]
The proof follows by a straightforward combination of Lemmas~\ref{lem:packing},~\ref{lem:matchingmoments}, and~\ref{lem:ell2toell1}. 
As stated earlier, we will choose $R= \cmoments \log(1/\epsilon)$, and constants $\cmoments =8 \pi e$, 
$c_1 = 36$. 
First, by Lemma~\ref{lem:packing} we have that for at least $(1-\eta)$ fraction of the Gaussian mixtures in $\calF$, there is another mixture in the collection whose first $R=\cmoments \log(1/\eps)$ mean moments differ by at most $d^{-R/4}$ in symmetric injective tensor norm. To use Lemma~\ref{lem:matchingmoments} with $\eps' = \eps^2/4$, we see that 
\onote{8/24: this didn't really fix the problem: the problem is that Lemma 3.6 gives very weak guarantees for $d=1$; this is obviously stupid, and perhaps we can just start this proof by saying that we assume wlog $d\ge 2$; I would recommend reverting below to $d^{-R/4}$ otherwise it doesn't match the expression we derived above from Lemma 3.6}
\anote{8/29: I'm not sure what's the best way to handle this. Do you see a simple explanation why it is w.l.o.g?
I implemented another change instead: changed Lemma 3.6 statement and related statements from $d^{-R/4}$ to $(d+1)^{-R/4}$.}
\begin{align*}
\min_{r \in [R]} \eps_r &= \frac{\eps^2}{4} \min_{r \in [R]}\parens[\Big]{\frac{r}{ 8 \pi e \sqrt{\log(4/\eps^2)}}}^r \\ 
&\ge \frac{\eps^2}{4}\cdot 2^{-8 \pi e \sqrt{2\log(2/\eps)}} 
\ge 2^{- 2 \log(2/\eps) - 8 \pi e \sqrt{2\log(2/\eps)}} \ge 2^{-2\pi e \log(1/\eps)}\ge 2^{-R/4} \ge (d+1)^{-R/4},
\end{align*}
as required, 
where for the first inequality we used that $2^{-\alpha} < 1/\alpha$ for $\alpha>0$, and the second inequality uses $\log(1/\eps)\ge c_1 = 36$. We complete the proof by applying Lemma~\ref{lem:matchingmoments} (with $\eps$ in Lemma~\ref{lem:matchingmoments} taking the value $\eps^2/4$) and Lemma~\ref{lem:ell2toell1}. 
\end{proof}

%



\section{Iterative Algorithms for \texorpdfstring{$\min\set{\Omega(\sqrt{\log k}),\sqrt{d}}$}{min(Omega(sqrt(log k)),sqrt(d))} Separation} \label{sec:amplify}
We now give a new iterative algorithm that estimates the means of a mixture of $k$ spherical Gaussians up to arbitrary accuracy $\delta>0$ in $\poly(d,k,\log(1/\delta))$ time when the means have separation of order $\Omega(\sqrt{\log k})$ or $\Omega(\sqrt{d})$, when given coarse initializers. In all the bounds that follow, the most interesting setting of parameters is when $1/\delta$ is arbitrarily small compared to $k,d$ (e.g., $1/\wmin \le \poly(k)$ and $\delta= k^{-\omega(1)}$, or when $d=O(1)$ and $\delta =o(1)$). 
For sake of exposition, we will restrict our attention to the case when the standard deviations $\sigma_i$, and weights $w_i$ are known for all $i \in [k]$. We believe that similar techniques can be used to handle unknown $\sigma_i, w_i$ as well (see Remark~\ref{remark:unknown}).
We also note that the results of this section are interesting even in the case of uniform mixtures,
i.e., $w_i=1/k$ and $\sigma_i=1$ for all $i \in [k]$.

We will assume that we are given initializers that are inverse polynomially close in parameter distance. These initializers $\mu^{(0)}_1, \mu^{(0)}_2, \dots, \mu^{(0)}_k$ will be used to set up an ``approximate'' system of non-linear equations that has sufficient ``diagonal dominance'' properties, and then use the Newton method with the same initializers to solve it. In what follows, $\rhow$ and $\rhos$ denote the aspect ratio for the weights and variances respectively as defined in Section~\ref{sec:prelims}.

\begin{theorem}\label{thm:boxiterative}
  There exist universal constants $c, c_0>0$ such that the following holds.
  Suppose we are given samples from a mixture of $k$ spherical Gaussians $\calG$ having parameters $\set{(w_j, \mu_j,\sigma_j): j \in [k]}$, where the weights and variances are known, satisfying  
  \begin{equation}\label{eq:box:separation-condition}
  \forall i \ne j \in [k], ~ \norm{\mu_i - \mu_j}_2 \ge c(\sigma_i+\sigma_j) \cdot \min\set{\sqrt{d}+\sqrt{\log(\rho_w \rhos)}, \sqrt{\log(\rhos/\wmin)}} 
  \end{equation}
and suppose we are given initializers $\mu^{(0)}_1,\mu^{(0)}_2 , \dots, \mu^{(0)}_k$ satisfying
\begin{equation}\label{eq:iterative:initializer}
\forall j \in [k], ~  \frac{1}{\sigma_j} \norm{\mu^{(0)}_j  - \mu_j}_2 \le \frac{c_0}{\min\set{d,k}^{5/2}}.
\end{equation}
Then for any $\delta>0$, there is an iterative algorithm that runs in $\poly(\rho,d,1/\wmin,1/\delta)$ time (and samples), and after $T=O(\log\log(d/\delta))$ iterations recovers $\set{\mu_j : j \in [k]}$ up to  $\delta$ relative error w.h.p. i.e., finds $\set{\mu^{(T)}_j: j \in [k] }$ such that $\forall j \in [k]$, we have $\norm{\mu^{(T)}_j  - \mu_j}_2/\sigma_j  \le \delta $.
\end{theorem}

For standard mixtures, \eqref{eq:box:separation-condition} corresponds to a separation of order $\min\set{\sqrt{\log k}, \sqrt{d}}$. 

Firstly, we will assume without loss of generality that $d \le k$, since otherwise we can use a PCA-based dimension-reduction result 
due to Vempala and Wang~\cite{VW04}.  

 \begin{theorem}\label{thm:pca}
Let $\set{(w_i, \mu_i, \sigma_i): i \in [k]}$ be a mixture of $k$ spherical Gaussians that is $\rho$-bounded, and let $\wmin$ be the smallest mixing weight. Let $\mu'_1, \mu'_2, \dots, \mu'_k$ be the projections onto the subspace spanned by the top $k$ singular vectors of sample matrix $X \in \R^{d \times N}$. For any $\eps>0$, with $N \le \poly(d, \rho, \wmin^{-1},\eps^{-1})$ samples we have with high probability
 $$ \forall i \in [k],~ \norm{\mu_i - \mu'_i}_2 \le \eps.$$ 
 \end{theorem}
The above theorem is essentially Corollary 3 in~\cite{VW04}.
In~\cite{VW04}, however, they take the subspace spanned by the top $\max\{k, \log d\}$ singular vectors, 
most likely due to an artifact of their analysis. 
We give a different self-contained proof in Appendix~\ref{app:pca}. 
We will abuse notation, and use $\set{\mu_i: i \in [k]}$ to refer to the means in the dimension reduced space.
Note that after dimension reduction, the means are still well-separated for $c'>(c-1)$ i.e.,

\begin{equation}
\forall i,j \in [k],~ \norm{\mu_i - \mu_j} \ge c' (\sigma_i + \sigma_j)\min\set{\sqrt{d}+\sqrt{\log(\rho_w \rhos)}, \sqrt{\log(\rhos/\wmin)}}.
\end{equation}


%

\subsection{Description of the Non-linear Equations and Iterative Algorithm}

For each component $j \in [k]$ in $\calG$, we first define a region $S_j$ around $z_j$ as follows. 
We will show in Lemmas~\ref{lem:box:samecomponent} and~\ref{lem:box:othercomponentstotal} that the total probability mass in $S_j$ from other components is smaller than the probability mass from the component $j$.

\begin{definition}[Region $S_j$]\label{def:box:Sj}
Given initializers $z_1, z_2, \dots, z_k \in \R^d$, define $\ejl$ as the unit vector along $z_\ell - z_j$, and let
\begin{equation}
S_j = \set{x \in \R^d: \forall \ell \in [k] ~ \Abs{\iprod{x-z_j, \ejl}} \le 4 \sqrt{\log (\rhos/\wmin)} \sigma_j, ~\text{and}~\norm{x-z_j}_2 \le 4 (\sqrt{d}+\sqrt{\log(\rhos \rhow)}) \sigma_j}.
\end{equation}
\end{definition}

Based on those regions, we define the function $F:\R^{kd} \to \R^{kd}$ by
\begin{equation}
F_j(\vx) :=
   \frac{1}{w_j \sigma_j} 
	 \sum_{i=1}^k w_i \int_{y \in S_j} (y - z_j) g_{\sigma_i \vx_i, \sigma_i}(y) \dy ,
	\label{eq:box:mu}
\end{equation}
for $j=1,\ldots,k$, where $\vx = (\vx_1,\ldots,\vx_k) \in \R^{kd}$. 
We also define $\vx^* = (\mu_i /\sigma_i)_{i=1}^k$ as the intended solution. 
(Notice that for convenience, our variables correspond to ``normalized means'' $\mu / \sigma$ instead 
of with the means $\mu$ directly.) 
The system of non-linear equations is then simply
\begin{equation}\label{eq:systemofequations}
   F(\vx) = b
\end{equation}
where $b=F(\vx^*)$. This is a system of $kd$ equations in $kd$ unknowns. 
Obviously $\vx^*$ is a solution of the system, and if we could find
it we would be able to recover all the means, as desired. 

Our algorithm basically just applies Newton's method to solve~\eqref{eq:systemofequations}
with initializers given by $\vx^{(0)}_i = \mu^{(0)}_i/\sigma_i$ for $i \in [k]$. 
To recall, Newton's method uses the iterative update
\[
\vx^{(t+1)}=\vx^{(t)} - (F'(\vx^{(t)}))^{-1} (b - F(\vx^{(t)}) ), 
\]
where $F'(\vx^{(t)})$ is the first derivative matrix (Jacobian) of $F$ evaluated at $\vx^{(t)}$.
One issue, however, is that we are not given $b = F(\vx^*)$. 
Nevertheless, as we will show in Lemma~\ref{lem:box:sampling1}, 
we can easily estimate it to within any desired accuracy using 
a Monte Carlo approach based on the
given samples (from the Gaussian mixture corresponding to $\vx^*$). 
A related issue is that we do not have a closed-form expression for $F$ and $F'$,
but again, we can easily approximate their evaluation at any point
$\vx$ and to within any desired accuracy using a Monte Carlo approach (by generating
samples from the Gaussian mixture corresponding to $\vx$). 
The algorithm is given below in detail.

\vspace{10pt}

\def\compactify{\itemsep=0pt \topsep=0pt \partopsep=0pt \parsep=0pt}

\rule{0pt}{12pt}
\hrule height 0.8pt
\rule{0pt}{1pt}
\hrule height 0.4pt
\rule{0pt}{6pt}

\noindent \textbf{Iterative Algorithm for Amplifying Accuracy of Parameter Estimation}

\medskip

\noindent \textbf{Input:} Estimation accuracy $\delta>0$, $N$ samples 
$y^{(1)},\ldots,y^{(N)} \in \R^d$ 
from a mixture of well-separated Gaussians $\calG$ 
with parameters $\set{(w_j, \mu_j, \sigma_j): j \in [k]}$,
the weights $w_j$ and variances $\sigma_j^2$,
as well as initializers $\mu^{(0)}_i$ for each $i \in [k]$ such that $\norm{\mu^{(0)}_i - \mu_i}_\infty \le \eps_0 \sigma_i$. \onote{8/24: why is it now in the $\ell_2$ norm? do we ever address it?}\anote{8/28:I moved this conversion from $\ell_2$ to $\ell_\infty$ in the Proof of Theorem.}

\noindent \textbf{Parameters:} 
Set $T=C\log\log(d/\delta)$, for some sufficiently large constant $C>0$,
$\eps_0 = \basin$ where $c_0>0$ is a sufficiently small constant, 
and $\eta_1, \eta_2, \eta_3= \delta\wmin/(c' \sqrt{d}\rhos)$ where $c'>0$ is a sufficiently
small constant.
\anote{8/28: modified $\eta$ valued. was:$\eta_1, \eta_2, \eta_3= \delta/(8c' d^2 \rhos^2 k^{8})$}

\noindent \textbf{Output:} Estimates $\parens{\mu^{(T)}_i: i \in [k]}$ for each component $i \in [k]$ such that $\norm{\mu^{(T)}_i - \mu_i}_\infty \le \delta \sigma_i$.

\begin{enumerate} \compactify
\item If $\delta\ge \eps_0$, then we just output $\mu^{(T)}_i=\mu^{(0)}_i$ for each $i \in [k]$.
\item Set $\vx^{(0)}_i= \tfrac{1}{\sigma_i} \mu^{(0)}_i$ for each $i \in [k]$. 
\item 
Obtain using Lemma~\ref{lem:box:sampling1} an estimate $\tb$ of $b$  
up to accuracy $\eta_1$ (in $\ell_\infty$ norm). 
In more detail, define the empirical average 
\[
\forall j \in [k], \tb_j = \frac{1}{w_j \sigma_j N}\sum_{\ell \in [N]} \I[y^{(\ell)} \in S_j] \Paren{y^{(\ell)}-z_j} .
\]
	(This is the only place the given samples are used)

\item For $t=1$ to $T=O(\log \log (dk/\delta))$ steps do the following:
	\begin{enumerate}
			\item Obtain using Lemma~\ref{lem:box:sampling1} an estimate $\tF(\vx^{(t)})$ of $F(\vx^{(t)})$ at $\vx^{(t)}$ up to accuracy $\eta_2$ (in $\ell_\infty$ norm). 
		\item Obtain using Lemma~\ref{lem:box:sampling2} an estimate $\widetilde{F'}(\vx^{(t)})$ of $F'(\vx^{(t)})=\grad_{\vx} F(\vx^{(t)})$ up to accuracy $\eta_3$ (in $\infty \to \infty$ operator norm).
		\item Update $\vx^{(t+1)}=\vx^{(t)} - \Paren{\widetilde{F'}(\vx^{(t)})}^{-1} \Paren{\tilde{b} - \tF(\vx^{(t)}) }.$
	\end{enumerate}
\item Output $\mu^{(T)}_i=\sigma_i \vx^{(T)}_i$ for each $i \in [k]$.

\end{enumerate}

\hrule height 0.4pt
\rule{0pt}{1pt}
\hrule height 0.8pt
\rule{0pt}{12pt}


The proof of the two approximation lemmas below is based on a rather standard Chernoff argument; 
see Section~\ref{sec:samplingerrors}.
 
\begin{lemma} [Estimating $F_j(\vx)$] \label{lem:box:sampling1}
Suppose samples $y^{(1)}, y^{(2)}, \dots, y^{(N)}$ are generated from a mixture of $k$ spherical Gaussians with parameters $\{(w_i, \sigma_i\vx_i, \sigma_i)\}_{i \in [k]}$ in $d$ dimensions, and $\max\set{\norm{\vx_i}_\infty,1}\le \rho/\sigma_i~ \forall i \in [k]$. 
There exists a constant $C>0$ such that for any $\eta>0$ the following holds for $N\ge C \rho^3 \log (\tfrac{dk}{\gamma})/(\eta^2 \wmin)$ samples: with probability at least $1-\gamma$, for each $j \in [k]$ the empirical estimate for $F_j(\vx)$ has error at most $\eta$ i.e.,
\begin{equation}\label{eq:box:errorbound1}
\bignorm{ F_j(\vx) - \frac{1}{w_j \sigma_j N}\sum_{\ell \in [N]} \I[y^{(\ell)} \in S_j] \Paren{y^{(\ell)}-z_j} }_\infty < \eta,
\end{equation}
where $S_j$ is defined in Definition~\ref{def:box:Sj}.
\end{lemma}

\onote{8/25: moved this here:}
We now state a similar lemma for the Jacobian $F':\R^{k\cdot d} \rightarrow \R^{k\cdot d}$,
which by an easy calculation is given for all $i,j \in [k]$ by
\onote{8/24: should we avoid notation $J$ and instead always write $F'$?}\anote{8/28:modified the algorithm and newton section to do this.}
\onote{8/25: added missing factor 2... need to fix everywhere else too!}\anote{8/28: I think I propagated it fully.}
\[
\grad_{\vx_i} F_j (\vx) =
   \frac{2 \pi w_i}{w_j \sigma_j \sigma_i}  
	 \int_{y \in S_j} (y - z_j) (y- \sigma_i \vx_i)^T g_{\sigma_i \vx_i, \sigma_i}(y) \dy  \nonumber 
\]
i.e., for all $r, r_0 \in [d]$, 
\begin{align}
\frac{\partial F_{j, r_0}(\vx)}{\partial \vx_i(r)}
&= \frac{2 \pi w_i}{w_j \sigma_j \sigma_i} 
   \int_{y \in S_j} (y(r_0) - z_j(r_0)) (y(r_1)- \sigma_i \vx_i(r_1)) g_{\sigma_i \vx_i, \sigma_i}(y) \dy \label{eq:box:grad}
\end{align}

\begin{lemma} [Estimating $\grad_{\vx_i} F_j(\vx)$] \label{lem:box:sampling2}
Suppose we are given the parameters of a mixture of $k$ spherical Gaussians $\{(w_i, \sigma_i\vx_i, \sigma_i)\}_{i \in [k]}$ in $d$ dimensions, with $\max\set{\norm{\vx_i}_\infty,1}\le \rho/\sigma_i$ for each $i \in [k]$, and region $S_j$ is defined as in Definition~\ref{def:box:Sj}.
There exists a constant $C>0$ such that for any $\eta>0$ the following holds for $N\ge C \rho^4 d^2 k^2 \log (\tfrac{dk}{\gamma})/(\eta^2 \wmin)$. Given $N$ samples $y^{(1)}, y^{(2)}, \dots, y^{(N)}$ generated from spherical Gaussian with mean $\sigma_i \vx_i$ and variance $\sigma_i^2/(2\pi)$, we have with probability at least $1-\gamma$, for each $j \in [k]$ the following empirical estimate for $\grad_{\vx_i} F_j$ has error at most $\eta$ i.e.,
\begin{equation}\label{eq:box:errorbound2}
\bignorm{ \grad_{\vx_i} F_j(\vx) - \widetilde{\grad_{\vx_i}F_j}(\vx) }_{\infty \to \infty}< \frac{\eta}{k}, ~\text{ where } \widetilde{\grad_{\vx_i}F_j}(\vx):=\frac{2 \pi w_i}{w_j \sigma_i \sigma_j N}\sum_{\ell \in [N]} \I[y^{(\ell)} \in S_j] \Paren{y^{(\ell)}-z_j} \Paren{y^{(\ell)}-\sigma_i \vx_i}^T.
\end{equation}
Furthermore, we have that the estimated second derivative $\widetilde{F'}(\vx)=(\widetilde{\grad_{\vx_i}F_j}(\vx): i, j \in [k])$ satisfies $\bignorm{F'(\vx) - \widetilde{F'}(\vx)}_{\infty \to \infty} \le \eta$.
\end{lemma}

\onote{8/25: the lemma gives us estimate in $\infty\to\infty$ norm of \emph{each block} separately
even though Appendix B requires an estimate of the error in the entire matrix; do we address 
this somewhere?} \anote{8/28: addressed it now.}

\begin{remark}\label{remark:unknown}
Although for simplicity we focus here on the case that only the means are unknown, 
we believe that our approach can be extended to the case that the weights and variances 
are also unknown, and only coarse estimates for them are given. 
In order to handle that case, we need to collect more statistics about the given
samples, in addition to just the mean in each region. 
Namely, using the samples restricted to $S_j$, we estimate the total probability mass in each $S_j$, i.e., 
$b^{(w)}_j= \E_{y \leftarrow \calG} \I [y \in S_j]$, and the average squared Euclidean norm 
$b^{(\sigma)}_j=\tfrac{1}{d}\E_{y \leftarrow \calG} \I[y \in S_j] \norm{y}_2^2$. 
We now modify~\eqref{eq:systemofequations} by adding new unknowns (for the weights and variances), 
as well as new equations for $b^{(\sigma)}_j$ and $b^{(w)}_j$. 
This corresponds to $k(d+2)$ non-linear equations with $k(d+2)$ unknowns.
\end{remark}

\subsection{Convergence Analysis using the Newton method}

We will now analyze the convergence of the Newton algorithm. 
We want each parameter $\vx^{(T)}_i \in \R^d$ to be close to $\vx^{*}_i$ in an appropriate norm (e.g., $\ell_2$ or $\ell_\infty$). Hence, we will measure the convergence and error of $\vx=\parens{\vx_i: i \in [k]}$ to be measured in $\ell_\infty$ norm.

\begin{definition}[Neighborhood]\label{def:box:neigh}
  Consider a mixture of Gaussians with parameters $\parens{(\mu_i, \sigma_i, w_i): i \in [k]}$, and let $\parens{\vx_i: i \in [k]} \in \R^{kd}$ be the corresponding parameters of the non-linear system $F(\vx)=b$. The neighborhood set
  $$\calN = \set{\parens{\vx_i: i \in [k]} \in \R^{kd} ~\mid~ \forall i \in [k], \norm{\vx_i - \vx^*_i}_\infty < \eps_0=\basin}, $$
  is the set of values of the variables that are close to the true values of the variables given by $\vx^*_i=\tfrac{\mu_i}{\sigma_i}~ \forall i \in [k] $, and $c_0>0$ is an appropriately large universal constant given in Theorem~\ref{thm:boxiterative}.
\end{definition}

We will now show the convergence of the Newton method by showing diagonal dominance properties of the non-linear system given in Lemma~\ref{lem:diagdominant}. This diagonal dominance arises from the separation between the means of the components. Lemmas~\ref{lem:box:samecomponent} and ~\ref{lem:box:othercomponentstotal} show that most of the probability mass from $j$th component around $\mu_j$ is confined to $S_j$, while the other components of $\calG$ are far enough from $z_j$, that they do not contribute much $\ell_1$ mass in total to $S_j$.

The following lemma lower bounds the contribution to region $S_j$ from the $j$th component.

\begin{lemma} \label{lem:box:samecomponent}
In the notation of Theorem~\ref{thm:boxiterative}, for all $j \in [k]$, we have
\begin{align}
& \int_{S_j} g_{\mu_j, \sigma_j}(y) \dy \ge 1- \frac{1}{8 \pi d}.\label{eq:box:sc:0}\\
\forall r_1 \in [d], ~& \frac{1}{\sigma_j} \Abs{\int_{S_j} \left( y(r_1) - \mu_j(r_1) \right) g_{\mu_j, \sigma_j}(y)  \dy } \le  \frac{1}{8 \pi d} .\label{eq:box:sc:1}\\
\forall r_1, r_2 \in [d], ~& \frac{1}{\sigma_j^2}\Abs{\int_{S_j} \left( y(r_1) - \mu_j(r_1) \right) \left( y(r_2) - \mu_j(r_2) \right) g_{\mu_j, \sigma_j}(y) \dy  - \frac{\sigma_j^2}{2 \pi} I[r_1=r_2]}  \le \frac{1}{8 \pi d} \label{eq:box:sc:2}.
\end{align}
\end{lemma}
The proof of the above lemma follows from concentration bounds for multivariate Gaussians. 
The following lemma upper bounds the contribution in $S_j$ from the other components. 

\begin{lemma}[Leakage from other components]\label{lem:box:othercomponentstotal}
In the notation of Theorem~\ref{thm:boxiterative}, and for a component $j \in [k]$ in $\calG$, 
the contribution in $S_j$ from the other components is small, namely, for all $j \in [k]$ 
\begin{align}
& \sum_{i\in [k], i \ne j} w_i \int_{S_j} g_{\mu_i, \sigma_i}(y) \dy <  \frac{w_j}{16 \pi d}.\label{eq:box:other0}\\
 \forall r_1 \in [d], ~& \sum_{i\in [k], i \ne j} \frac{w_i}{\sigma_i} \cdot\frac{\norm{\mu_i-\mu_j}_2}{\sigma_i} \int_{S_j} \Abs{y(r_1) - \mu_i(r_1)} g_{\mu_i, \sigma_i}(y) \dy <  \frac{w_j}{16 \pi d \rhos}.\label{eq:box:other1}\\
\forall r_1, r_2 \in [d], ~& \sum_{i\in [k], i \ne j} \frac{w_i}{\sigma_i \sigma_j} \int_{S_j} \Abs{y(r_1) - \mu_i(r_1)} \Abs{y(r_2) - \mu_i(r_2)} g_{\mu_i, \sigma_i}(y) \dy < \frac{w_j}{16 \pi d \rhos}. \label{eq:box:other2}
\end{align}
\anote{8/28:Modified the second equation replacing $\max\set{1,\frac{\norm{\mu_i-\mu_j}_2}{\sigma_i}}$ by just the expression, since always larger than 1.}
\end{lemma}
The above lemma is the more technical of the two, and it is crucial in showing diagonal dominance of the Jacobian. The proof of the above lemma is very different for separation of order $\sqrt{\log k}$ and $\sqrt{d}$ -- hence we will handle this separately in Sections~\ref{sec:separation-k} and \ref{sec:separation-d} respectively. 

We now show the convergence of Newton algorithm assuming the above two lemmas. Theorem~\ref{thm:boxiterative} follows in a straightforward manner from the guarantees of the Newton algorithm. We mainly need to show that $\norm{F'} \norm{F''} \eps_0 < 1/2$. 

We now prove that the function $(F'(\vx))^{-1}$ has bounded operator norm using the diagonal dominance properties of $F$.
\begin{lemma} \label{lem:box:gradinv}
For any point $\vx \in \calN$, the operator $F': \R^{d\cdot k} \rightarrow \R^{d \cdot k}$ satisfies   $\norm{(F'(\vx))^{-1}}_{\infty \to \infty} \le 4$. 
\end{lemma}
\begin{proof}
We will divide the matrix $F'$ into $k \times k=k^2$ blocks of size $d \times d$ each, and show that the matrix satisfies the required diagonal dominance property. Let us first consider the diagonal blocks i.e., we have from \eqref{eq:box:grad} for each $j \in [k]$
\begin{align*}
\grad_{\vx_j} F_j(\vx) &= \frac{2\pi}{\sigma_j^2}\int_{y \in S_j} (y - z_j) (y - \sigma_j \vx_j)^T g_{\sigma_j \vx_j, \sigma_j}(y) \dy\\
&= \frac{2\pi}{\sigma_j^2} \int_{y \in S_j} (y - \sigma_j \vx_j) (y - \sigma_j \vx_j)^T g_{\sigma_j \vx_j, \sigma_j}(y) \dy + \frac{2\pi (\sigma_j \vx_j - z_j)}{\sigma_j^2} \int_{y \in S_j} (y - \sigma_j \vx_j)^T g_{\sigma_j \vx_j, \sigma_j}(y) \dy
\end{align*}
Consider a mixture of Gaussians where the $j$th component has mean $\sigma_j \vx_j$ and standard deviation $\sigma_j/\sqrt{2 \pi}$. 
\onote{8/7 note to self: I found the following confusing}
It satisfies the required separation conditions since $\norm{\sigma_i \vx_i - \mu_i}_2 \le \eps_0 \sigma_i$. Applying Lemma~\ref{lem:box:samecomponent},
\begin{align*}
\grad_{\vx_j} F_j(\vx) &= I_{d \times d}  - \frac{2\pi}{\sigma_j^2} \int_{y \notin S_j} (y - \sigma_j \vx_j) (y - \sigma_j \vx_j)^T g_{\sigma_j \vx_j, \sigma_j}(y) \dy\\
&\qquad\qquad - \frac{2\pi (\sigma_j \vx_j - z_j)}{\sigma_j} \int_{y \notin S_j} \frac{(y - \sigma_j \vx_j)^T}{\sigma_j} g_{\sigma_j \vx_j, \sigma_j}(y) \dy \\
& = I_{d \times d} - E,
\end{align*}
where $E \in \R^{d \times d}$ satisfies from \eqref{eq:box:sc:2} and \eqref{eq:box:sc:1} that
\begin{align}
\forall r_1, r_2 \in [d],~~|E(r_1, r_2)|&\le \frac{1}{4 d}+ \eps_0 \cdot \frac{1}{4 d} \le \frac{1}{2 d} \nonumber\\
\forall r_1 \in [d],~~\sum_{r_2 \in [d]}|E(r_1, r_2)| &\le \frac{1}{2}.
\label{eq:box:diagonalentries}
\end{align}

\noindent Using a similar calculation, we see that for the off-diagonal blocks $M_{ji} = \grad_{\vx_i} F_j(\vx)$,
\begin{align*}
M_{ji} = \grad_{\vx_i} F_j(\vx) &= \frac{2\pi w_i}{w_j \sigma_i \sigma_j} \int_{y \in S_j} (y - z_j) (y - \sigma_i \vx_i)^T g_{\sigma_i \vx_i, \sigma_i}(y) \dy \\
&= \frac{2\pi w_i}{w_j} \int_{y \in S_j} \frac{(y - \sigma_i \vx_i) (y - \sigma_i \vx_i)^T}{\sigma_i \sigma_j} g_{\sigma_i \vx_i, \sigma_i}(y) \dy \\
&\quad ~~+  \frac{2\pi w_i(\sigma_i \vx_i - z_j)}{w_j \sigma_j} \int_{y \in S_j} \frac{(y - \sigma_i \vx_i)^T}{\sigma_i} g_{\sigma_i \vx_i, \sigma_i}(y) \dy 
\end{align*}
Also, $\norm{\sigma_i \vx_i - z_j}_\infty \le \norm{\mu_i - \mu_j}_\infty + \sigma_i+ \sigma_j  \le 2 \norm{\mu_i - \mu_j}_2$. 
For each $r_1, r_2 \in [d]$
\begin{align*}
\Abs{M_{ji}(r_1, r_2)}&\le \frac{2\pi w_i \rhos}{w_j \sigma^2_i} \int_{y \in S_j} \Abs{y(r_1) - \sigma_i \vx_i(r_1)} \Abs{y(r_2) - \sigma_i \vx_i(r_2)} g_{\sigma_i \vx_i, \sigma_i}(y) \dy \\
&\quad +  \frac{4 \pi w_i \rhos}{w_j} \frac{\norm{\mu_i - \mu_j}_2}{\sigma_i} \int_{y \in S_j} \frac{\Abs{y(r_2) - \sigma_i \vx_i(r_2)}}{\sigma_i} g_{\sigma_i \vx_i, \sigma_i}(y) \dy.
\end{align*}
Summing over all $i \ne j$, and using the bounds in Lemma~\ref{lem:box:othercomponentstotal},
\begin{align}
\sum_{i \in [k], i \ne j} \Abs{M_{ji}(r_1, r_2)}&\le \frac{1}{4 d} \nonumber\\
\sum_{r_2=1}^d \sum_{i \ne j} \abs{M_{ji}(r_1, r_2)} &\le \frac{1}{4}  \label{eq:box:offdiagonals}
\end{align}

\noindent Hence using Lemma~\ref{lem:diagdominant} for diagonally dominant matrices, we get from \eqref{eq:box:diagonalentries} and \eqref{eq:box:offdiagonals}
$$\norm{(F'(\vx))^{-1}}_{\infty \to \infty} \le \frac{1}{1 - \tfrac{1}{2}-\tfrac{1}{4}} \le 4.$$
\end{proof}

We now prove that the function $F'(\vx)$ is locally Lipschitz by upper bounding the second derivative operator.
\anote{8/31: Rewrote statement in terms of being locally Lipschitz.}
\begin{lemma} \label{lem:gradLipschitz}
There exists a universal constant $c'>0$ such that the derivative $F'$ is locally $L$-Lipschitz in the neighborhood $\calN$ for $L \le c' d^{5/2}$. 
\end{lemma}
\begin{proof}
We proceed by showing that the operator norm of $F'':\R^{dk \times dk} \rightarrow \R^{dk}$ is upper bounded by $L$ at any point $\vx \in \calN$. 
  We first write down expressions for $F''$, and then prove that the operator norm of $F''$ is bounded. We first observe that for each $j, i_1, i_2 \in [k]$,
  $$ \forall r_0, r_1, r_2 \in [d], ~\frac{\partial^2 F_{j,r_0}(\vx)}{\partial \vx_{i_1}(r_1) \partial \vx_{i_2}(r_2)} =0 \text{ if } i_1 \ne i_2.$$
Hence the second derivatives are non-zero only for diagonal blocks $i_1 = i_2$. For each $j, i \in [k]$ the second derivatives for $\forall r_0, r_1, r_2 \in [d]$ are given by
\begin{align*}
\frac{\partial^2 F_{j,r_0}(\vx)}{\partial \vx_{i}(r_1) \partial \vx_{i}(r_2)}
&= \frac{4 \pi w_i}{w_j \sigma_j} \int_{y \in S_j} (y(r_0) - z_j(r_0)) \parens[\Big]{\frac{\parens{y(r_1)- \sigma_i \vx_i(r_1)} \parens{y(r_2)- \sigma_i \vx_i(r_2)}}{\sigma_i^2}  - \I[r_1 = r_2]} g_{\sigma_i \vx_i, \sigma_i}(y) \dy \\
\Abs{\frac{\partial^2 F_{j,r_0}(\vx)}{\partial \vx_{i}(r_1) \partial \vx_{i}(r_2)}}
&\le \frac{4 \pi w_i}{w_j} \max_{y \in S_j} \frac{\norm{y - z_j}_2}{\sigma_j} \cdot \int_{y \in S_j}  \parens[\Big]{ \Abs{\frac{\parens{y(r_1)- \sigma_i \vx_i(r_1)} \parens{y(r_2)- \sigma_i \vx_i(r_2)}}{\sigma_i^2}}+1 } g_{\sigma_i \vx_i, \sigma_i}(y) \dy.\\
&\le \frac{8 \pi w_i \sqrt{d}}{w_j} \int_{y \in S_j} \parens[\Big]{ \Abs{\frac{\parens{y(r_1)- \sigma_i \vx_i(r_1)} \parens{y(r_2)- \sigma_i \vx_i(r_2)}}{\sigma_i^2}}+1 } g_{\sigma_i \vx_i, \sigma_i}(y) \dy.
\end{align*}
A simple bound on the operator norm of $F'': \R^{kd} \times \R^{kd} \rightarrow \R^{kd}$ is given by summing over all $i \in [k]$ and $r_1, r_2 \in [d]$. Consider a mixture of Gaussians where the $j$th component has mean $\sigma_j \vx_j$ and standard deviation $\sigma_j/\sqrt{2\pi}$. It satisfies the required separation conditions since $\norm{\sigma_i \vx_i - \mu_i}_2 \le \eps_0 \sigma_i$. Applying Lemma~\ref{lem:box:othercomponentstotal} we get
\begin{align*}
\norm{F''(\vx)}_{\infty, \infty \to \infty} &\le \sum_{i \in [k]} \sum_{r_1, r_2 \in [d]} \Abs{\frac{\partial^2 F_{j,r_0}(\vx)}{\partial \vx_{i}(r_1) \partial \vx_{i}(r_2)}}\\
\sum_{i \ne j} \sum_{r_1, r_2 \in [d]} \Abs{\frac{\partial^2 F_{j,r_0}(\vx)}{\partial \vx_{i}(r_1) \partial \vx_{i}(r_2)}} 
&\le \frac{16 \pi d^{5/2}}{w_j}\cdot \frac{w_j}{16 \pi d} \le d^{3/2}, ~~\text{and }\\
\sum_{r_1, r_2 \in [d]} \Abs{\frac{\partial^2 F_{j,r_0}(\vx)}{\partial \vx_{j}(r_1) \partial \vx_{j}(r_2)}} 
&\le 8 \pi d^{5/2} \int_{y \in \R^d}  \Abs{\frac{\parens{y(r_1)- \sigma_j \vx_j(r_1)} \parens{y(r_2)- \sigma_j \vx_j(r_2)}}{\sigma_j^2}} g_{\sigma_j \vx_j, \sigma_j}(y) \\
&~\quad~+ 8 \pi d^{5/2} \int_{y \in \R^d} g_{\sigma_j \vx_j, \sigma_j}(y) \dy\\
&\le 8 \pi d^{5/2} \parens[\Big]{\frac{1}{\pi^2}+1} \le 15 \pi d^{5/2}\\
\forall \vx \in \calN,~ \norm{F''(\vx)}_{\infty, \infty \to \infty}& \le 16 \pi d^{5/2}. 
\end{align*}
Hence, applying Lemma~\ref{lem:meanvalue} with $F'$ in the open set $K=\calN$, the lemma follows.   
\end{proof}

\begin{proof}[Proof of Theorem~\ref{thm:boxiterative}]
We first note that we have from \eqref{eq:iterative:initializer},
 for each $i \in [k]$ that $\norm{\vx^*_i - \vx^{(0)}_i}_\infty \le \norm{\vx^*_i - \vx^{(0)}_i}_2 \le \norm{\mu_i-\mu^{(0)}_i}_2/\sigma_i \le \eps_0$. We will use the algorithm to find $\norm{\vx^{(T)}-\vx^*}_\infty < \delta/\sqrt{d}$ (use the algorithm with $\delta':=\delta/\sqrt{d}$). 
The proof follows in a straightforward manner from Corollary~\ref{corr:newton:errors}. \anote{8/31:Changed to reflect local Lipschitz.} Lemma~\ref{lem:gradLipschitz} shows that $F'$ is locally $L$-Lipschitz for $L=c' d^{5/2}$. Together with
Lemma~\ref{lem:box:gradinv} we have $\eps_0 L \norm{(F')^{-1}}_{\infty \rightarrow \infty}  \le 1/2$ for our choice of $\eps_0$. Also from Lemma~\ref{lem:box:sampling1},  by using $N=O(\delta^{-2} \rho^6 k^{6}\wmin^{-3})$ samples (and $d \le k$), $\eta_1+ \eta_2 \le \frac{\delta}{4\sqrt{d}\norm{(F')^{-1}}_{\infty \rightarrow \infty}}$. It is easy to check that $B=\max_y \norm{F'(y)}\le \frac{\max_i \sigma_i}{\wmin \min_i \sigma_j} \le 4\rho_\sigma/\wmin$. Hence, similarly from Lemma~\ref{lem:box:sampling2} $\eta_3 \le \frac{\delta}{4\sqrt{d}B \norm{(F')^{-1}}_{\infty \rightarrow \infty}^2}$ as required. Hence, from Corollary~\ref{corr:newton:errors}, $T=O(\log\log(d/\delta))$ iterations of the Newton's method suffices. Further, each iteration mainly involves inverting a $dk \times dk$ matrix which is polynomial time in $d,k$. 
\end{proof}

\subsection{Bounding Leakage from Own Component: Proof of Lemma~\ref{lem:box:samecomponent}}

\begin{lemma} \label{lem:box:helper}
Let $\ev_0 \in \R^d$ be a unit vector. Then for any coordinates $r_1, r_2 \in [d]$ we have for any $t\ge 3$
\begin{align}
\frac{1}{\sigma^d} \int_{y: \iprod{y, \ev_0}\ge t \sigma} \exp(- \pi \norm{y}_2^2 / \sigma^2) \dy  &< \tfrac{1}{16 \pi} \exp(- t^2) \label{eq:box:helper0}\\
\frac{1}{\sigma^d} \int_{y: \iprod{y, \ev_0}\ge t \sigma} \abs{y(r_1)} \exp(- \pi \norm{y}_2^2  / \sigma^2) \dy  &< \tfrac{1}{16 \pi} \exp(- t^2) \sigma \label{eq:box:helper1}\\
\frac{1}{\sigma^d} \int_{y: \iprod{y, \ev_0}\ge t \sigma} \Abs{y(r_1) y(r_2)} \exp(- \pi \norm{y}_2^2  / \sigma^2) \dy  &< \tfrac{1}{16 \pi} \exp(- t^2) \sigma^2 \label{eq:box:helper2}
\end{align}
\end{lemma}
\begin{proof}
The first equation~\eqref{eq:box:helper0} follows easily from the rotational invariance of a spherical Gaussian. Suppose $g_0=\iprod{y, \ev_0}$, then $g_0 \sim N(0,\sigma^2/(2\pi))$. Hence
$$\frac{1}{\sigma^d} \int_{y: \iprod{y, \ev_0}\ge t \sigma} \exp(- \pi \norm{y}_2^2/\sigma^2) \dy = 
  \frac{1}{\sigma}\int_{g_0 \ge t \sigma} \exp(- \pi g_0^2 / \sigma^2) ~d{g_0} 
	~\le \Phi_{0,1}\left(\sqrt{2 \pi } t \right) < \exp(- t^2) .$$

We will now prove \eqref{eq:box:helper2}; the proof of \eqref{eq:box:helper1} follows the same argument. We first observe that using the rotational invariance, it suffices to consider the span of $\ev_0, \er{1}, \er{2}$. 

Consider the case $r_1 \ne r_2$. Let $\ev_1, \ev_2$ be the unit vectors along $\er{1} - \iprod{\er{1}, \ev_0} \ev_0$ and $\er{2} - \iprod{\er{2}, \ev_0} \ev_0 - \iprod{\er{2}, \ev_1} \ev_1$ respectively. If $g_\ell =\iprod{y,\ev_\ell}$ for $\ell=0,1,2$, then $\set{g_\ell}$ are mutually independent and distributed as $N(0, \sigma^2/(2\pi))$. Let $\alpha_1=\iprod{\ev_0, \er{1}}$, $\alpha_2=\iprod{\ev_0, \er{2}}$ and $\beta_2 = \iprod{\ev_1, \er{2}}$. Then, 
\begin{align*}
y(r_1)= \alpha_1 g_0 + \sqrt{1 - \alpha_1^2}\cdot g_1, &~\text{ and } y(r_2)= \alpha_2 g_0 + \beta_2 g_1 + \sqrt{1 - \alpha_1^2 - \beta_2^2} \cdot g_2, \\
\abs{y(r_1)} \le \abs{g_0}+\abs{g_1}, \text{ and }\abs{y(r_1) y(r_2)} &\le (\abs{g_0}+\abs{g_1})(\abs{g_0}+\abs{g_1}+\abs{g_2}) \le 3( \abs{g_0}^2 + \abs{g_1}^2+\abs{g_2}^2).
\end{align*}
Hence, 
\begin{align*}
&\frac{1}{\sigma^d}\int_{y: \iprod{y, \ev_0}\ge t \sigma} \Abs{y(r_1)y(r_2)} \exp(- \pi \norm{y}_2^2) \dy \le \frac{3}{\sigma^d}\int_{y: \iprod{y, \ev_0}\ge t \sigma} \parens[\big]{\abs{g_0}^2+\abs{g_1}^2+\abs{g_2}^2} \exp(- \pi \norm{y}_2^2) \dy  \\
&~\le \frac{3}{\sigma^3} \int_{g_0 > t \sigma} \int_{g_1} \int_{g_2} (\abs{g_0}^2+ \abs{g_1}^2+\abs{g_2}^2) \cdot \exp(-\pi g_2^2) ~d{g_2} \cdot \exp(-\pi g_1^2) ~d{g_1} \cdot \exp(-\pi g_0^2) ~d{g_0}\\
&~\le \frac{3}{\sigma}
\int_{g_0 > t \sigma} \abs{g_0}^2 \exp(-\pi g_0^2) ~d{g_0}+ \frac{6}{\sigma} \int_{g_1} \abs{g_1}^2 \exp(-\pi g_1^2) ~d{g_1} \cdot \frac{1}{\sigma}\int_{g_0 > t \sigma} \exp(-\pi g_0^2) ~d{g_0}~~(\text{by symmetry} )\\
&~\le 3\exp(-2t^2) \sigma^2+ 6 \sigma^2 \exp(-2t^2)  \le 9 \sigma^2 \exp(-2t^2),
\end{align*}
 where the last line follows from truncated moments of a normal variable with variance $\sigma^2/(2\pi)$ (Lemma~\ref{lem:app:truncmoments}). 
Using the fact that $\exp(t^2) > 9\cdot 16 \pi$ for $t\ge 3$, the lemma follows.  The proof for $r_1=r_2$, and for \eqref{eq:box:helper1} follows from an identical argument. 
\end{proof}

\begin{proof}[Proof of Lemma~\ref{lem:box:samecomponent}]
We will consider the loss from the region outside of $S_j$. 
The region $S_j$ is defined by $k$ linear inequalities, one for each of the components $\ell \in [k]$, and a constraint that points lie in an $\ell_2$ ball around $z_j$. Consider $y$ drawn from $g_{\mu_j, \sigma_j}$. 
Since $\norm{z_j - \mu_j}_2 \le \sigma_j$ and Lemma~\ref{lem:lengthconc} (applied with $t=6d$),
\onote{8/7: why does the first inequality hold? (sorry if it's obvious, I'm still digesting this part)}\anote{8/8: Added some explanation. Does this clarify things?}
$$\Pr[\norm{y-z_j}_2 \ge 4\sqrt{d} \sigma_j] \le \Pr[\norm{y-\mu_j}_2 \ge 3\sqrt{d} \sigma_j] \le \exp(-6d) \le \frac{1}{16\pi d}.$$
 When $y$ is drawn from $g_{\mu_j, \sigma_j}$, from Lemma~\ref{lem:box:helper}, each of the $k$ linear constraints are satisfied with high probability i.e., for a fixed $\ell \in [k]$,
\begin{equation}\label{eq:subtract}
\Pr\Big[\Abs{\iprod{y-z_j, \ejl}} > 4 \sigma_j \sqrt{\log (1/ \wmin)} \Big] \le \Pr\Big[\Abs{\iprod{y-\mu_j, \ejl}} > 3 \sigma_j \sqrt{\log (1/ \wmin)} \Big] \le \frac{\wmin^2}{16 \pi} \le \frac{1}{k} \cdot \frac{\wmin}{16 \pi},
\end{equation}
where the last step follows since $\wmin \le 1/k$. Hence, performing a union bound over the $k$ components and the $\ell_2$ constraint, and using $d \le k$, we get
$$\int_{S_j} g_{\mu_j, \sigma_j}(y) \dy \ge 1- \frac{1}{16 \pi d} - \frac{1}{16 \pi d} \ge 1 - \frac{1}{8 \pi d}.$$

To prove \eqref{eq:box:sc:1}, $\forall r_1 \in [d]$ we have
\begin{align*}
\int_{S_j} \left( y(r_1) - \mu_j(r_1) \right) g_{\mu_j, \sigma_j}(y)  \dy &= \int_{\R^d} \left( y(r_1) - \mu_j(r_1) \right) g_{\mu_j, \sigma_j}(y)\dy - \int_{y \notin S_j} \left( y(r_1) - \mu_j(r_1) \right) g_{\mu_j, \sigma_j}(y) \dy \\
&=0 - \int_{y \notin S_j} \left( y(r_1) - \mu_j(r_1) \right) g_{\mu_j, \sigma_j}(y) \dy
\end{align*}
Any point $y \notin S_j$ implies that one of the $k$ linear constraints $\iprod{y-\mu_j, \eji} > 3 \sqrt{\log(1/\wmin)}\sigma_j$, or $\norm{y - \mu_j}_2 > 3 \sqrt{d} \sigma_j$ is true. Again, for the $\ell_2$ ball constraint, from Lemma~\ref{lem:app:truncmoments:highd} with $q \le 2$, 
$$\Abs{\int_{\norm{y - \mu_j}_2 > 3 \sqrt{d} \sigma_j} \frac{\left( y(r_1) - \mu_j(r_1) \right)}{\sigma_j}\cdot g_{\mu_j, \sigma_j}(y)  \dy } \le \Abs{\int_{\norm{y - \mu_j}_2 > 3 \sqrt{d} \sigma_j} \frac{\norm{y-\mu_j}_2}{\sigma_j} \cdot g_{\mu_j, \sigma_j}(y)  \dy } \le \frac{1}{16 \pi d}.$$
Let $Z_{ji}=\set{y: \iprod{y-\mu_j, \eji} > 3 \sigma_j \sqrt{\log(\rhos/\wmin)}}$ be the set $y$ that do not satisfy the constraint along $\eji$. Since $\norm{z_j -\mu_j}_2 \le \sigma_j$, we have that the total loss from those $y \in Z_{ji}$  is 
\begin{align*}
\Abs{\int_{y \in Z_{ji}} \left( y(r_1) - \mu_j(r_1) \right) g_{\mu_j, \sigma_j}(y) \dy }
&\le \int_{y \in Z_{ji}} \Abs{y(r_1) - \mu_j(r_1)} g_{\mu_j, \sigma_j}(y) \dy\\
&\le \sigma_j \Paren{\frac{\wmin^2}{16 \pi \rhos^2}},
\end{align*}
by applying Lemma~\ref{lem:box:helper} with the Gaussian $(y-\mu_j)$, and $t=3 \sqrt{\log(\rhos/\wmin)}$. Hence, we have

\begin{align*}
\Abs{\int_{S_j} \left( y(r_1) - \mu_j(r_1) \right) g_{\mu_j, \sigma_j}(y) \dy} &\le \Abs{\int_{y \notin S_j} \left( y(r_1) - \mu_j(r_1) \right) g_{\mu_j, \sigma_j}(y) \dy} \\
&\le \sigma_j \Paren{\frac{\wmin^2}{16 \pi \rhos^2}}k + \frac{\sigma_j}{16 \pi d}\le \sigma_j \parens[\Big]{\frac{1}{8 \pi d}},
\end{align*}
since $\wmin \le 1/k$. 
The proof of the last equation follows in an identical fashion. From \eqref{eq:box:helper2} of Lemma~\ref{lem:box:helper}, we have
$$ \forall r_1, r_2 \in [d], ~ \frac{1}{\sigma_j^2} \Abs{\int_{y \notin S_j} \parens[\big]{ y(r_1) - \mu_j(r_1)} \parens[\big]{ y(r_2) - \mu_j(r_2)} g_{\mu_j, \sigma_j}(y) \dy} \le \frac{1}{8 \pi d}.$$
Further,  
\begin{align*}
\forall r_1, r_2 \in [d], ~ \int_{\R^d} \left( y(r_1) - \mu_j(r_1) \right) \left( y(r_2) - \mu_j(r_2) \right) g_{\mu_j, \sigma_j}(y) \dy  = \frac{\sigma_j^2}{2 \pi} I[r_1=r_2] \\
\text{ Hence, }~\Abs{\frac{1}{\sigma_j^2}\int_{S_j} \left( y(r_1) - \mu_j(r_1) \right) \left( y(r_2) - \mu_j(r_2) \right) g_{\mu_j, \sigma_j}(y) \dy - \frac{1}{2\pi} I[r_1=r_2]}\le \frac{1}{8 \pi d}. 
\end{align*}

\end{proof}

\subsection{Bounding Leakage from Other Components: Proof of Lemma~\ref{lem:box:othercomponentstotal}}

Since our algorithm works when we either have separation of order $\sqrt{\log k}$ or separation of order $\sqrt{d}$, we have two different proofs for Lemma~\ref{lem:box:othercomponentstotal} depending on whether $\sqrt{\log(1/\wmin)} \le \sqrt{d}$ or not\footnote{More accurately whether $\sqrt{\log(\rhos/\wmin)} \le \sqrt{d}+\sqrt{\log(\rho_w \rhos)}$, where $\rhow \ge 1/\wmin$.}.

\subsubsection{Separation of Order \texorpdfstring{$\sqrt{\log k}$}{sqrt(log k)}} \label{sec:separation-k}

In this case, we have a separation of 
\[\forall i \ne j \in [k], ~ \norm{\mu_i - \mu_j}_2 \ge c(\sigma_i+\sigma_j) \sqrt{\log(\rhos/\wmin)}  .\]

Let $\eji$ be the unit vector along $z_i - z_j$. We will now show that every point $y \in S_j$, is far from $z_i$ along the direction $\eji$.
\begin{align*}
\Iprod{z_i - z_j, \eji} &= \norm{z_j - z_i}_2 \ge \norm{\mu_i - \mu_j}_2 - \sigma_i - \sigma_j \ge 5 \sqrt{\log(\rhos/\wmin)} (\sigma_i + \sigma_j) + \tfrac{1}{2} \norm{\mu_i - \mu_j}_2
\end{align*}
Since $y \in S_j$, $\abs{\iprod{y - z_j, \eji}} \le 4 \sigma_j\sqrt{\log(\rhos/\wmin)}$, and
\begin{align*}
\Iprod{z_i - y, \eji} &\ge \Iprod{z_i - z_j, \eji} - \abs{\Iprod{y - z_j, \eji}} \ge 5 \sigma_i \sqrt{\log(\rhos/\wmin)} +\tfrac{1}{2} \norm{\mu_i - \mu_j}_2\\
\Iprod{\mu_i - y, \eji} & \ge \Iprod{z_i - y, \eji} - \norm{\mu_i - z_i}_2 \ge 4 \sigma_i \sqrt{\log(\rhos/\wmin)} +\tfrac{1}{2} \norm{\mu_i - \mu_j}_2
\end{align*}

We will now use Lemma~\ref{lem:box:helper} for each component $i \in [k]$ with mean zero Gaussian $y-\mu_i$ and $t^2=16 \log(\rhos/\wmin)+\tfrac{1}{8 \pi \sigma_i^2}\norm{\mu_j - \mu_i}^2_2$. We first prove \eqref{eq:box:other0}; using equation~\eqref{eq:box:helper0} with Gaussian $(y-\mu_i)$ of variance $\sigma_i^2/(2\pi)$,
\begin{align*}
w_i \int_{S_j} g_{\mu_i, \sigma_i}(y) \dy & \le w_i \exp(-t^2) < w_i \cdot \frac{\wmin^2}{16\pi } \exp\parens[\Big]{-\frac{\norm{\mu_i - \mu_j}_2^2}{16\sigma_i^2}}< \frac{\wmin^2}{16\pi }.\\
\text{ Hence, } \sum_{i\in [k], i \ne j} w_i \int_{S_j} g_{\mu_i, \sigma_i}(y) \dy &<  \frac{w_j}{16 \pi d}.
\end{align*}

For \eqref{eq:box:other1}, we use \eqref{eq:box:helper1} similarly with $(y-\mu_i)$ and variance $\sigma_i^2/(2\pi)$ and $t$ as above:
\begin{align*}
w_i \int_{S_j} \frac{\Abs{y(r_1) - \mu_i(r_1)}}{\sigma_i} g_{\mu_i, \sigma_i}(y) \dy & \le w_i \cdot \frac{\wmin^2}{16\pi \rhos} \exp\parens[\Big]{-\frac{\norm{\mu_i - \mu_j}_2^2}{16\sigma_i^2}}.\\
\sum_{i\in [k], i \ne j} \frac{w_i}{\sigma_i} \cdot\frac{\norm{\mu_i-\mu_j}_2}{\sigma_i} \int_{S_j} \Abs{y(r_1) - \mu_i(r_1)} g_{\mu_i, \sigma_i}(y) \dy &<  \frac{w_j}{16 \pi d \rhos}.
\end{align*}

The proof of \eqref{eq:box:other2} follows in an identical fashion from \eqref{eq:box:helper2} for $r_1, r_2 \in [d]$
\begin{align*}
\frac{w_i}{\sigma_i \sigma_j} \int_{S_j} \Abs{y(r_1) - \mu_i(r_1)} \Abs{y(r_2) - \mu_i(r_2)} g_{\mu_i, \sigma_i}(y) \dy 
& \le \frac{w_i \sigma_i}{\sigma_j} \cdot \frac{\wmin^2}{16\pi \rhos^2}\le \frac{w_j w_i}{16 \pi d \rhos}\\
\sum_{i\in [k], i \ne j} \frac{w_i}{\sigma_i \sigma_j} \int_{S_j} \Abs{y(r_1) - \mu_i(r_1)} \Abs{y(r_2) - \mu_i(r_2)} g_{\mu_i, \sigma_i}(y) \dy &< \frac{w_j}{16 \pi d \rhos}. 
\end{align*}

\subsubsection{Separation of Order \texorpdfstring{$\sqrt{d}$}{sqrt(d)}} \label{sec:separation-d}

The proof of Lemma~\ref{lem:box:othercomponentstotal} uses the following useful lemma that shows that for any point within $\sigma_j \sqrt{d}$ distance of $\mu_j$, the total probability mass from the other Gaussian components is negligible. Note that in the following lemma, $k$ can be arbitrary large in terms of $d$.

\begin{lemma} \label{lem:smallleakage}
There exists a universal constant $C_0>1$ such that for any $C \ge C_0$, if 
$\calG$ is a mixture of $k$ spherical Gaussians $\set{(w_j, \mu_j,\sigma_j): j \in [k]}$ in $\R^d$ satisfying
\begin{equation}\label{eq:leakage:condition}
\forall i \ne j \in [k], ~ \norm{\mu_i - \mu_j}_2 \ge C (\sigma_i+\sigma_j) \Paren{\sqrt{d}+\sqrt{\log(\rho_w \rhos)} },
\end{equation}
then for every component $j \in [k]$ and $ \forall x^*: \norm{x^*-\mu_j}\le  4\sigma_j \sqrt{d}$ and $\forall 0 \le m \le 2$,
\begin{align}
\sum_{i \in [k], i \ne j} w_i g_{\mu_i, \sigma_i} (x^*) &< \exp(-C^2 d/8) \cdot w_j g_j(x). \label{eq:leak1}\\
\sum_{i \in [k], i \ne j} w_i \left(\frac{\norm{x-\mu_i}^m}{\sigma_i^m}\right) \cdot g_{\mu_i, \sigma_i}(x^*)  &\le  \frac{\exp(-C^2 d/8)}{\rhos^m} \cdot w_j g_j(x^*), \label{eq:leak2}
\end{align}
where $\rhos = \max_i \sigma_i / (\min_i \sigma_i)$ and $\rho_w=\max_i w_i / (\min_i w_i)$.  
\end{lemma}

We now prove Lemma~\ref{lem:box:othercomponentstotal} under separation $\norm{\mu_i - \mu_j}_2 \ge c(\sigma_i+\sigma_j)(\sqrt{d}+ \sqrt{\log(\rhos \rhow)})$ assuming the above lemma. 
\begin{proof}[Proof for Lemma~\ref{lem:box:othercomponentstotal}]
Our proof will follow by applying Lemma~\ref{lem:smallleakage}, and integrating over all $y \in S_j$. 
For any point $y \in S_j$, $\norm{y - \mu_j}_2 \le 3 \sigma_j (\sqrt{d}+ \sqrt{\log(\rhow \rhos)})$, and the separation of the means satisfies \eqref{eq:leakage:condition}. To prove \eqref{eq:box:sc:0} we get from \eqref{eq:leak1}, 
\begin{align*}
\sum_{i\in [k], i \ne j} w_i g_{\mu_i, \sigma_i}(y)  &\le w_j \exp(-C^2 d/8)\cdot g_{\mu_j, \sigma_j}(y).\\
\text{Integrating over } S_j, ~\sum_{i\in [k], i \ne j} w_i \int_{S_j} g_{\mu_i, \sigma_i}(y) \dy &\le w_j \exp(-C^2 d/8) \int_{\R^d} g_{\mu_j, \sigma_j}(y) \dy \le \frac{w_j}{16 \pi d}.
\end{align*}
To prove \eqref{eq:box:sc:2}, we get from \eqref{eq:leak2} that for each $y \in S_j$ and $r_1, r_2 \in [d]$,
\begin{align*}
\sum_{i\in [k]\setminus \set{j}} \frac{w_i}{\sigma_i \sigma_j} \Abs{y(r_1) - \mu_i(r_1)} \Abs{y(r_2) - \mu_i(r_2)} g_{\mu_i, \sigma_i}(y)  &\le 
\sum_{i \in [k]\setminus \set{j}}\frac{w_i \rhos}{w_j} \frac{\norm{y - \mu_i}_2^2}{\sigma_i^2} g_{\mu_i, \sigma_i}(y) \\
&\le  \frac{w_j \exp(-C^2 d/8)}{\rhos}\cdot g_{\mu_j, \sigma_j}(y). 
\end{align*}
Integrating over all $y \in S_j$,
$$\sum_{i\in [k], i \ne j} \frac{w_i}{\sigma_i \sigma_j} \int_{S_j} \Abs{y(r_1) - \mu_i(r_1)} \Abs{y(r_2) - \mu_i(r_2)} g_{\mu_i, \sigma_i}(y) \dy \le \frac{w_j}{\rhos} \exp(-C^2 d/8) \cdot \int_{S_j} g_{\mu_j, \sigma_j}(y) \le \frac{w_j}{16 \pi d \rhos}$$
%

To prove \eqref{eq:box:sc:1}, we first note that for each $y \in S_j$, $\norm{\mu_i - \mu_j}_2 \le \norm{\mu_i - y}+(\sqrt{d}+ \sqrt{\log \rho_w})\sigma_j \le 2 \norm{\mu_i - y}$. Hence, by applying \eqref{eq:leak2} in Lemma~\ref{lem:smallleakage}, the following inequality holds:
\begin{equation}
\sum_{i \in [k], i \ne j} \frac{w_i \norm{\mu_i - \mu_j}_2 \norm{y - \mu_i}_2}{\sigma_i \sigma_j} g_{\mu_i, \sigma_i}(y)  \le  2\sum_{i \in [k], i \ne j} w_i \parens[\Big]{\frac{\norm{y-\mu_i}}{\sigma_i}}^2 \cdot g_{\mu_i, \sigma_i}(y) \le 2\exp(-C^2 d/8)\rhos^{-2} \cdot w_j g_j(y). \label{eq:leak3}
\end{equation}
Hence \eqref{eq:box:sc:1} follows from a similar argument as before.
\end{proof}

The proof of Lemma~\ref{lem:smallleakage} proceeds by using a packing-style argument.  Roughly speaking, in the uniform case when all the variances are roughly equal, the separation condition can be used to establish an upper bound on the number of other Gaussian means that are at a distance of $r$ from a certain mean. We now present the proof for the case when the variances are roughly equal, and this will also be important for the general case.

\begin{lemma} \label{lem:leakage:uniform}
In the notation of Lemma~\ref{lem:smallleakage}, let us denote for convenience, the $j$th component by $(w_0, \mu_0, \sigma_0)$.
Then, the total probability mass at any point  $x^* $ s.t. $\norm{x^*-\mu_0}\le 4\sqrt{d} \sigma_0$, from components $(w_1, \mu_1, \sigma_1), \dots, (w_i, \mu_i, \sigma_i), \dots, (w_{k'},\mu_{k'},\sigma_{k'})$
satisfying $ \forall i \in [k'], ~ \sigma_0 \le \sigma_i \le 2 \sigma_0$,  and $\norm{\mu_i - \mu_0}_2 \ge C \Paren{ \sqrt{d}+\sqrt{\log \rho_w} }$ is given by
\begin{equation}
\sum_{i \in [k']} w_i g_i(x^*) < w_0 \sigma_0^{-d} \exp(-C^2 d/4). \label{eq:leak}
\end{equation}
Furthermore, we have $\sum_{i \in [k']} w_i g_i^{1/2} (x^*) < w_0 \sigma_0^{-d} \exp(-C^2 d/4).$
\end{lemma}
\begin{proof}
  We can assume without loss of generality that $\mu_0=0, \sigma_0=1$.
  Suppose $r_i = \norm{\mu_i - x}_2$. Considering $1 \le \sigma_i \le 2$ and $g_i(x)=\sigma_i^{-d} \exp(-\pi \norm{x- \mu_i}_2^2 / \sigma_i^2)$,
  the lemma will follow, if we prove the following inequality for any $C \ge C_0 \ge 16$:

  \begin{equation}
    \sum_{i \in [k']} w_i \exp\left(-\tfrac{\pi}{2} \cdot  r_i^2 \right) < w_0 \exp(-C^2 d/4), 
  \end{equation}

  Let $B_i=\set{y: \norm{y-\mu_i}_2 \le \sqrt{d}}$. From the separation conditions, we have that
  \begin{itemize}
  \item the balls $B_1,\ldots,B_{k'}$ are pairwise disjoint, and
  \item the balls are far from the origin i.e, $ \forall y \in B_i, \norm{y} \ge C (\sqrt{d}+\sqrt{\log \rho_w})$.
  \end{itemize}

  We first relate the p.d.f. value $g_i(x)$ to the Gaussian measure of a ball of radius $\sqrt{d}$ around $\mu_i$. The volume of the ball $\Vol(B_i) \ge 1$ and
  for every $y \in B_i$, the separation conditions and triangle inequality imply that $\norm{y} \ge \norm{\mu_i-x^*} - \norm{y - \mu_i} - \norm{x^*}  \ge r_i / \sqrt{2}$.
  Hence,
  \begin{equation}\label{eq:relationtogaussianmeasure}  \exp\left(-\frac{\pi r_i^2}{2} \right) \le \int_{y \in B_i} \exp\left(-\pi \norm{y}^2 \right) \dy.
  \end{equation}
  By using the disjointness of balls $B_i$, we get that
  \begin{align*}
  \sum_{i \in [k']} w_i \exp\left(-\frac{\pi r_i^2 }{2} \right) &\le
  \sum_{i \in [k']} w_i \int_{y \in B_i} \exp\left(-\pi \norm{y}^2 \right) \dy  \le w_{\max} \sum_{i \in [k']} \int_{y \in B_i} \exp\left(-\pi \norm{y}^2 \right) \dy\\
   &\le w_{\max} \int_{\substack{y: \norm{y} \ge C \parens{\sqrt{d} + \sqrt{\log \rho_w }}}} \exp\left( -\pi\norm{y}^2 \right) \dy \\
   &\le \rhow w_0 \cdot \frac{1}{\rhow}  \exp(-C^2 d/4) \le w_0 \exp(-C^2 d/4), 
   \end{align*}
where the last line follows since a Gaussian random variable in $d$ dimensions with mean $0$ and unit variance in each direction has measure at most $\exp(-s^2/ 2)$ outside a ball of radius $(\sqrt{d}+ s)$ (see Lemma~\ref{lem:lengthconc}).   
\end{proof}

We now proceed to the proof of Lemma~\ref{lem:smallleakage}.
\begin{proof}[Proof of Lemma~\ref{lem:smallleakage}]
We may again assume without loss of generality that $\mu_j =0$ and $\sigma_j=1$ (by shifting the origin and scaling). Hence $\norm{x^*}_2 \le \sqrt{d}/\sqrt{\pi}$.
We will divide the components $i \in [k] \setminus \set{j}$ depending on their standard deviation $\sigma_i$ into buckets $I_0, I_1, \dots, I_{s}$ where $s \le \lceil \log(\max_i \sigma_i / \sigma_j) \rceil$ as follows:
$$ I_0= \set{i \in [k] \setminus \set{j}: \sigma_i \le \sigma_j },~\quad \forall q \in [s] ~~ I_q = \set{i \in [k]: 2^{q-1} \sigma_j < \sigma_i \le 2^{q} \sigma_j}.$$

Let us first consider the components in the bucket $I_0$, and suppose we scale so that $\sigma_j=1$. As before let $r_i=\norm{x^*-\mu_i}_2$. We first note that if $\sigma_i \le \sigma_j$, since $\norm{x^*-\mu_i}_2 \ge C \sqrt{d}(\sigma_i +\sigma_j)$, a simple calculation shows that $g_{\mu_i, \sigma_i} (x^*) \le g_{\mu_i, \sigma_j}(x^*) $. Hence, by applying Lemma~\ref{lem:leakage:uniform} to $I_0$ with a uniform variance of $\sigma_j^2 =1$ for all Gaussians, we see that
\begin{equation} \sum_{i \in I_0} w_i g_{\mu_i, \sigma_i} (x^*) \le \sum_{i \in I_0} w_i g_{\mu_i, \sigma_j} (x^*)  \le w_j \exp(-C^2 d/4). \label{eq:bucket0} \end{equation}
Consider any bucket $I_q$. We will scale the points down by $2^{q-1} \sigma_j$ so that $y' = y/ (2^{q -1} \sigma_j)$, and let $\forall i \in I_q$ $r'_i=r_i/ (2^{q-1} \sigma_j)$.
Again from Lemma~\ref{lem:leakage:uniform} we have that

\begin{align}
\sum_{i \in I_q} w_i \exp(-\pi r_i^2/ \sigma_i^2) \le \sum_{i \in I_q} w_i \exp(- \pi {r'_i}^2/ 4) &\le w_j \exp(-C^2 d/4).\\
\text{Hence, } ~\sum_{i \in I_q} w_i g_i (x) \le \frac{1}{2^{(q-1)d} \sigma^d_j} \sum_{i \in I_q} w_i \exp(-\pi r_i^2 / \sigma_i^2) &\le w_j 2^{-(q-1)d} \exp(-C^2 d/4). \label{eq:bucketq}
\end{align}

\noindent Hence, by summing up over all buckets (using \eqref{eq:bucket0} and \eqref{eq:bucketq}),
the total contribution $\sum_{i \in [k]\setminus \set{j}} w_i g_i(x) \le 4w_j \exp(-C^d/4) \le \exp(-C^2 d/8) \cdot w_j g_j(x)$.

The final equation~\eqref{eq:leak2} follows from separation conditions, since $\norm{x-\mu_i} > C \sqrt{d} \sigma_i$, hence $\exp\parens{\pi \norm{x-\mu_i}^2/2 \sigma_i^2}> \norm{x - \mu_i}^m / \sigma_i^m$ for some sufficiently large constant $C=C(m) \ge 1$. Hence, by using an identical argument with the furthermore part of Lemma~\ref{lem:leakage:uniform}, it follows. 

\end{proof}

\subsection{Sampling Errors}\label{sec:samplingerrors}

\begin{lemma}[Error estimates for Gaussians]\label{lem:app:box:samplingerror}
Let $S \subset \set{x \in \R^d: \norm{x}_2 \le \rho}$ be any region, and suppose samples $x^{(1)}, x^{(2)}, \dots, x^{(N)}$ are generated from a mixture of $k$ spherical Gaussians with parameters $\{(w_i, \mu_i, \sigma_i)\}_{i \in [k]}$ in $d$ dimensions, with $\forall i, ~ \norm{\mu_i}, \sigma_i \le \rho$. There exists a constant $C>0$ such that for any $\eps>0$, with $N\ge C \left(\eps^{-2} \rho^2 \log d \log(1/\gamma) \right)$ samples, we have for all $r, r' \in [d]$ with probability at least $1-\gamma$
\begin{align}\label{eq:sample:errorbound1}
\Abs{\int_{x \in S} \sum_{i \in [k]} (x(r)-\mu_i(r)) g_{\mu_i, \sigma_i}(x) \dx -\frac{1}{N}\sum_{\ell \in [N]} \Paren{x^{(\ell)}(r)-\mu_i(r)} \I[x^{(\ell)} \in S] } &< \eps.\\
\Big\lvert \int_{x \in S} \sum_{i \in [k]} (x(r) - \mu_i(r)) (x(r')-\mu_i(r'))^T g_{\mu_i, \sigma_i}(x) \dx ~\qquad~~& \nonumber\\
\qquad-\frac{1}{N}{\sum_{\ell \in [N]} \Paren{x^{(\ell)}(r)-\mu_i(r)}\Paren{x^{(\ell)}(r')-\mu_i(r')}^T \I[x^{(\ell)} \in S] } \Big\rvert &< \eps. \label{eq:sample:errorbound2}
\end{align}
\end{lemma}
\begin{proof}

Fix an element $r \in [d]$. Each term $\ell \in [N]$ in the sum corresponds to an i.i.d random variable $Z_\ell=\Paren{x^{(\ell)}(r)-\mu_i(r)} \I[x^{(\ell)} \in S]$. We are interested in the deviation of the sum $Z=\frac{1}{N}\sum_{\ell \in [N]} Z_\ell$.

Firstly, $\E[Z]= \int_{x \in S} \sum_{i \in [k]} (x(r)-\mu_i(r)) g_{\mu_i, \sigma_i}(x) \dx$. Further, each of the i.i.d r.v.s has value
$\abs{Z_\ell - \E Z_\ell} \le \abs{x^{(\ell)}(r) - \mu_i(r)} + \rho$. Hence, $\abs{Z_\ell} > (2\rho+t \max_i \sigma_i)$ with probability $O\left(\exp(-t^2/2)\right)$. Hence, by using standard sub-gaussian tail inequalities, we get
$$\Pr[\abs{Z-\E{Z}} > \eps] < \exp\left(-\frac{\eps^2 N}{(2 \rho+\max_i \sigma_i)^2}\right)$$
Hence, to union bound over all $d$ events $N=O\left(\eps^{-2} \rho^2 \log d \log(1/\gamma) \right)$ suffices.

A similar proof also works for the second equation \eqref{eq:sample:errorbound2}.
\end{proof}

%
%
%
%

\begin{proof}[Proof of Lemma~\ref{lem:box:sampling1}]
Let $Z=\frac{1}{w_j \sigma_j N}\sum_{\ell \in [N]} \I[y^{(\ell)} \in S_j] \Paren{y^{(\ell)}-z_j}$. We see that
$$\E[Z]= \frac{1}{w_j \sigma_j} \sum_{i=1}^k w_i \int_{y \in \R^d} \cI_j(y) (y - z_j) g_{\sigma_i \vx_i, \sigma_i}(y) \dy =F_j(\vx).$$
Further, we can write $Z=Z_1 + Z_2$, where
\begin{align*}
Z_1&=\frac{1}{w_j \sigma_j N}\sum_{\ell \in [N]} \I[y^{(\ell)} \in S_j] \Paren{y^{(\ell)}-\sigma_i \vx_i} ,\\
Z_2&=\frac{1}{w_j \sigma_j N} \Paren{\sigma_i \vx_i -z_j} \sum_{\ell \in [N]} \I[y^{(\ell)} \in S_j].
\end{align*}

By applying Lemma~\ref{lem:app:box:samplingerror}, we get
$$\norm{Z- \E[Z]}_{\infty} \le \norm{Z_1 - \E Z_1}_\infty + \norm{Z_2 - \E Z_2}_\infty \le \frac{\eta}{2} + \frac{\eta}{2} = \eta.$$
\end{proof}

\begin{proof} [Proof of Lemma~\ref{lem:box:sampling2}]
Let $Z=\frac{w_i}{w_j \sigma_i \sigma_j N}\sum_{\ell \in [N]} \I[y^{(\ell)} \in S_j] \Paren{y^{(\ell)}-z_j} \Paren{y^{(\ell)}-\sigma_i \vx_i}^T $, where the samples are drawn just from the spherical Gaussian with mean $\sigma_i \vx_i$ and variance $\sigma_i^2$. Hence,
$$\E[Z]= \frac{w_i}{w_j \sigma_j \sigma_i}  \int_{y \in \R^d} \cI_j(y) (y - z_j) (y- \sigma_i \vx_i)^T g_{\sigma_i \vx_i, \sigma_i}(y) \dy =F_j(\vx).$$
Also $Z  = Z_1+Z_2$   where
\begin{align*}
Z_1 &= \frac{w_i}{w_j \sigma_i \sigma_j N} \sum_{\ell \in [N]} \I[y^{(\ell)} \in S_j] \Paren{\sigma_i \vx_i-z_j} \Paren{y^{(\ell)}-\sigma_i \vx_i}^T \text{ and } \\
Z_2&=\frac{w_i}{w_j \sigma_i \sigma_j N} \sum_{\ell \in [N]} \I[y^{(\ell)} \in S_j]  \Paren{y^{(\ell)}-\sigma_i \vx_i} \Paren{y^{(\ell)}-\sigma_i \vx_i}^T.
\end{align*}
The $\norm{Z-\E Z}_{\infty \to \infty}$ is the sum of the absolute values of the $d$ entries in a row. 
From Lemma~\ref{lem:app:box:samplingerror},
$$ \norm{Z - \E Z}_{\infty \to \infty}= \norm{Z_1 - \E Z_1}_{\infty \to \infty}+ \norm{Z_2 - \E Z_2}_{\infty \to \infty} \le \frac{\eta}{2dk}\cdot d + \frac{\eta}{2dk} \cdot d= \eta/k. $$
Finally, we have from upper bound of the error in the individual blocks that
$$\bignorm{F'(\vx) - \widetilde{F'}(\vx)}_{\infty \to \infty} \le \max_{j \in [k]} \sum_{i \in [k]} \bignorm{ \grad_{\vx_i} F_j(\vx) - \widetilde{\grad_{\vx_i}F_j}(\vx) }_{\infty \to \infty} \le \eta .$$
\end{proof}


\section{Sample Complexity Upper Bounds with \texorpdfstring{$\Omega(\sqrt{\log k})$}{Omega(sqrt(log k))} Separation}
\label{sec:oldub}

We now show that a mean separation of order $\Omega(\sqrt{\log k})$ suffices to learn the model parameters up to arbitrary accuracy $\delta>0$, with $\poly(d,k,\log(1/\delta))$ samples. In all the bounds that follow, the interesting settings of parameters are when $\rho, 1/\wmin \le \poly(k)$.
We note that these upper bounds are interesting even in the case of uniform mixtures: $w_i, \sigma_i$ being equal across the components.

For sake of exposition, we will restrict our attention to the case when the standard deviations $\sigma_i$, and weights $w_i$ are known for all $i \in [k]$. We believe that similar techniques can also see used to handle unknown $\sigma_i, w_i$ as well (see Remark~\ref{remark:unknown}). 
In what follows $\rhos$ corresponds to the aspect ratio of the covariances i.e., $\rhos=\max_{i \in [k]} \sigma_i / \min_{i \in [k]} \sigma_i$. 

\begin{theorem}[Same as Theorem~\ref{thm:informal:betterupperbounds}]\label{thm:betterupperbounds}
There exists a universal constant $c>0$ such that suppose we are given samples from a mixture of spherical Gaussians $\calG=\set{\parens{w_i, \mu_i , \sigma_i}: i \in [k]}$ (with known weights and variances) that are $\rho$-bounded and the means are well-separated i.e.
\begin{equation}\label{eq:box:wellsep}
\forall i, j \in [k], i \ne j:~ \norm{\mu_i - \mu_j}_2 \ge  c \sqrt{\log (\rhos/\wmin)} (\sigma_i + \sigma_j),
\end{equation}
there is an algorithm that for any $\delta>0$, uses $\poly(k,d,\rho,\log(1/\wmin), 1/\delta)$ samples and recovers with high probability the means up to  $\delta$ relative error i.e., finds $\set{\mu'_i: i \in [k] }$ such that $\dparam\Big(\calG, \set{(w_i, \mu'_i, \sigma_i): i \in [k]} \Big)\le \delta$.
\end{theorem}

Such results are commonly referred to as {\em polynomial identifiability} or {\em robust identifiability} results. 
We can again assume as in Section~\ref{sec:amplify} that without loss of generality that $d \le k$ due to the following dimension-reduction technique using PCA~\cite{VW04}. Theorem~\ref{thm:betterupperbounds} follows in a straightforward manner by combining the iterative algorithm, with initializers given by the following theorem. 

\begin{theorem}[Initializers Using Polynomial Samples] \label{thm:upperbounds}
For any constant $c\ge 10$, suppose we are given samples from a mixture of spherical Gaussians $\calG=\set{\parens{w_i, \mu_i , \sigma_i}: i \in [k]}$ that are $\rho$-bounded and the means are well-separated i.e.
\begin{equation} \label{eq:well-separated}
\forall i, j \in [k], i \ne j:~ \norm{\mu_i - \mu_j}_2 \ge 4c \sqrt{\log (\rhos/\wmin)} (\sigma_i + \sigma_j).
\end{equation}
There is an algorithm that uses $\poly(k^{c},d,\rho)$ samples and with high probability learns the parameters of $\calG$ up to $k^{-c}$ accuracy, i.e., finds another mixture of spherical Gaussians $\tcalG$ that has parameter distance $\dparam(\calG, \tcalG)\le k^{-c}$.   
\end{theorem}

The key difference between Theorem~\ref{thm:betterupperbounds} and Theorem~\ref{thm:upperbounds} is that 
 in the former, the parameter estimation accuracy is {\em independent} of 
the separation (the constant $c>0$ does not depend on $\delta$). 
If $\rho_s=k^{O(1)}$ and $\wmin \ge k^{-O(1)}$, then we need means $\mu_j, \mu_\ell$ 
to be separated by $\Omega(\sqrt{\log k}) (\sigma_j+\sigma_\ell)$ to get reasonable 
estimates of the parameters with $\poly(k)$ samples. While the algorithm is sample 
efficient, it takes time that is $(\rho/\wmin)^{O(ck^2)}$ time since it runs over 
all possible settings of the $O(k)$ parameters in $d \le k$ dimensions. Note that 
the above theorem holds even with unknown variances and weights that are unequal. 

We can use Theorem~\ref{thm:upperbounds} as a black-box to obtain a mixture $\calG^{(0)}$ such that $\dparam(\calG,\calG^{(0)}) \le k^{-c}$ whose parameters will serve as initializers for the iterative algorithm in Section~\ref{sec:amplify}.
Theorem~\ref{thm:upperbounds} follows by exhibiting a lower bound on the statistical distance between any two sufficiently separated spherical Gaussian mixtures which have non-negligible parameter distance. 

\begin{proposition}\label{prop:upperbounds}
Consider any two spherical Gaussian mixtures $\calG, \calG^*$ in $\R^d$ as in Theorem~\ref{thm:upperbounds}, with their corresponding p.d.fs being $f$ and $f^*$ respectively. There is a universal constant $c'>0$, such that for any $c \ge 5$, if the parameter distance $\dparam(\calG, \calG^*) \ge \frac{1}{k^{c-1}}$, and $\calG$ is well-separated:
$$ \forall i, j \in [k], i \ne j:~ \norm{\mu_i - \mu_j}_2 \ge 4c \sqrt{\log (\rho_s/\wmin)} (\sigma_i + \sigma_j),$$
and both mixtures have minimum weight $\wmin \ge 1/k^{c}$, then
\begin{equation}
\norm{f-f^*}_1 = \int_{\R^d} \Abs{f(x) - f^*(x)} \dx \ge \frac{c'}{dk^{2c} \rho_s^2},
\end{equation}
where $\rho_s = \max\Big\{\frac{\max_{j \in [k]} \sigma^*_j}{\min_{j \in [k]} \sigma_j} , \frac{\max_{j \in [k]} \sigma_j}{\min_{j \in [k]} \sigma^*_j}\Big\}$.
\end{proposition}

In the above proposition, $\rho_s \le \rho^2$ is a simple upper bound since we will only search for all parameters of magnitude at most $\rho$. But it can be much smaller if we have a better knowledge of the range of $\set{\sigma_j: j \in [k]}$. We note that in very recent independent work, Diakonikolas et al. established a similar statement about mixtures of Gaussians where the components have small overlap (see Appendix B in ~\cite{DKS16}). We first see how the above proposition implies Theorem~\ref{thm:upperbounds}.

%
%

\subsection{Proposition~\ref{prop:upperbounds} to Theorem~\ref{thm:upperbounds}}

The following simple lemma gives a sample-efficient algorithm to find a distribution from a net of distributions $\calT$ that is close to the given distribution. This tournament-based argument is a commonly used technique in learning distributions.

\begin{lemma}\label{lem:samplealg}
Suppose $\calT$ is a set of probability distributions over $\calX$, and we are given $m$ samples from a distribution $D$, which is $\delta$ close to some distribution $D'' \in \calT$ in statistical distance, i.e., $\norm{D-D''}_{TV} \le \delta$. Then there is an algorithm that uses $m=O(\delta^{-2} \log\abs{\calT})$ samples from $D$ and with probability at least $1-1/\abs{\calT}$ finds a distribution $D^*$ such that $\norm{D-D^*}_{TV}\le 4 \delta$.
\end{lemma}

\begin{proof}
Let $\calT=\set{D_1, D_2, \dots, D_{|\calT|}}$ be the set of distributions over $\calX$.
For any $i \neq j$ let $A_{ij} \subset \calX$ be such that
\begin{equation}
\label{eq:defofaij}
   \Pr_{x \leftarrow D_i} [ x \in A_{ij}] - \Pr_{x \leftarrow D_j} [ x \in A_{ij}] = \norm{D_i - D_j}_{TV} \; .
\end{equation}

The algorithm is as follows.
\begin{itemize}
\item Use $m=O( \delta^{-2} \log \abs{\calT})$ samples to obtain estimates $p_{ij}$ satisfying
with probability at least $1-1/\abs{\calT}$ that
\begin{equation}
\label{eq:estimatespij}
\forall i,j,~ \abs{p_{ij}-\Pr_{x \leftarrow D} [ x \in A_{ij}]} \le \delta/2 \; .
\end{equation}

\item Output the first distribution $D_i \in \calT$ that satisfies
\begin{equation}
\label{eq:testofdi}
 \forall j \in \{1,\ldots,|\calT|\},~ \len[\big]{p_{ij} - \Pr_{x \leftarrow D_i} [ x \in A_{ij}]} \le 3\delta/2 \; .
\end{equation}
\end{itemize}

First, notice that by~\eqref{eq:estimatespij} and the assumption that
$\norm{D-D''}_{TV} \le \delta$, $D''$ satisfies the test
in~\eqref{eq:testofdi}.
We next observe by the definition of $A_{ij}$ in~\eqref{eq:defofaij}
that for any $i \neq j$ such that
both $D_i$ and $D_j$ pass the test, we must have $\norm{D_i-D_j}_{TV} \le 3\delta$.
This implies that the output of the algorithm must be within statistical
distance $4\delta$ of $D$, as desired.
\end{proof}

We now prove Proposition~\ref{thm:upperbounds} assuming the above Proposition~\ref{prop:upperbounds}. This follows in a straightforward manner by using the algorithm from Lemma~\ref{lem:samplealg} where $\calT$ is chosen to be a net over all possible configuration of means, variances and weights. We give the proof below for completeness.

\begin{proof}[Proof of Proposition~\ref{thm:upperbounds}]

Set $\delta := c'/(8 \rho^8 d k^{c+1})$, where $c'$ is given in Proposition~\ref{prop:upperbounds}, and $\eps=\delta/(6k d\rho)$.
We first pick $\calT$ to be distributions given by an $\eps$-net over the parameter space, and the algorithm will only output one of the mixtures in the net. Each Gaussian mixture has a standard deviation $\sigma_j \in \R$ and $k$ components each of which has a mean in $\R^d$, and weight in $[0,1]$. Further, all the means, variances are $\rho$ bounded. Hence, considering a net $\calT$ of size
$M \le \left(\rho/(\wmin \eps)\right)^{(d+2)k}$ of $\rho$-bounded well-separated mixtures with weights at least $\wmin$, we can ensure that every Gaussian mixture is $\eps$ close to some $\rho$ bounded mixture in the net in parameter distance $\dparam$. This corresponds to a total variation distance distance of at most $\delta = 6 k \sqrt{d}\rho \eps$  as well, using Lemma~\ref{app:kl}. Hence, there is a mixture of well-separated $\rho$-bounded spherical Gaussian mixture $\calG'' \in \calT$ such that $\norm{\calG'' - \calG}_{TV} \le 6k \sqrt{d}\rho \eps= \delta$.

By using Lemma~\ref{lem:samplealg}, we can use $m=O(\delta^{-2} \log\abs{\calT})=O(k^{2c+3} (d+2)^3 \rho^{16} \log(\rho k d/ \wmin))$ and with high probability, find some well-separated $\rho$-bounded mixture of spherical Gaussians $\calG^*$ such that $\norm{\calG- \calG^*}_{TV} \le 4\delta$. Proposition~\ref{prop:upperbounds} implies that the parameters of $\calG^*$ satisfy $\dparam(\calG, \calG^*) \le k^{-c}$.

\end{proof}

\subsection{Proof of Proposition~\ref{prop:upperbounds}}

To show Proposition~\ref{prop:upperbounds}, we will consider any two mixtures of well-separated Gaussians, and show that the statistical distance is at least inverse polynomial in $k$. This argument becomes particularly tricky when the different components can have different values of $\sigma_i$ e.g., instances where one component of $\calG$ with large $\sigma_i$ is covered by multiple components from $\calG^*$ with small $\sigma^*_j$ values.

For each component $j \in [k]$ in $\calG$, we define the region $S_j$ around $\mu_j$, where we hope to show a statistical distance.

\begin{definition}[Region $S_j$]\label{def:Sj}
In the notation above, for the component of $\calG$ centered around $\mu_j$, define $\ejl$ as the unit vector along $\mu^*_\ell - \mu_j$.
\begin{equation}
S_j = \set{x \in \R^d: \forall \ell \in [k] ~ \Abs{\iprod{x-\mu_j, \ejl}} \le 2c \sqrt{\log (\rhos/\wmin)} \sigma_j }.
\end{equation}

\end{definition}

The following lemma shows that most of the probability mass from $j$th component around $\mu_j$ is confined to $S_j$.
\begin{lemma} \label{lem:Sj}
\begin{equation}
\forall j \in [k], \int_{S_j} g_{\mu_j, \sigma}(x) \dx \ge 1- \left(\frac{\wmin}{\rhos}\right)^{4c}.
\end{equation}
\end{lemma}
\begin{proof}
The region $S_j$ is defined by $k$ equations, one for each of the components $\ell$ in $\calG^*$. When $x$ is drawn from the spherical Gaussian $g_{\mu_j, \sigma_j}$
$$\Pr_{x \leftarrow G_{\mu_j, \sigma_j}}\Big[\Abs{\iprod{x-\mu_j, \ejl}} > 2c \sigma_j \sqrt{\log (\rhos/ \wmin)} \Big] \le \Phi_{0,1}\left(2c \log(\rhos/\wmin)\right) \le \left(\frac{\wmin}{\rhos}\right)^{2c^2}\le \frac{1}{k} \left(\frac{\wmin}{\rhos}\right)^{4c},$$
where the last step follows since $c \ge 5$ and $\wmin \le 1/k$. Hence, performing a union bound over the $k$ components in $\calG^*$ completes the proof.
\end{proof}

The following lemma shows that components of $\calG^*$ that are far from $\mu_j$ do not contribute much $\ell_1$ mass to $S_j$.

\begin{lemma}\label{lem:othercomponents}
For a component $j \in [k]$ in $\calG$, 
and let $\ell \in [k]$ be a component of $\calG^*$ such that $\norm{\mu_j - \mu^*_\ell}_2 \ge 2c \sqrt{\log(\rhos/\wmin)} (\sigma_j+\sigma^*_\ell)$.
Then the total probability mass from the $\ell$th component of $\calG^*$ is negligible, i.e.,
\begin{equation}
 \int_{S_j} w^*_\ell g_{\mu^*_\ell, \sigma^*_\ell}(x) \dx < \ctails w^*_\ell \left(\frac{\wmin}{\rhos}\right)^{2c^2}.
\end{equation}
\end{lemma}
\begin{proof}
Let $\ejl$ be the unit vector along $\mu^*_\ell - \mu_j$. We will now show that every point $x \in S_j$, is far from $\mu^*_\ell$ along the direction $e_{jl}$:
\begin{align*}
x \in S_j  \implies &~~ \Abs{\Iprod{x - \mu_j, \ejl}} \le 2c \sqrt{\log(\rhos/\wmin)} \cdot \sigma_j\\
\norm{\mu_j - \mu^*_\ell}_2 \ge 2c \sqrt{\log(\rhos/\wmin)} (\sigma_j+\sigma^*_\ell) \implies &~~  \Iprod{\mu^*_\ell - \mu_j, \ejl} \ge 2c \sqrt{\log(\rhos/\wmin)} (\sigma_j + \sigma^*_\ell)\\
\text{Hence } \forall x \in S_j&~~ \Iprod{\mu^*_\ell - x, \ejl} \ge 2c \sqrt{\log(\rhos/\wmin)} \cdot \sigma^*_\ell
\end{align*}

\noindent Hence the $\ell_1$ contribution from $g_{\mu^*_\ell, \sigma^*}$ restricted to $S_j$ is bounded as
\begin{align*}
\int_{S_j} w^*_\ell g_{\mu^*_\ell, \sigma^*_\ell}(x) \dx &\le \int_{\substack{x: \Iprod{\mu^*_\ell - x, \ejl} \ge 2c \sqrt{\log(\rhos/\wmin)} \sigma^*_\ell}} w^*_\ell g_{\mu^*_\ell, \sigma^*_\ell}(x) ~\dx \\
&\le w^*_\ell \Phi_{0,1}\left(2c \log(\rhos/\wmin)\right) \le w^*_\ell \left(\frac{\wmin}{\rhos}\right)^{2c^2},
\end{align*}
as required.
\end{proof}

The following lemma shows that there is at most one component of $\calG^*$ that is close to the component $(w_j, \mu_j, \sigma_j^2 I)$ in $\calG$.

\begin{lemma} [Mapping between centers] \label{lem:centermapping}
Suppose we are given two spherical mixtures of Gaussians $\calG$, $\calG^*$ as in Proposition~\ref{prop:upperbounds}.
Then for every $j \in [k]$ there is at most one $\ell \in [k]$ such that
\begin{equation}\label{eq:centers}
\norm{\mu_j - \mu^*_{\ell}}_2 \le 4c \sqrt{\log(\rhos/\wmin)} \sigma^*_{\ell}.
\end{equation}
\end{lemma}
\begin{proof}
This follows from triangle inequality. 
Suppose for contradiction that there are two centers corresponding to indices $\ell_r$ with $r=1,2$, that satisfy \eqref{eq:centers}. By using triangle inequality this shows that $\norm{\mu^*_{\ell_1} - \mu^*_{\ell_2}} \le 4c \sqrt{\log(\rhos /\wmin)} (\sigma^*_{\ell_1}+\sigma^*_{\ell_2})$ which contradicts the assumption.
\end{proof}

We now proceed to the proof of the main proposition (Proposition~\ref{prop:upperbounds}) of this section.
We will try to match up components in $\calG, \calG^*$ that are very close to each other in parameter distance and remove them from their respective mixtures. Then we will consider among unmatched components the one with the {\em smallest variance}. Suppose $(w_j, \mu_j, \sigma_j)$ were this component, we will show a significant statistical distance in the region $S_j$ around $\mu_j$.

The following lemma considers two components, $G_j= (w_j, \mu_j, \sigma_j)$ from $\calG$ , and $G^*_j=(w^*_j, \mu^*_j, \sigma^*_j)$ from $\calG^*$ that have a non-negligible difference in parameters (we use the same index $j$ for convenience, since this is without loss of generality). This lemma shows that if $\sigma_j \le \sigma^*_j$, then there is some region $S \subset S_j$ where the component $G_j$ has significantly larger probability mass than $G^*_j$. We note that it is crucial for our purposes that we obtain non-negligible statistical distance in a region around $S_j$. Suppose $f_j$ and $f^*_j$ are the p.d.f.s of the two components (with weights), it is easier to lower bound $\norm{f_j-f^*_j}_1$ (e.g. Lemma 38 in \cite{MV10}). However, this does not translate to a corresponding lower bound restricted to region $S$ i.e., $\norm{f_j-f^*_j}_{1,S}$ since $f_j$ and $f^*_j$ do not represent distributions (e.g. $\norm{f_j}_1=w_j < 1$).

\newcommand{\csd}{c_1}

\begin{lemma} \label{lem:statdiff}
For some universal constant $\csd>0$, suppose we are given two spherical Gaussian components with parameters $(w_j,\mu_j, \sigma_j)$ and $(w^*_j,\mu^*_j, \sigma^*_j)$ that are $\rhos$-bounded satisfying
\begin{equation} \sigma_j \le \sigma^*_j \text{ and } \frac{\norm{\mu_{j}- \mu^*_{j}}_2}{\sigma_j} + \frac{\abs{\sigma_j - {\sigma^*_j}}}{\sigma_j} + \abs{w_j - w^*_j} \ge \gamma.  \label{eq:tvsep:cond}
\end{equation}
Then, there exists a set $S \subset S_j$ such that 
\begin{equation} \label{eq:tvseparation}
 \int_{S} \Abs{w_j g_{\mu_j, \sigma_j}(x)- w^*_j g_{\mu^*_j, \sigma^*_j}(x)} \dx > \frac{\csd \gamma^2}{d\rho_s^2}.
\end{equation}

\end{lemma}
Before we proceed, we present two lemmas which lower bound the statistical distance when the means of the components differ, or if the means are identical but the variances differ.

\begin{figure}
\centering

\begin{minipage}{0.4\textwidth}
\centering
    \includegraphics[width=1.2 \linewidth]{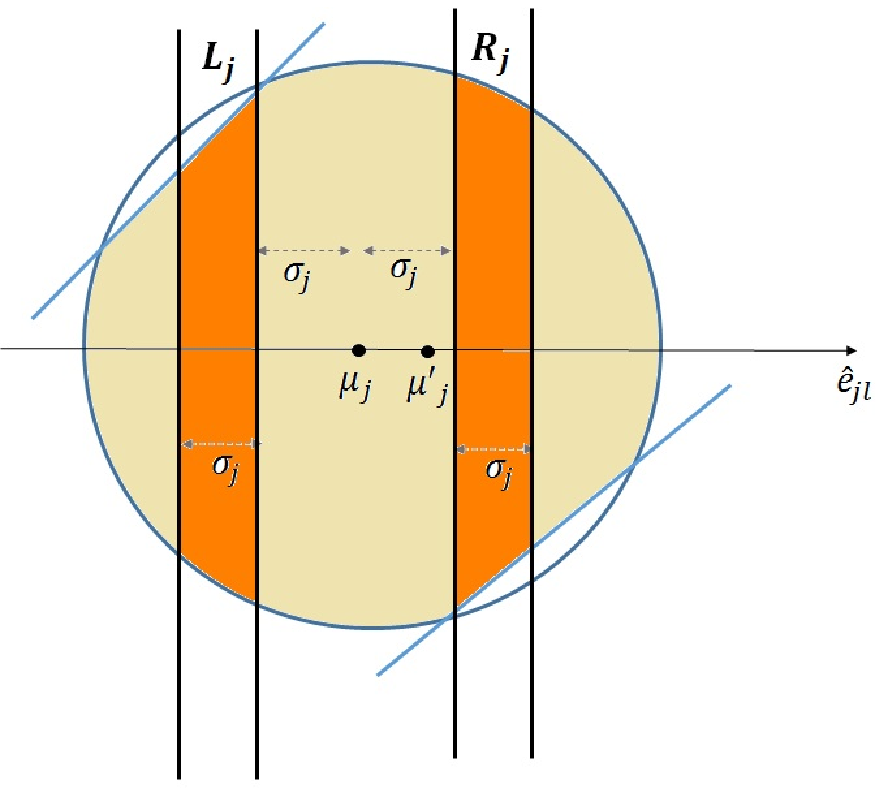}
  \end{minipage}%
 \begin{minipage}{0.5\textwidth}
 \centering
    \includegraphics[width= 0.8 \linewidth]{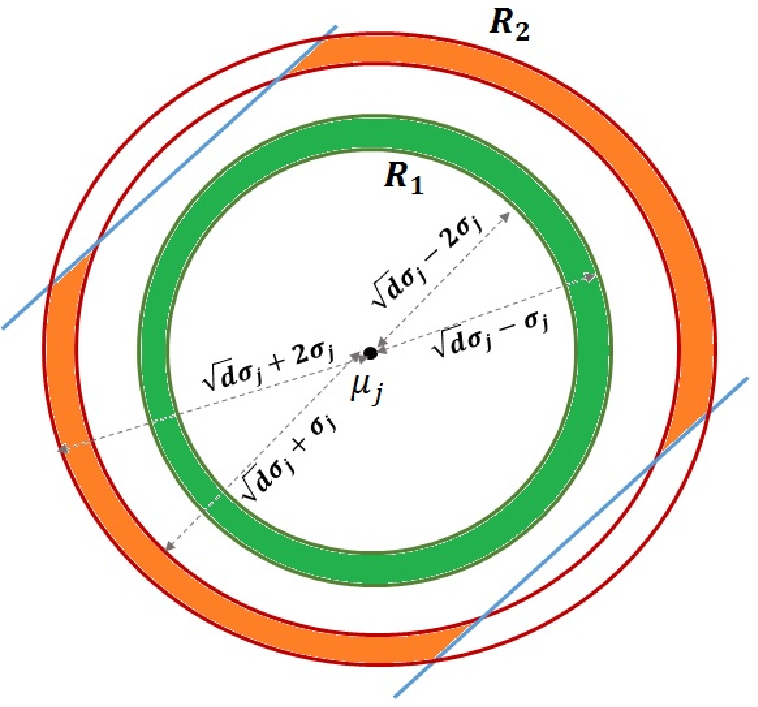}
  \end{minipage}

\caption{{\em Regions $S_j$ used in Lemma~\ref{lem:diff:means} and Lemma~\ref{lem:diff:variances}:} {\small Lemma~\ref{lem:diff:means} shows a statistical difference in one of the strips $L_j$ or $R_j$. Lemma~\ref{lem:diff:variances} shows a statistical difference in one of the annular strips $R_1$ or $R_2$}} \label{fig:regions}

\vspace{-10pt}
\end{figure}

%
%
%

\begin{lemma}[Statistical Distance from Separation of Means]\label{lem:diff:means}
In the notation of Lemma~\ref{lem:statdiff}, suppose $\norm{\mu_j - \mu^*_j}_2 \ge \gamma \sigma_j$.
Then, there exists a set $S \subset S_j$ such that
\begin{equation} \label{eq:diff:means}
 \int_{S} \Abs{w_j g_{\mu_j, \sigma_j}(x)- w^*_j g_{\mu^*_j, \sigma^*_j}(x)} \dx > \frac{\gamma \sigma^2_j}{8{(\sigma^*_j)}^2}> \frac{\gamma}{8 \rho_s^2}.
\end{equation}
\end{lemma}
\begin{proof}
Without loss of generality we can assume that $\mu_j=0$ (by shifting the origin to $\mu_j$).
Let us consider two regions
\[R_j = S_j \cap \set{x: \iprod{x, \ejj} \in [\tfrac{1}{\sqrt{2 \pi}}\sigma_j, \tfrac{2}{\sqrt{2 \pi}}\sigma_j]} ~\text{ and }~ L_j = S_j \cap \set{x: \iprod{x, \ejj} \in [-\tfrac{2}{\sqrt{2 \pi}}\sigma_j-\tfrac{1}{\sqrt{2 \pi}}\sigma_j]} .\]
Due to the symmetry of $S_j$ (Definition~\ref{def:Sj}), it is easy to see that
\begin{equation}\int_{L_j} f(x) \dx= \int_{R_j} f(x) \dx \ge \frac{1}{8}. \label{eq:equalmass}
\end{equation}
On the other hand, we will now show that
\begin{equation}
 \int_{R_j} f^*(x) \dx > \left( 1+ \frac{ 2\gamma \sigma^2_j}{{\sigma^*_j}^2} \right) \int_{L_j} f^*(x) \dx \label{eq:diffmassfromdiffmeans}
\end{equation}

For any point $x$, let $x_{jl}= \iprod{x, \ejj}$. Further, $\mu^*_j= \norm{\mu^*_j} \ejj$.
We first note that $x \in R_j \Longleftrightarrow (-x) \in L_j$ due to the symmetric definition of $S_j$. Hence,
\begin{align*}
\int_{L_j} f^*(x) \dx &= \int_{R_j} \frac{f^*(-x)}{f^*(x)} f^*(x) \dx = \int_{R_j} \exp\Paren{\frac{-\pi}{{\sigma^*_j}^2} (\norm{-x - \mu^*_j}_2^2 - \norm{x - \mu^*_j}_2^2)} f^*(x) \dx \\
&= \int_{R_j} \exp\left( -2 \pi \iprod{x, \mu^*_j}/{\sigma^*_j}^2 \right) f^*(x) \dx = \int_{R_j} \exp\left( -2\pi \norm{\mu_j}_2 \iprod{x, \ejj}/{(\sigma^*_j)}^2 \right) f^*(x) \dx\\
& \le \exp(- \sqrt{2 \pi} \sigma^2_j \gamma/{\sigma^*_j}^2) \int_{R_j} f^*(x) \dx, ~~ \text{since }\iprod{x, \ejj} \ge \frac{\sigma_j}{\sqrt{2\pi}} \text{ for } x \in R_j.
\end{align*}
\noindent Hence from \eqref{eq:equalmass} and \eqref{eq:diffmassfromdiffmeans} either
\[ \frac{\int_{L_j} f(x) \dx}{\int_{L_j} f^*(x) \dx} > 1+ \frac{ \sqrt{2 \pi}\gamma \sigma^2_j}{{\sigma^*_j}^2} ~\text{ or }~ \frac{\int_{R_j} f^*(x) \dx}{\int_{R_j} f(x) \dx} >  1+ \frac{ \sqrt{2 \pi} \gamma \sigma^2_j}{{\sigma^*_j}^2}, \]
which completes the proof, since $\int_{L_j} f(x) \dx= \int_{R_j} f(x) \dx \ge 1/8$.
\end{proof}

\begin{lemma} [Statistical Distance from Different Variances] \label{lem:diff:variances}
In the notation of Lemma~\ref{lem:statdiff}, suppose $\mu_j =\mu^*_j$, but $\sigma_j=(1-\eta){\sigma^*_j}$.
Then, there exists a set $S \subset S_j$ such that
\begin{equation} \label{eq:diff:variances}
 \int_{S} \Abs{w_j g_{\mu_j, \sigma_j}(x)- w^*_j g_{\mu^*_j, \sigma^*_j}(x)} \dx > \min\set{c_2 \eta,1},
\end{equation}
where $c_2>0$ is a universal constant.
\end{lemma}
\begin{proof}
Without loss of generality, let us assume that $\mu_j = \mu^*_j=0$ (by shifting the origin), and $\sigma_j=1$ (by scaling).
Let $R_1 = [\tfrac{1}{\sqrt{2\pi}}(\sqrt{d} - 2) , \tfrac{1}{\sqrt{2\pi}}(\sqrt{d}-1)]$ and $R_2 = [\tfrac{1}{\sqrt{2\pi}}(\sqrt{d} +1), \tfrac{1}{\sqrt{2\pi}}(\sqrt{d}+2)]$ and consider two annular strips $T_1 = \set{x : \norm{x} \in R_1}$ and $T_2 = \set{x :  \norm{x} \in R_2}$.

First, we note that using standard facts about the $\chi^2(d)$ distribution,
for some appropriately chosen universal constant $c_2>0$ we have
\[\int_{T_1} f(x) \dx \ge  c_2~,~~ \text{ and }~~ \int_{T_2} f(x) \dx \ge  c_2 .\]

We will show that there is a significant statistical difference between $f, f^*$ in either $T_1$ or $T_2$. Let us assume for contradiction that
\begin{equation}\label{eq:diff:variances:contr}
(1-\eta) \le \frac{\int_{T_1} f(x) \dx}{\int_{T_1} f^*(x) \dx}  \le (1+\eta) ~~ \text{ and } ~~(1-\eta) \le \frac{\int_{T_2} f(x) \dx}{\int_{T_2} f^*(x) \dx}  \le (1+\eta).
\end{equation}

Let us consider the range of values that $f$ takes in $T_1$ and $T_2$. We
\begin{align}
\forall x \in T_1 ~&  \frac{w_j}{\sigma_j^d} \exp\left(-(\sqrt{d} - 1)^2/2 \right) \le f(x) \le \frac{w_j}{\sigma_j^d} \exp\left(-(\sqrt{d}-2)^2/2 \right)\\
\forall x \in T_2 ~&  \frac{w_j}{\sigma_j^d} \exp\left(-(\sqrt{d}+2)^2/2 \right) \le f(x) \le \frac{w_j}{\sigma_j^d} \exp\left(-(\sqrt{d}+1)^2/2 \right)  \\
\forall x_1 \in T_1, x_2 \in T_2,~ & \frac{f(x_1)}{f(x_2)} \ge \exp\left(-\tfrac{1}{2} \left((\sqrt{d}-1)^2) - (\sqrt{d} +1)^2 \right)\right) \ge \exp(2 \sqrt{d}) \ge e^2.   \label{eq:enoughgapforcontradiction}
\end{align}

Let $A(r)$ be the $d-1$ dimensional volume of $\set{x \in S_j: \norm{x}=r}$. From \eqref{eq:diff:variances:contr}, we have that
\begin{align*}
\frac{\int_{T_1} f(x) \dx}{\int_{T_1} f^*(x) \dx} &= \frac{w_j \int_{r \in R_1} \exp(- r^2/ 2) A(r) \dr}{w^*_j (1-\eta)^{d}\int_{r \in R_1} \exp(- (1-\eta )^2r^2/ 2) A(r) \dr} \\
&=\frac{w_j \int_{r \in R_1} \exp(- r^2/ 2) A(r) \dr}{w^*_j (1-\eta)^{d} \int_{r \in R_1} \exp\left((2\eta-\eta^2) r^2/2 \right) \cdot \exp(-r^2/ 2) A(r) \dr}
\\
&\ge \frac{w_j}{w^*_j (1-\eta)^{d}} \cdot \exp\left(-\eta (\sqrt{d} -1)^2/2 \right).\\
\text{Hence, }~ \frac{w_j}{w^*_j (1-\eta)^{d}} &\le (1+\eta) \exp\left(\eta (\sqrt{d} -1)^2/2 \right).
\end{align*}
\noindent Using a similar argument for $T_2$ we get
$$ \frac{w_j}{w^*_j (1-\eta)^{d}} \ge (1-\eta) \exp\left(\eta (\sqrt{d} +1)^2/2 \right). $$
Combining the two we get
\begin{align*}
(1+\eta) \exp\left(\eta (\sqrt{d} -1)^2/2 \right) &\ge (1-\eta) \exp\left(\eta (\sqrt{d} +1)^2/2 \right)\\
\exp\left(2 \eta \sqrt{d} \right)\le (1+\eta)^2 \le e^{2 \eta},
\end{align*}
which is a contradiction.
\end{proof}

\begin{proof}[Proof of Lemma~\ref{lem:statdiff}]
Let $f_j(x)=w_j g_{\mu_j, \sigma_j}(x), f^*_j(x)=w^*_j g_{\mu^*_j, \sigma^*_j}(x)$ be the p.d.f. of the two components. From \eqref{eq:tvsep:cond} we know that there is some non-negligible separation in the parameters. Hence either the weights, or means, or variances are separated by at least $\Omega(\gamma)$.
We will now consider three cases depending on whether there is non-negligible separation in the means, variances or weights (in that order).

Set $\gamma_1:=\gamma^2/(256d)$, $\gamma_2:=c_2\gamma_1/2$ where $c_2$ is the constant in Lemma~\ref{lem:diff:variances}. Suppose there is some non-negligible separation in the means i.e., $\norm{\mu_j - \mu^*_j} \ge \gamma_2 \sigma_j/\sqrt{2\pi}$.
Lemma~\ref{lem:diff:means} shows that there is a set $S \subset S_j$ that has
$$\norm{f-f^*}_{1,S} \ge \frac{\gamma_2 \sigma^2_j}{ 8 {\sigma^*_j}^2} \ge \frac{\csd \gamma^2}{d\rho_s^2}, $$
where $\csd= c_2/(512 \sqrt{2\pi})$.

Otherwise we have that $\norm{\mu_j - \mu^*_j} < \gamma_2 \sigma_j/\sqrt{2\pi}$. Let $g^*_j$ be the p.d.f. of the Gaussian component $(w^*_j, \mu_j, \sigma^*_j)$. From Lemma~\ref{app:kl} we know that $\norm{f^*_j-g^*_j}_1 \le \gamma_2$. Now, $f_j$ and $g^*_j$ are p.d.f. of components that have the same mean.

If $\sigma^*_j -\sigma_j> \gamma_1 \sigma_j$. From Lemma~\ref{lem:diff:variances} we see that
$$\norm{f-g^*}_{1,S} \ge c' \gamma_1  \implies \norm{f-f^*}_{1,S} \ge c' \gamma_1 - \gamma_2 \ge c' \gamma_1/2.$$

Finally if $ \sigma^*_j -\sigma_j \le \gamma_1 \sigma_j$, then one can bound the statistical distance over $S=S_j$ between two components with equal means and variances, but different weights:
\begin{align*}
\norm{f_j-f^*_j}_{1,S} &\ge \norm{f_j-g^*_j}_{1,S} - \norm{g^*_j-f^*_j}=  \int_{S} \Abs{w_j g_{\mu_j, \sigma_j}(x)- w^*_j g_{\mu_j, \sigma_\ell}(x)} \dx - \gamma_2 \\
&\ge \int_{S} \Abs{w_j g_{\mu_j, \sigma_j}(x)- w^*_j g_{\mu_j, \sigma_j}(x)} \dx - \int_S \Abs{w^*_j g_{\mu_j, \sigma_j}(x)- w^*_j g_{\mu_j, \sigma^*_j}(x)} \dx - \gamma_2  \\
&\ge \abs{w_j - w^*_j} \int_S g_{\mu_j, \sigma_j}(x) - 2\sqrt{d \gamma_1}  - \gamma_2 ~~\qquad \text{( from Lemma~\ref{app:kl})}\\
& \ge \frac{2\gamma}{3} \left(1- (\wmin/\rhos)^{4c} \right) - \frac{\gamma}{8} - \frac{\gamma}{8} \ge \gamma/4 ,
\end{align*}
where the last line follows from Lemma~\ref{lem:Sj}.
\end{proof}

We now complete the proof of the main proposition of this section, that lower bounds the statistical difference between two mixtures of well-separated Gaussians which differ in their parameters.
\begin{proof}[Proof of Proposition~\ref{prop:upperbounds}]

We follow the approach outlined earlier in the section. We first match up components from $(w_j, \mu_j, \sigma_j)\in \calG$ and $(w^*_j, \mu^*_j, \sigma^*_j) \in \calG^*$ that are at most $\gamma= k^{-c}$ close in parameter distance. From triangle inequality and well-separatedness of $\calG, \calG^*$, it is easy to see that one component from $\calG$ (or $\calG^*$) cannot be matched to more than one component from $\calG^*$ (or $\calG$ respectively). Let $\set{1,2,\dots,k'}$ be the indices of the components in $\calG$ and $\calG^*$ that are unmatched, for $\ell \in \set{k'+1,\dots,k}$, let the component $(w_\ell, \mu_\ell, \sigma_\ell)$ of $\calG$ be matched to the component $(w^*_\ell, \mu^*_\ell, \sigma^*_\ell)$ of $\calG^*$. Since the parameter distance $\dparam(\calG, \calG^*) \ge k^{-c+1}$, we have that $k' \ge 1$.

Consider among the unmatched components in both $\calG$ and $\calG^*$, the one with the smallest variance: let this component be $(w_j,\mu_j,\sigma_j)$ from $\calG$ without loss of generality. From Lemma~\ref{lem:centermapping}, we know that at most one component of $\calG^*$ satisfies \eqref{eq:centers}. Again, without loss of generality, let $(w^*_j,\mu^*_j,\sigma^*_j)$ be this component of $\calG^*$ (if it exists). Hence
$$ \forall \ell \in [k'], \ell \ne j, ~ \norm{\mu_j - \mu^*_\ell}_2 > 4c \sqrt{\log(\rhos/\wmin)} \sigma^*_\ell \ge 2c \sqrt{\log(\rhos/\wmin)} (\sigma_j+\sigma^*_\ell),$$
where the last inequality follows because $\sigma^2_j$ is the smallest variance.\footnote{This is in fact the only point where we use the fact that we picked the small variance Gaussian.}
Further, all the matched components $\ell \in \set{k'+1, \dots k}$ in $\calG^*$ are all far from $\mu_j$ since
\begin{align*}
\forall \ell \in \set{k'+1, \dots, k}, ~ \norm{\mu_j - \mu^*_\ell}_2 &\ge \norm{\mu_j - \mu_\ell}_2 - \norm{\mu_\ell - \mu^*_\ell}\\
& \ge 4c \sqrt{\log(\rhos/\wmin)} (\sigma_j+\sigma_\ell) - \gamma \\
&\ge 4c \sqrt{\log(\rhos/\wmin)} (\sigma_j+\sigma^*_\ell) - \gamma(1+4c\sqrt{\log(\rhos/\wmin)})\\
&\ge 2c \sqrt{\log(\rhos/\wmin)} (\sigma_j+\sigma_\ell),
\end{align*}
since $\gamma< 2/(5\rhos)$. From Lemma~\ref{lem:othercomponents}, there is negligible contribution from the rest of the components (using $c \ge 3$):
\begin{equation}\sum_{\ell \in [k]: \ell \ne j} w^*_\ell \int_{S_j} g_{\mu^*_\ell, \sigma^*_\ell}(x) \dx < \ctails \left(\frac{\wmin}{\rhos}\right)^{4c}. \label{eq:propupper2}
\end{equation}

\noindent From Lemma~\ref{lem:statdiff} there is a subset $S \subset S_j$ where there is significant statistical distance
$$ \int_{S} \Abs{w_j g_{\mu_j, \sigma_j}(x)- w^*_j g_{\mu^*_j, \sigma^*_j}(x)} \dx > \frac{\csd \gamma^2}{d\rho_s^2}.$$

\noindent Combining the last two equations, we have
\begin{align*}
\norm{f - f^*}_1 &\ge \int_{S}  \Abs{ \sum_{\ell} w_\ell g_{\mu_\ell, \sigma_\ell}(x) - \sum_{\ell} w^*_\ell  g_{\mu^*_\ell, \sigma^*_\ell}(x) } ~\dx\\
 & \ge \int_{S} \Abs{ w_j g_{\mu_j, \sigma_j}(x)- w^*_j g_{\mu^*_j, \sigma^*_j}(x) }  \dx - \sum_{\ell \ne j} \int_{S} w^*_\ell g_{\mu^*_\ell, \sigma^*_\ell}(x)\dx \\
 & \ge  \frac{\csd \gamma^2}{d\rho_s^{4}} - \ctails \left(\frac{1}{\rhos k}\right)^{4c} \ge \frac{c' \gamma^2}{d\rho_s^{2}}
\end{align*}
for some universal constant $c'>0$, since $\gamma \ge k^{-c}$.
\end{proof}


\section{Efficient Algorithms in Low Dimensions} \label{sec:algorithms}


In this section, we give a computationally efficient algorithm that works in $d=O(1)$ dimensions, and learns the mixture of $k$ spherical Gaussians even when the separation between centers is $O(\sigma)$. In comparison, previous algorithms need separation of the order of $\Omega(\sigma \sqrt{\log (k/\delta)})$. We prove the following theorem.
\begin{theorem}\label{thm:lowdims}
There exists universal constants $c>0$ such that the following holds. Suppose we are given samples from a mixture of spherical Gaussians $\calG=\set{\parens{w_j, \mu_j , \sigma_j}: j \in [k]}$, where the weights and covariances are known, such that $\norm{\mu_j} \le \rho~ \forall j \in [k]$ and
\begin{equation}
\forall i,j \in [k], i \ne j:~ \norm{\mu_i - \mu_j}_2 \ge c \Paren{\sqrt{d}+\sqrt{\log (\rhos \rhow)} }\cdot (\sigma_i + \sigma_{j}).
\end{equation}
For any $\delta>0$, there is an algorithm using time (and samples) $\poly\Paren{\wmin^{-1},\delta^{-1},\rho, \rho_\sigma}^{O(d)}$ that with high probability recovers the means up to $\delta$ accuracy i.e. finds for each $j \in [k]$, $\tmu_j$ such that 
$\norm{\tmu_j  - \mu_j}_2 \le \delta \sigma_j$.
\end{theorem}
In the above theorem, when both $\rho_w, \rho_\sigma=O(1)$ as in the case of uniform mixtures, this corresponds to a separation of order $\Omega(\sqrt{d})$.

The above theorem follows by applying the guarantees of the iterative algorithm (Theorem~\ref{thm:boxiterative}) along with a computationally efficient procedure that finds appropriate initializers. The following theorem shows how 
to find reasonable initializers for $\mu_j, \sigma_j, w_j$ for each of the $k$ components. 

\begin{theorem}\label{thm:initializer}
Let $c_0 \ge 2$ be any constant, and suppose $\eps_0 = \exp(-c_0 d)$.
There is an algorithm running in time $\Paren{\frac{\rho d}{\eps_0^3}}^{O(d)} \poly(1/\wmin)$ that given samples from a $\rho$-bounded mixture of $k$ spherical Gaussians $\calG=\set{(w_j, \mu_j,\sigma_j): j \in [k]}$ in $d$ dimensions satisfying
\begin{equation}
\forall i \ne j \in [k], ~ \norm{\mu_i - \mu_j}_2 \ge 4 c_0 \Paren{\sqrt{d}+\sqrt{\log (\rho_w \rho_\sigma)}}(\sigma_i+\sigma_j),
\end{equation}
can find with high probability $\set{(\tmu_j, \tsigma_j, \tw_j): j \in [k]}$ s.t.
$$\forall j \in [k], \norm{\tmu_j - \mu_j}_2 \le \eps_0 \sigma_j \sqrt{d}, ~ \abs{\tsigma_j -\sigma_j} \le \eps_0 \sigma_j, \text{ and } \abs{\tw_j - w_j} \le \eps_0 w_j.$$
\end{theorem}
Note that the above theorem also finds initializers for the weights and variances when they are unknown. Hence, Theorem~\ref{thm:lowdims} will also apply to the setting with unknown weights and variances if we get similar guarantees for Theorem~\ref{thm:boxiterative} (see Remark~\ref{remark:unknown}). 

To prove Theorem~\ref{thm:initializer}, we will first find estimates $\tmu_j$ for the means $\mu_j$ (Proposition~\ref{prop:initial}), and then obtain good estimates $\tsigma_j, \tw_j$ for each $j \in [k]$. Proposition~\ref{prop:initial} already suffices in the case of known weights and variances.

\begin{proposition}\label{prop:initial}
In the notation of Theorem~\ref{thm:initializer}, there is an algorithm running in $\Paren{\frac{\rho}{\eps_0^3 \wmin}}^{O(d)}$ time that finds w.h.p.
$\tmu_1, \dots, \tmu_k$ s.t. $\forall j \in [k], \norm{\tmu_j - \mu_j}_2 \le \eps_0 \sigma_j \sqrt{d}$.
\end{proposition}

We first start with expressions for the first derivative (gradient) and second derivative (Hessian) at a point $x$ in terms of the model parameters. 
\begin{align}
f'(x)= \grad f(x) &= - 2\pi \sum_{j=1}^k \frac{w_j}{\sigma_j}\cdot g_{\mu_j, \sigma_j}(x) \cdot \frac{(x - \mu_j)}{\sigma_j} \label{eq:pdf:grad}\\
f''(x)=\grad^2 f(x) &= 4\pi^2 \sum_{j=1}^k \frac{w_j}{\sigma_j^{2}}\cdot g_{\mu_j, \sigma_j}(x) \cdot \left( \frac{1}{\sigma_j^2}(x - \mu_j)^{\otimes 2} - \frac{1}{2\pi} I_{d \times d} \right) \label{eq:pdf:hessian}.
\end{align}
Note that using $\poly_d(k,\rho)$ samples, we have access to the p.d.f. $f(x)$ and its derivatives $f'=\grad f$, $f''=\grad^2 f$ up to $1/\poly_d(k)$ accuracy at any point $x$. 

The algorithm will consider a $\delta$-net of points, and find ``approximate local-maxima'' of the p.d.f., which are defined as follows.
\anote{Changed the derivative condition.}
\begin{definition}[Approximate Local Maximum]\label{def:approx-max}
Consider a mixture of $k$ spherical Gaussians in $\R^d$ with parameters $\set{(w_j,\mu_j, \sigma_j):j \in [k]}$ and p.d.f. $f$. Then $x \in \R^d$ is an approximate local-maxima iff 
\begin{align} \label{eq:approx-max}
&\text{(i) } f(x) \ge \frac{\wmin}{ 2 \smax^d} ,&~ \text{(ii) }\norm{f'(x)}_2 \le \frac{\pi \eps_0 \sqrt{d} \smin}{4\smax^2} f(x) ,\\
&\text{(iii) } f''(x)= \grad^2 f (x) \preceq -\frac{\pi}{2\smax^2} f(x) I_{d \times d} & .\nonumber
\end{align}
\end{definition}

To show the above proposition, we will show that all approximate local-maxima are close to one of the means $\mu_j$ (Lemma~\ref{lem:initial:notmaxima} and Lemma~\ref{lem:initial:notcritical}), and there is at least one such approximate local maxima near each mean $\mu_j$ (Lemma~\ref{lem:initial:goodstart}). Further, since the parameters are separated, this will allow us to pick all approximate local-maxima in a net, cluster them geometrically and pick one such point from each cluster to get good initializers for each mean $\mu_j$. These statements will also allow for some slack to tolerate estimation errors.

We start with some inequalities that use the separation between means to show the p.d.f. and the first few ``moments'' near one of the means $\mu_j$ is dominated by the the $j$th component. By applying Lemma~\ref{lem:smallleakage} we get that at any point $x \in \R^d$ such that $\norm{x-\mu_j}_2 \le \sigma_j \sqrt{d/\pi} $
\begin{align}
\sum_{i \ne j \in [k]}w_i g_{\mu_i, \sigma_i}(x) &< w_j g_{\mu_j,\sigma_j}(x) \exp(-c_0 d) \parens[\Big]{\frac{\smin}{\smax}}^2. \label{eq:initial:leak0}\\
\sum_{i \ne j \in [k]}w_i \cdot \frac{\norm{x-\mu_i}_2}{\sigma_i} \cdot g_{\mu_i, \sigma_i}(x) &< w_j g_{\mu_j,\sigma_j}(x) \exp(-c_0 d)\cdot \parens[\Big]{\frac{\smin}{\smax}}^4.\label{eq:initial:leak1}\\
\sum_{i \ne j \in [k]}w_i \cdot \frac{\norm{x-\mu_i}^2_2}{\sigma^2_i} \cdot g_{\mu_i, \sigma_i}(x) &< w_j g_{\mu_j,\sigma_j}(x) \exp(-c_0 d)\cdot\parens[\Big]{\frac{\smin}{\smax}}^4 \label{eq:initial:leak2}.
\end{align}

The following lemma shows that any point that is far from all of the means is {\em not} an approximate local maximum (does not satisfy condition (iii) of Def.~\ref{def:approx-max}).

\begin{lemma}\label{lem:initial:notmaxima}
In the notation of Proposition~\ref{prop:initial}, for any $x \in \R^d$ that satisfies
$$ \forall j \in [k] ~~\norm{x-\mu_j}_2 >  \sigma_j \sqrt{\frac{d}{\pi}}, \exists u \in \R^d \text{ s.t. }u^T H_{x} u >0, $$
where where $H_{x}=f''(x)$ represents the Hessian evaluated at $x$.
Hence, such a point $x$ is not an approximate local maxima. 
\end{lemma}
\begin{proof}
Let $H_{x,j} = (x-\mu_j)^{\otimes 2} - \sigma_j^2 I/(2\pi)$.
Then  
$$\forall j \in [k],~\tr(H_{x,j})= \norm{x - \mu_j}^2 - \frac{d \sigma_j^2}{2\pi} >  \frac{d \sigma_j^2}{2\pi} ~~\text{ since } \frac{\norm{x-\mu_j}_2}{\sigma_j} > \sqrt{\frac{d}{\pi}}.$$ 
Let $v$ a random unit standard Gaussian vector drawn from $\calN(0,1)^d$.
\begin{align*}
\frac{\E[v^T H_{x} v]}{\E [\norm{v}_2^2]}&= \frac{\tr(H_{x})}{d}= \frac{4 \pi^2}{d} \sum_{j=1}^k \frac{w_j}{\sigma_j^{d+4}}\cdot \exp\Paren{- \frac{\pi \norm{x-\mu_j}^2}{\sigma_j^2}} \cdot \tr(H_j) \\
& >\frac{2 \pi}{d} \sum_{j=1}^k \frac{w_j}{\sigma_j^d}\exp\Paren{- \frac{\pi \norm{x-\mu_j}^2}{\sigma_j^2}} \cdot \frac{d}{\sigma_j^2} > 2\pi \cdot f(x) \cdot \frac{1}{\max_j \sigma_j^2}.
\end{align*}
Hence, there is a direction $v^T H_{x} v \ge 2 \pi f(x) \norm{v}_2^2/\smax^2$.
\end{proof}

The following lemma shows that we cannot have approximate local maxima (or more generally, critical points) whose distance from $\mu_j$ is between $[\eps_0 \sqrt{d} \smin,  \sqrt{d/\pi}\cdot \sigma_j]$. Hence, together with Lemma~\ref{lem:initial:notmaxima}, this shows that every approximate local maximum is within $\eps_0 \sqrt{d} \smin$ from one of the true means. 

\begin{lemma}\label{lem:initial:notcritical}
In the notation of Proposition~\ref{prop:initial}, for any $\eps_0 = 2\exp(-c_0 d)$ with $c_0 \ge 2$, suppose $x \in \R^d$ satisfies for some component $j \in [k]$, $\eps_0 \sqrt{d} \smin< \norm{x-\mu_j}_2 \le  \sqrt{d/\pi} \sigma_j$, then 
\[ \norm{f'(x)}_2=\norm{\grad f(x)}_2 > \frac{\eps_0 \sqrt{d} \smin}{4\sigma_j^2} \cdot f(x) \ge \frac{\eps_0 \sqrt{d} \smin}{4\smax^2} \cdot f(x) . \]
\end{lemma}
\begin{proof}
From \eqref{eq:pdf:grad}, the first derivative satisfies
\begin{align}\label{eq:notcritical:temp}
\norm{f'(x)}_2 &\ge \frac{2 \pi w_j \norm{x - \mu_j}_2}{\sigma_j^2} \cdot g_{\mu_j, \sigma_j}(x)- \sum_{i \ne j \in [k]}\frac{2\pi w_i}{\sigma_i}\cdot \frac{\norm{x-\mu_i}_2}{\sigma_i} \cdot g_{\mu_i, \sigma_i}(x).
\end{align}

\noindent Applying \eqref{eq:initial:leak1} with the given separation,
\begin{equation}\label{eq:notcritical:temp2}
\sum_{i \ne j \in [k]}\frac{w_i}{\sigma_i}\cdot \frac{\norm{x-\mu_i}_2}{\sigma_i} \cdot g_{\mu_i, \sigma_i}(x) < \frac{1}{\smin} \cdot w_j g_{\mu_j,\sigma_j}(x) \exp(-c_0 d)\cdot \parens[\Big]{\frac{\smin}{\smax}}^2 < \frac{ \eps_0 w_j \smin}{2\smax^2} g_{\mu_j,\sigma_j}(x).
\end{equation}

\noindent Further, $\norm{x - \mu_j}_2 \ge \eps_0 \sqrt{d} \sigma_j$. Using \eqref{eq:notcritical:temp} and \eqref{eq:notcritical:temp2}, 
\begin{align*}
\norm{f'(x)}_2 &\ge 2\pi w_j g_{\mu_j, \sigma_j}(x) \cdot \frac{\norm{x - \mu_j}_2^2}{\sigma_j^2} -  \sum_{i \ne j}^k 2\pi w_i g_{\mu_i, \sigma_i}(x) \cdot \frac{\norm{x - \mu_i}_2^2}{\sigma_i^2} \\
&\ge \frac{2\pi \norm{x - \mu_j}_2}{\sigma_j^2} \cdot w_j g_{\mu_j, \sigma_j}(x) - \frac{ \pi \eps_0 w_j \smin}{\smax^2} g_{\mu_j,\sigma_j}(x) ~~~\text{ (from \eqref{eq:notcritical:temp2})},\\
& \ge \frac{\pi \eps_0 \sqrt{d} \smin}{4\smax^2} \cdot f(x), 
\end{align*}
where the last inequality follows from \eqref{eq:initial:leak0} and using $\norm{x-\mu_j}_2 > \eps_0 \sqrt{d} \smin$.
\end{proof}

We now proceed to the proof of Lemma~\ref{lem:initial:goodstart}, which shows that any point that is sufficiently close to one of the component means is an approximate local maxima. This shows that in any $\delta$-net (for sufficiently small $\delta < \eps_0 \sqrt{d} \smin^2/\smax$) , there will be an approximate local maxima.

\begin{lemma}\label{lem:initial:goodstart}
In the notation of Proposition~\ref{prop:initial}, for any $\eps' \le \exp(-c_0 d)$ with $c_0 \ge 2$ and $\forall j \in [k]$, any point $x \in \R^d$ with $\norm{x - \mu_j}_2 \le \tfrac{\eps' \sqrt{d}}{32}\cdot \tfrac{\sigma_j^2}{\smax}$ is an approximate local maxima. In particular,
\begin{align} \label{eq:goodstart}
&\text{(i) } f(x) \ge \frac{3\wmin}{4 \smax^d} ,&~ \text{(ii) }\norm{f'(x)}_2 \le \frac{4\pi \eps' \sqrt{d}}{\smax} f(x) ,\\
&\text{(iii) } f''(x)= \grad^2 f (x) \preceq -\frac{\pi}{\sigma_j^2} f(x) I \preceq -\frac{3\pi}{4\smax^2} f(x) I  & .\nonumber
\end{align}
\end{lemma}
\begin{proof}[Proof of Lemma~\ref{lem:initial:goodstart}]
The lemma follows in a straightforward way from \eqref{eq:initial:leak0}, \eqref{eq:initial:leak1}, \eqref{eq:initial:leak2}, since $f, f', f''$ at $x$ are dominated by the $j$th component. Firstly by considering just the contribution to the p.d.f. from the $j$th component, the lower bound on $f(x)$ follows. Now we bound $\norm{f'(x)}_2$.
\begin{align*}
\norm{f'(x)}_2 & \le   2\pi \sum_{i=1}^k w_i g_{\mu_i, \sigma_i}(x) \cdot \frac{\norm{x - \mu_i}_2}{\sigma_i^2} \le 4 \pi w_j g_{\mu_j, \sigma_j}(x) \cdot \frac{\norm{x - \mu_j}_2}{\sigma_j^2} ~\text{ (from \eqref{eq:initial:leak1})}. \\
&\le \frac{4\pi \eps' \sqrt{d}}{\smax} \cdot w_j g_{\mu_j, \sigma_j}(x) \le \frac{4\pi \eps' \sqrt{d}}{\smax} \cdot f(x).
\end{align*}
\noindent 
We argue about $f''(x)$ similarly.
Suppose we use $M$ to denote the following matrix, and $\norm{M}$ to represent its maximum singular value,
\begin{align*}
M& =4\pi^2 \sum_{i \ne j}  w_i g_{\mu_i, \sigma_i}(x)\cdot \frac{1}{\sigma_i^{2}} \left( \frac{1}{\sigma_i^2}(x - \mu_i)^{\otimes 2} - \frac{1}{2\pi} I_{d \times d} \right).\\
\norm{M} &\le 4\pi^2 \sum_{i \ne j} w_i g_{\mu_i, \sigma_i}(x) \cdot \frac{1}{\sigma_i^{2}}\Paren{ \frac{\norm{x - \mu_i}_2^2}{\sigma_i^2} + \frac{1}{2\pi} }  \\
&\le \frac{4 \pi^2}{\smax^2} w_j g_{\mu_j, \sigma_j}(x) \exp(-c_0 d) < \frac{\pi}{2\sigma_j^2} w_j g_{\mu_j, \sigma_j}(x),
\end{align*}
where the last line follows from \eqref{eq:initial:leak0}, \eqref{eq:initial:leak2} and since $\exp(-c_0 d)< 1/(8\pi)$.  
Further $2 \pi \norm{x - \mu_j}^2_2/\sigma_j^2 < 1/4$. Substituting in \eqref{eq:pdf:hessian}, we get
\begin{align*}
f'' (x) &\preceq \frac{2 \pi}{\sigma_j^2} w_j g_{\mu_j, \sigma_j}(x) \Paren{ - I + \frac{2\pi\norm{x - \mu_j}_2^2}{ \sigma_j^2} I } + \norm{M} I \\
&\preceq -\frac{2 \pi}{\sigma_j^2} w_j g_{\mu_j, \sigma_j}(x) \Paren{I - \tfrac{1}{4}I - \tfrac{1}{4} I} \preceq \frac{- \pi}{\sigma_j^2} w_j g_{\mu_j, \sigma_j}(x) \preceq \frac{- 3\pi}{4\sigma_j^2} f(x),
\end{align*}
where the last inequality follows from \eqref{eq:initial:leak0}.
\end{proof}

%
%
%
%
%
%
%
%
%
%
%
%

We now proceed to the algorithm and proof of Proposition~\ref{prop:initial}.

\begin{proof}[Proof of Proposition~\ref{prop:initial}]
The algorithm first considers a $\delta$-net ${\calX_\delta}$ in $\R^d$ over a ball of radius $2 \rho$, and estimates $f(y)$ up to additive accuracy $\gamma$ where $\gamma = \wmin \smax^{-(d+4)} \eps_0^3 \delta^2/4$ and $\delta =\eps_0 \sqrt{d}\smin^3/(64\smax^2)$ will suffice\footnote{For our purposes, it suffices to have estimates of $f, f', f''$ at each of the points of the net ${\calX_\delta}$. Hence we can just take histogram counts in a small ball around each of the net points, and estimate $f$ up to accuracy $\gamma$ with $\abs{\calX_\delta}/\delta^2$ samples. Similarly the derivatives can also estimated using $\poly(\abs{{\calX_\delta}}, \tfrac{1}{\delta}, \tfrac{1}{\gamma})$ samples.}. Similarly, we can also estimate $f'(y)$ in $\ell_2$ norm, and $f''(y)$ in operator norm within additive accuracy $\gamma$. The size of the net is $\abs{{\calX_\delta}} = (\rho/ \delta)^{O(d)}$, and the sample complexity is $O(1/\gamma^2) \cdot (\rho/\delta)^{O(d)}$ samples. 

From Lemma~\ref{lem:initial:notcritical} and Lemma~\ref{lem:initial:notmaxima} we have that if
\begin{equation} \label{eq:init:soundness}
f(x) \ge \wmin \smax^{-d}/2, ~~ f''(x) \preceq -\frac{\pi}{2\smax^2} f(x) I, ~\text{and }~ \frac{\norm{f'(x)}_2}{f(x)} \le \frac{\pi \eps_0 \sqrt{d}\smin}{4\smax^2},
\end{equation}
then there exists $j \in [k]$ s.t. $\norm{x-\mu_j}_2 \le \eps_0 \sqrt{d} \smin$. On the other hand, applying Lemma~\ref{lem:initial:goodstart} with $\eps'=\eps_0 \smin/(32\smax)$, any point that is within $O(\eps_0 \sqrt{d}\smin^3/\smax^2)$ close to $\mu_j$ satisfy
\begin{equation}\label{eq:init:completeness}
f(x) \ge \frac{3\wmin}{4 \smax^d} ,~~ f''(x)= \grad^2 f (x) \preceq -\frac{3\pi}{4\smax^2} f(x) I , ~\text{and }~\frac{\norm{f'(x)}_2}{f(x)} \le \frac{\pi \eps_0 \sqrt{d}\smin}{8\smax^2}. 
\end{equation}

Our accuracy $\gamma$ of estimating $f,f',f''$ is chosen so that we can distinguish between the bounds in \eqref{eq:init:soundness} and \eqref{eq:init:completeness}. For convenience, since we have sufficiently accurate estimates, we will abuse notation and also use $f(x), f'(x), f''(x)$ to represent the estimate of the $f, f',f''$ at $x$.

First using our estimates, we consider all points
$$ T = \Big\{y \in {\calX_\delta} \mid f(y) \ge \wmin \smax^{-d}/2, ~ f''(y) \preceq -\frac{\pi}{2\smax^2} f(y) I, ~~ \frac{\norm{f'(y)}_2}{f(y)} \le \frac{\pi \eps_0 \sqrt{d}\smin}{4\smax^2}\Big\}.$$
We can find $T$ from our estimates since $\gamma < \wmin \smax^{-(d+2)}/8$. From \eqref{eq:init:soundness}, we have that for every $y \in T$, there is some $j \in [k]$, such that $\norm{y-\mu_j}_2 \le \eps_0 \sqrt{d} \smin$.  

Further the means are well separated i.e., $\norm{\mu_i - \mu_j}_2 > 4 (\sigma_i + \sigma_j)$. Hence, suppose we define
$$ T^*_j = \set{y \in T \mid \norm{\mu_j -y}_2 \le \eps_0 \sqrt{d} \sigma_j },$$
the sets $\set{T^*_j: j \in [k]}$ are disjoint and form a partition of $T$ that is consistent with the $k$ components of the Gaussian mixture $\calG$.
From the separation conditions we see that the distances between any two points in the same cluster $T^*_j$ are smaller than the distance between any two points in different clusters $T^*_{j_1}$ and $T^*_{j_2}$ ($j_1 \ne j_2$). Hence, we can use single-linkage clustering algorithm (see Awasthi et al.~\cite{ABS12} for a proof that single-linkage algorithm suffices)
to find the clustering $\set{T^*_j: j \in [k]}$ in time $\poly(|T|)\le \poly(|{\calX_\delta}|)$.

Finally, for each $j \in [k]$, let $\tmu_j$ be any point in $T^*_j$. Since, the coarseness of the net $\delta$ satisfies $\delta< \eps_0 \sqrt{d} \smin^3/(64\smax^2)$, we have from \eqref{eq:init:completeness} that there is at least one point $y \in T$ close to each $\mu_j$. Hence, $\norm{\tmu_j - \mu_j}_2 \le  \eps_0 \sqrt{d} \smin$, as required.


\end{proof}

\noindent { {\em Note.} In fact, the above proposition can also be used to show that $f$ has exactly $k$ local maxima $r_1, r_2 , \dots, r_k$, such that there is a unique $r_j$ satisfying $\norm{r_j - \mu_j}_2 \le \eps \sigma_j \sqrt{d}$. This is by using a quantitative version of the inverse function mapping theorem with the function $h(x)=\grad f(x)=0$ (one can use the Newton method as in Section~\ref{sec:amplify}).

\vspace{5pt}
Again, in what follows, when it is clear that we have sufficiently accurate estimates, we will abuse notation and also use $f(x), f'(x), f''(x)$ to represent the estimate of the $f,f',f''$ at $x$.

\begin{lemma}\label{lem:initial:sigma}
Assume the conditions in Theorem~\ref{thm:initializer}, and let $\eps_0=\exp(-c_0 d)$. Suppose we have $\tmu_j \in \R^d$ such that $\norm{\tmu_j - \mu_j}_2 \le \exp(-2c_0 d) \sigma_j$. Then there is an algorithm running in $\Paren{\frac{\rho}{\eps_0^3 \wmin}}^{O(d)}$ time, that w.h.p. finds $\tsigma_j \in R_+$ such that $\abs{\tsigma_j - \sigma_j} \le \eps_0 \sigma_j$.
\end{lemma}
\begin{proof}
Let $\kappa$ be a fixed number chosen so that $\kappa \le \sigma_j$, and pick any point $y \in \R^d$ such that
$\norm{y - \tmu_j}_2 = \tfrac{1}{\sqrt{\pi}}\kappa \sqrt{d}$.\footnote{We can guess such a $\kappa$ by either doing a binary search, or set it to be one-eighth of the diameter of cluster $T^*_j$ defined in proof of Proposition~\ref{prop:initial}} Based on the estimates of the p.d.f. at $\tmu_j, y$, we will set
$$\tsigma_j = \frac{\kappa \sqrt{d}}{\sqrt{\log\Paren{\frac{f(\tmu_j)}{f(y)}}}}.$$
Both $\tmu_j, y$ are in a ball of radius at most $\sigma_j \sqrt{d}/\pi$ around $\mu_j$. Further from Lemma~\ref{lem:smallleakage}, and since we have good estimate for the p.d.f. $f(\tmu_j)$ and $f(y)$ at these points, we have
$f(\tmu_j) = w_j \sigma_j^{-d} \Paren{1 \pm \exp(-2c_0 d)}$, and $f(y) = w_j \sigma_j^{-d} \exp\Paren{-\frac{\kappa^2 d}{\sigma_j^2}}\Paren{1 \pm \exp(-2c_0 d)}$. Dividing,
\begin{align*}
\Abs{\log\Paren{\frac{f(\tmu_j)}{f(y)}} - \frac{\kappa^2 d}{\sigma_j^2}} &\le 2\exp(-2c_0 d).
\end{align*}
Substituting for $\tsigma_j$ we have
\begin{align*}
\sigma_j^2 &= \frac{\tsigma_j^2 \log\Paren{\frac{f(\tmu_j)}{f(y)}}}{\log\Paren{\frac{f(\tmu_j)}{f(y)}} +\eta} ~~\text{ where } \abs{\eta}\le 2\exp (-2c_0 d) \\
&= \tsigma_j^2 \Paren{ 1 + \frac{\eta}{\log\Paren{\frac{f(\tmu_j)}{f(y)}}}}.
\end{align*}
Hence,  $\Abs{\tfrac{\sigma^2_j}{\tsigma_j^2} - 1}  \le \eps_0$,
where the last inequality follows from our choice of $y$, since $\log(f(\tmu_j)/f(y))\le d < \exp(c_0 d)$.
\end{proof}

\begin{lemma}\label{lem:initial:w}
Assume the conditions in Theorem~\ref{thm:initializer}, and let $\eps_0=\exp(-c_0 d)$. Suppose we have $\tmu_j \in \R^d, \tsigma_j \in \R^d$ such that $\norm{\tmu_j - \mu_j}_2 + \abs{\tsigma_j - \sigma_j} \le \exp(-2c_0 d) \sigma$. Then there is an algorithm running in $\Paren{\frac{\rho}{\eps_0^3 \wmin}}^{O(d)}$ time, that w.h.p. finds $\tw_j \in [0,1]$ such that $\abs{\tw_j - w_j} \le \eps_0 w_j$.
\end{lemma}
\begin{proof}
Let $\eta=w_j \exp(-2c_0 d)$. Let $c_d$ be the constant that is only dependent on $d$ given by
$$c_d = \int_{x \in \R^d: \norm{x}_2 \le \tfrac{1}{\sqrt{2\pi}}\sqrt{d}}  \exp\Paren{-\pi \norm{x}_2^2} \dx.$$
In fact, $c_d= \gamma(\tfrac{d}{2}-1,\tfrac{1}{2})$ is the incomplete Gamma function evaluated at $(\tfrac{d}{2} - 1, \tfrac{1}{2})$ which has an asymptotic approximation given in \cite[8.11(ii)]{DLMF}. Also $c_d \ge 2^{-d/2}/d$.
To get an estimate of $\tw_j$, we will consider the set
$$T_j=\set{x \in \R^d \mid \norm{x-\tmu_j}_2 \le \tfrac{1}{\sqrt{2\pi}}\sqrt{d} \tsigma_j}.$$
We will now generate $N=O(\rho\log(dk)/\eta^2)$ samples $x^{(1)}, \dots, x^{(N)}$ from the mixture of $k$ Gaussians and estimate the fraction of samples that are in $T_j$:\footnote{We could also integrate the estimated p.d.f. over the set $T_j$ to get this estimate. }
\begin{equation} \label{eq:initial:tw}
\tw_j = \frac{1}{c_d N} \sum_{\ell=1}^N \I[x^{(\ell)} \in T_j].
\end{equation}
\noindent From Lemma~\ref{lem:app:box:samplingerror}, we have small contribution from the other components
$$\Abs{\tw_j - \E[\tw_j]}= \Abs{\tw_j - \frac{1}{c_d} \int_{y \in T_j} f(y) \dy}\le \eta \le \exp(-2cd) w_j. $$
From Lemma~\ref{lem:smallleakage}, we have
\begin{align*}
\forall y \in T_j, ~ \Abs{f(y) - w_j g_{\mu_j, \sigma_j}(y) } & \le \exp(-2c_0 d) w_j g_{\mu_j, \sigma_j}(y) \\
\text{Hence, } \Abs{\E[\tw_j] -  \frac{w_j}{c_d} \int_{y \in T_j} g_{\mu_j, \sigma_j}(y) \dy} & \le \exp(-2c_0 d) \cdot \frac{w_j}{c_d} \int_{y \in T_j} g_{\mu_j, \sigma_j}(y) \dy < \frac{w_j}{c_d} \exp(-2c_0 d).
\end{align*}
Let $B = \set{y \mid \norm{y-\mu_j}_2 \le \tfrac{1}{2\pi}\sqrt{d}\sigma_j}$, and let $S_{d}$ be the surface area of the unit ball in $d$-dimensions (volume of the $d$-sphere). Since $\norm{\tmu_j - \mu_j} + \abs{\sigma_j - \tsigma_j} \le \exp(-2c_0 d) \sigma_j$, the probability mass on $T_j \setminus B$ is small
\begin{align*}
\Abs{ \int_{y \in T_j} g_{\mu_j,\sigma_j}(y) \dy - \int_{y \in B} g_{\mu_j,\sigma_j}(y) \dy} &\le \int_{\abs{\frac{\sqrt{2 \pi}\norm{y-\mu_j}_2}{\sigma_j} -1} \le \exp(-2c_0 d)}
g_{\mu_j, \sigma_j}(y) \dy  .\\
&\le S_d \Paren{\frac{d}{2\pi}}^{d/2} \times 2\exp(-2c_0 d) \le \exp(-2(c_0 -1)d).
\end{align*}
Further, the probability mass inside $B$ is given by $\frac{w_j}{c_d}\int_{y \in B} g_{\mu_j,\sigma_j}(y) \dy = w_j$. Hence,
\begin{align*}
\Abs{\E[\tw_j] - w_j} &\le  \Abs{\E[\tw_j] -  \frac{w_j}{c_d} \int_{y \in T_j} g_{\mu_j, \sigma_j}(y) \dy}+\Abs{ \frac{w_j}{c_d}\int_{y \in B} g_{\mu_j,\sigma_j}(y) \dy -  \frac{w_j}{c_d} \int_{y \in T_j} g_{\mu_j, \sigma_j}(y) \dy}\\
&\le \frac{w_j}{c_d} \exp(-2(c_0-1)d) + \frac{w_j}{c_d} \exp(-2c_0 d) \le w_j \exp(-c_0 d).
\end{align*}

\end{proof}

\begin{proof}[Proof of Theorem~\ref{thm:initializer}]
The proof follows by using Proposition~\ref{prop:initial}, followed by Lemma~\ref{lem:initial:sigma} and Lemma~\ref{lem:initial:w} in that order.
Set $\eps_0 = \exp(-c_0 d)$. First we use Proposition~\ref{prop:initial} to find w.h.p. initializers for the means $(\tmu_j :j \in [k])$ such that $\norm{\tmu_j - \mu_j}_2 \le \exp(-4 c_0 d) \smin$. Then using Lemma~\ref{lem:initial:sigma}, we find w.h.p. initializers $(\tsigma_j : j \in [k])$ such that $\abs{\tsigma_j - \sigma_j} \le \exp(-2c_0 d) \sigma_j$. Finally, these initializers $\tmu_j, \tsigma_j$ can be used in Lemma~\ref{lem:initial:w} to find w.h.p. $\tw_j~ \forall j \in [k]$ such that $\abs{\tw_j - w_j} \le \exp(-c_0 d) w_j$. By choosing the failure probability of at most $1/3k^2$ in each step, we see that the algorithm succeeds w.h.p. and runs in time $\Paren{\frac{\rho}{\eps_0^3 \wmin}}^{O(d)}$.
\end{proof}


\appendix
\noindent {\Large\bf Appendix}

\section{Standard Properties of Gaussians}
\begin{lemma}
Suppose $x \in \R$ be generated according to $N(0,\sigma^2)$, let $\phit(t)$ represents the probability that $x >t$, and let $\iphit{y}$ represent the quantile $t$ at which $\phit(t)\le y$.
Then 
\begin{equation}\label{eq:tails}
\frac{\frac{t}{\sigma}}{(\frac{t^2}{\sigma^2}+1)} e^{-\frac{t^2}{2 \sigma^2}}  \le \phit(t) \le \frac{\sigma}{t} e^{-\frac{t^2}{2 \sigma^2}}.
\end{equation}
Further, there exists a universal constant $c \in (1,4)$ such that
\begin{equation}\label{eq:tails2}
 \frac{1}{c} \sqrt{\log (1/y)} \le \frac{t}{\sigma}  \le c \sqrt{\log (1/y)}.
\end{equation}
\end{lemma}

\begin{lemma}\label{lem:app:truncmoments}
For any $\sigma >0$, $q \in \Z_+$ and any $\tau\ge 2q$, 
\begin{equation}\label{eq:truncmoments}
\int_{\abs{x} \ge \tau \sigma} |x|^q \exp(-\pi x^2/\sigma^2) \dx \le \sigma^q \exp(-2\tau^2).
\end{equation}
\end{lemma}

%

\begin{lemma}\label{app:kl}
Let $p,q$ correspond to the (weighted) probability density functions of the spherical Gaussian components in $d$ dimensions with parameters $(w_1,\mu_1,\sigma_1^2)$ and $(w_2,\mu_2,\sigma_2^2)$ respectively. Then
\begin{equation}
\norm{p-q}_1 \le 
  |w_1-w_2|+ 
	\min\set{w_1,w_2}\left(
	    \frac{\sqrt{2\pi}\cdot \norm{\mu_1 - \mu_2}_2}{ \sigma_2}+
			\frac{\sqrt{d\abs{\sigma_1^2 - \sigma_2^2}}}{\sigma_2}+
			\sqrt{2d \ln\parens[\Big]{\frac{\sigma_2}{\sigma_1}}} 
	\right).
\end{equation}
\end{lemma}
\begin{proof}
Without loss of generality let $w_2 \le w_1$. 
The KL divergence between any two multivariate Gaussian distributions with means $\mu_1, \mu_2$ and covariances $\Sigma_1, \Sigma_2$ respectively is given by~\cite{Duchinotes}
$$d_{KL}\left(N(\mu_1, \Sigma_1) \| N(\mu_2, \Sigma_2) \right)=\frac{1}{2} \left(\text{tr}(\Sigma_2^{-1} \Sigma_1)+
(\mu_1-\mu_2)^T \Sigma_2^{-1} (\mu_1 - \mu_2)-d +\ln\Big(\frac{\det(\Sigma_2)}{\det(\Sigma_1)} \Big)  \right).$$

\noindent Applying this to $p':=p/w_1,q':=q/w_2$ we get
\begin{align*}
d_{KL}(p' \| q') &= \frac{1}{2}\left(\frac{2 \pi \norm{\mu_1-\mu_2}_2^2}{\sigma_2^2}+ \frac{d\sigma_1^2}{\sigma_2^2} -d +2d\ln\parens[\Big]{\frac{\sigma_2}{\sigma_1}} \right)\\
&= \frac{\pi \norm{\mu_1-\mu_2}_2^2}{\sigma_2^2}+ \frac{d(\sigma_1^2 - \sigma_2^2)}{2\sigma_2^2}+d \ln\parens[\Big]{\frac{\sigma_2}{\sigma_1}}.
\end{align*}
Hence, by Pinsker inequality
\[ 
 \norm{p'-q'}_1 \le 
 \sqrt{2 d_{KL}(p' \| q')} \le  
 \frac{\sqrt{2\pi}\norm{\mu_1-\mu_2}_2}{\sigma_2}+
     \frac{\sqrt{d\abs{\sigma_1^2 - \sigma_2^2}}}{\sigma_2}+
		 \sqrt{2 d \log\parens[\Big]{\frac{\sigma_2}{\sigma_1}}}.
\]
By triangle inequality, 
$$\norm{p-q}_1 \le \norm{p-w_2 p'}_1+ \norm{w_2 p'-w_2 q'}_1 +\norm{w_2 q' - q}_1 \le |w_1 - w_2|+w_2 \norm{p'-q'}_1,$$
which gives the required bound. An identical proof works when $w_1 \le w_2$. 
\end{proof}

%
\paragraph{Higher Dimensional Gaussians and Approximations.}

Let $\gamma_d$ be the Gaussian measure associated with a standard Gaussian with mean $0$ and variance $1$ in each direction. 

Using concentration bounds for the $\chi^2$ random variables, we have the following bounds for the lengths of vectors picked according to a standard Gaussian in $d$ dimensions (see (4.3) in \cite{laurent2000}).
\begin{lemma}\label{lem:lengthconc}
For a standard Gaussian in $d$ dimensions (mean $0$ and variance $1/(2\pi)$ in each direction), and any $t>0$  
\begin{align*}
\Pr_{x \sim \gamma_d}\Big[ \norm{x}^2 \ge \frac{1}{2\pi}(d + 2 \sqrt{d t}+2 t) \Big] &\le e^{-t}.\\
\Pr_{x \sim \gamma_d}\Big[ \norm{x}^2 \le \frac{1}{2\pi}(d - 2 \sqrt{d t}) \Big] &\le e^{-t}.
\end{align*}
\end{lemma}

Similarly, the following lemma shows a simple bound for the truncated moments, when $x$ is generated according to $N(0,\sigma^2/2\pi)^d$. 
\begin{lemma}\label{lem:app:truncmoments:highd}
For any $\tau\ge q$, and any $q \in \Z_+$
\begin{equation}\label{eq:truncmoments:highd}
\int_{\norm{x}_2 \ge 2q \sqrt{d} \sigma} \norm{x}_2^q \exp(-\pi \norm{x}_2^2/\sigma^2) \dx \le \sigma^q \exp(-4d).
\end{equation}
\end{lemma}
\begin{proof}
Assume w.l.o.g that $\sigma=1$. For $\norm{x}_2 \ge 2q \sqrt{d}/\sqrt{2 \pi}$, $\norm{x}_2^q \le \exp(\pi \norm{x}_2^2/2)$. Hence,
\begin{align*}
\int_{\norm{x}_2 \ge 2q \sqrt{d}} \norm{x}_2^q \exp( - \pi \norm{x}_2^2 )\dx & \le \int_{\norm{x}_2 \ge 2q \sqrt{d}} \exp(-\pi  \norm{x}^2/2) \dx \\
&\le 2^{d/2} \int_{\norm{y}_2 \ge \frac{2q \sqrt{d}}{\sqrt{2}}} \exp\parens[\Big]{-\pi \norm{y}_2^2/2} d~y \le \exp(-5d+d/2) \le \frac{1}{16 \pi d^2},
\end{align*} 
where $y$ is distributed as a normal $d$-dimensional r.v. with mean $0$ and variance $1$ in each direction. 
\end{proof}

%

\onote{8/7: why is it here? also, the only place it's used is in the stuff I just added about the volume of a ball... or was it already used before?}
\anote{8/8 It's also used in the lower bound derivations involving Taylor expansion, but I don't refer to this explicitly.}
\begin{fact} [Stirling Approximation] \label{fact:stirling}
For any $n \ge 1$, $\sqrt{2 \pi n} \parens{n/e}^n \le n! \le e \sqrt{n} \parens{n/e}^n$.
\end{fact}

\section{Newton's method for solving non-linear equations}

We use a standard theorem that shows quadratic convergence of the Newton method in any normed space~\cite{AtkinsonHall} in the restricted setting where both the range and domain of $f$ is $\R^m$.

Consider a system of $m$ non-linear equations in variables $u_1, u_2, \dots, u_m$:
$$\forall j \in [m], f_j(u_1, \dots, u_m) = b_j. $$

Newton's method starts with an initial point $u^{(0)}$ close to a solution $u^*$ of the non-linear system. Formally, $u^{(0)} \in \calN$ where $\calN$ is an appropriately defined neighborhood set $\calN=\set{y: \norm{y - u^*} \le \eps_0}$. Let $F'(u)=J_f(u) \in \R^{m \times m}$ be the Jacobian of the system given by the non-linear functional $f:\R^{m} \rightarrow \R^{m}$, where the $(j,i)^{th}$ entry is the partial derivate $J_f(j,i)= \frac{\partial f_j (u)}{\partial u_i} \vert_{y}$ is evaluated at $y$. Additionally, for our algorithm, we assume that given any $y \in \calN$, we only have access to an estimates $\tb, \tF(u),\tF'(u)$ of vector $b \in \R^m$, $F(u) \in \R^m$ and $F'(u) \in \R^{m \times m}$ respectively \footnote{These errors in the estimate may occur due to sampling errors or precision errors.}.

Newton's method starts with the initializer $u^{(0)}$, and updates the solution using the iteration:
\begin{equation} \label{eq:newtoniteration}
u^{(t+1)} = u^{(t)} + \left(\tF'(u^{(t)})\right)^{-1} \left( \tb- f(u^{(t)})\right) .
\end{equation}

The convergence error will be measured in the $\ell_p$ norm for any $p \ge 1$. In what follows, $\norm{x}:=\norm{x}_p$ for $x \in \R^m$, $\norm{M}:=\norm{M}_{p \to p}$ for $M \in \R^{m \times m}$ and $\norm{T}:=\norm{T}_{p \times p \to p}$ for $T \in \R^{m \times m \times m}$. We first state a simple mean-value theorem that will be useful in the analysis of the Newton method, as well as in applying the guarantees in the context of mixtures of Gaussians. 

\begin{lemma}[Proposition 5.3.11 in \cite{AtkinsonHall}]\label{lem:meanvalue}
Consider a function $H:K \subset \R^{m_1} \to \R^{m_2}$, with $K$ being an open set. Assume $H$ is differentiable on $K$ and that $F'(u)$ is a continuous function of $u$ on $K$. Assume $u, w \in K$ and assume line segment from joining them is also contained in $K$. Then
$$\norm{F(u)-F(w)} \le \sup_{0\le \theta \le 1} \norm{F'((1-\theta) u +\theta w)}\cdot \norm{u -w}.$$ 
\end{lemma}

The following theorem gives robustness guarantees for Newton's method. It is obtained by using matrix perturbation analysis along with a standard theorem regarding the quadratic convergence of the Newton's method (see Theorem 5.4.1 in \cite{AtkinsonHall}). 

\begin{theorem}\label{thm:newton:errors}
Assume $u^* \in \R^m$ is a solution to the equation $F(u)=b$ where $F: \R^m \rightarrow \R^m$ such that $J^{-1}=(F')^{-1}$ exists in a neighborhood $\calN=\set{y: \norm{y- u^*} \le \norm{u^{(0)} - u^*}}$,
and $F':\R^m \rightarrow \R^{m \times m}$ is locally $L$-Lipschitz continuous in the neighborhood $\calN$ i.e.,
$$ \norm{F'(u) - F'(v)} \le L \norm{u - v} ~~ \forall u, v \in \calN. $$
Further, let the estimates $\tb, \tF(u),\tF'(u)$ satisfy for some $\eta_1, \eta_2, \eta_3>0$ and all $y \in \calN$
$$ \norm{\tb-b} \le \eta_1, \norm{\tF(u) - F(u)}\le \eta_2, \norm{\tF'(u)-F'(u)} \le \eta_3. $$
Then if $\eta_3 \norm{F'(u^{(t)})^{-1}} < 1$, $\norm{F'(u)} \le B$, then for all $u \in \calN$, the error $\eps_t = \norm{u^{(t)}-u^*}$ after the $t$ iterations of \eqref{eq:newtoniteration} satisfies
\begin{equation}
  \eps_{t+1} \le \eps_t^2 \cdot L\norm{F'(u^{(t)})^{-1}} + \norm{F'(u^{(t)})^{-1}}  \left(\eta_1 + \eta_2 + 4\eta_3 \eps_t \norm{F'(u^{(t)})^{-1}} B  \right) \label{eq:newton:errors}
\end{equation}
\end{theorem}
\anote{8/30:Rewrote the proof here.}
\begin{proof}

We sketch the proof here. For convenience, let us use $A:=F'(u^{(t)})$ and $\tA:=\tF'(u^{(t)})$ to denote the derivate and its estimate at $u^{(t)}$. Let $z^{(t+1)}$ be the Newton iterate after the $(t+1)$th step if we had the exact values of $F(u^{(t)}$ and $F'(u^{(t)})$ i.e., $z^{(t+1)}= u^{(t)}+ A^{-1} \Paren{b-F(u^{(t)})}$. 
From the standard analysis of the Newton method (see Theorem 5.4.1 in ~\cite{AtkinsonHall}),
$$\norm{z^{(t+1)} - u^*}  \le L \norm{A^{-1}} \norm{u^{(t)}-u^*}^2.$$ 
\anote{8/30: Did not write down this proof since it's repetitive. Could include those 3-4 lines if we want it to be self-contained.}
Further, the error between the actual Newton update $u^{(t+1)}$ and $z^{(t)}$ due to the estimates $\tA$ and $\tF(u^{(t)})$ is
\begin{align*}
u^{(t+1)} - z^{(t+1)} &= \tA^{-1} \Paren{\tb-\tF(u^{(t)})} - A^{-1} \Paren{b-F(u^{(t)})}  \\
&= (\tA^{-1} - A^{-1})\Paren{\tb-\tF(u^{(t)})} + A^{-1} \Paren{\tb- b + \tF(u^{(t)})-F(u^{(t)})}  \\
\norm{u^{(t+1)} - z^{(t+1)}} &\le \norm{\tA^{-1} - A^{-1}} \norm{\tb-\tF(u^{(t)})} + \norm{A^{-1}} (\eta_1+\eta_2).
\end{align*}
From perturbation bounds on matrix inverses~\cite{Bhatia}, if $\norm{A^{-1} E} <1$, 
$$\norm{A^{-1} - (A+E)^{-1}} \le \norm{A^{-1}} \cdot \frac{\norm{A^{-1} E}}{1- \norm{A^{-1} E}}.$$
Also from Lemma~\ref{lem:meanvalue}, $\norm{b-F(u^{(t)})} \le \norm{F'(u')} \norm{u^{(t)} - u^*}$ for some $u' \in \calN$. Substituting $E=\tA-A$,
\begin{align*}
\norm{u^{(t+1)} - z^{(t+1)}} &\le \norm{A^{-1}} \frac{\norm{A^{-1}(\tA- A)}}{1- \norm{A^{-1}(\tA- A)}} \Paren{\eta_1+\eta_2+\norm{b-F(u^{(t)})}} + \norm{A^{-1}} (\eta_1+\eta_2)\\
& \le  2\eta_3 \norm{A^{-1}}^2 \Paren{\eta_1+\eta_2+ \norm{F'(u')} \norm{u^{(t)}-u^*}} + \norm{A^{-1}} (\eta_1+\eta_2)\\
& \le 4 \eta_3 \eps_t \norm{A^{-1}}^2 \norm{F'(u')}+ \norm{A^{-1}} (\eta_1+\eta_2).
\end{align*}
\end{proof}
\anote{8/31: Added remark. Reword maybe?} 
\begin{remark}
While the above theorem requires that the derivative $F'$ is locally $L$-Lipschitz, this is a weaker condition than requiring a upper bound on the operator norm of the second derivative $F''$. Lemma~\ref{lem:meanvalue} shows that it also suffices if $\norm{F''(u)} \le L$ for all $u \in \calN$.  
\end{remark}

\begin{corollary}\label{corr:newton:errors}
Under the conditions of Theorem~\ref{thm:newton:errors}, there exists $0<\eps_0< \frac{1}{2L\norm{F'(u^{(t)})^{-1}} }$, such that for any given $\delta \in (0,1)$, there is an $\eta_1, \eta_2, \eta_3>0$ with
$$(\eta_1 + \eta_2) < \frac{\delta}{ 4\norm{F'(u^{(t)})^{-1}}} \text{ and } \eta_3< \frac{\delta}{4\norm{F'(u^{(t)})^{-1}}^2} \cdot \min\set{1, \frac{1}{B}}$$
such that after $T=\log\log(1/\delta)$ iterations of the Newton's method, we have
$$	\norm{u^{(T)}- u^*} \le \delta.  $$
\end{corollary}
\begin{proof}
For the given setting of $\eta_1, \eta_2, \eta_3$, we have $(\eta_1 + \eta_2) \norm{F'(u^{(t)})^{-1}} < \tfrac{\delta}{4}$, and  $B \norm{F'(u^{(t)})^{-1}}^2 \eps_t \eta_3 \le \delta/4$.
From Theorem~\ref{thm:newton:errors}, we have that for any $t$,
$$ \eps_{t+1} \le \eps_t^2 L \norm{F'(u^{(t)})^{-1}} + \frac{\delta}{2}.$$
Further $\eps_1 \le \eps_0^2 L \norm{F'(u^{(t)})^{-1}} < \frac{1}{4L \norm{F'(u^{(t)})^{-1}}}$. By induction, it follows that
$$ \eps_{t} \le \frac{2^{-2^{t}}}{L \norm{F'(u^{(t)})^{-1}}}.$$
Hence, this gives the required guarantee.
\end{proof}

\section{Dimension Reduction using PCA}\label{app:pca}
\newcommand{\barsig}{\bar{\sigma}}
\newcommand{\Atilde}{\tilde{A}}
\newcommand{\Utilde}{\tilde{U}}
\newcommand{\Vtilde}{\tilde{V}}

Here we give a proof of the assertion that for mixtures of spherical Gaussians, we can assume without loss of generality that $d \le k$. 

\begin{theorem}[Same as Theorem~\ref{thm:pca}]
 Let $\set{(w_i, \mu_i, \sigma_i): i \in [k]}$ be a mixture of $k$ spherical Gaussians that is $\rho$ bounded, and let $\wmin$ be the smallest mixing weight. Let $\mu'_1, \mu'_2, \dots, \mu'_k$ be the projections onto the subspace spanned by the top $k$ singular vectors of sample matrix $X \in \R^{d \times N}$. For any $\eps>0$, with $N= \poly(d, \rho, \wmin^{-1},\eps^{-1})$ samples we have with high probability
 $$ \forall i \in [k],~ \norm{\mu_i - \mu'_i}_2 \le \eps.$$ 
 \end{theorem}
\begin{proof}
Let $\delta =  \wmin \eps^4/ (2\rho^2) $ and $\eta = \delta^2/2$.
Let $A$ be the population average, i.e.,  
$$ A = \E [ x x^T] = M+ \barsig^2 I, \text{ where } M=\sum_i w_i \mu_i \mu_i^T, ~\text{and }  \barsig^2= \frac{1}{2 \pi}\sum_{i \in [k]} w_i \sigma_i^2.$$ 
Let $\lambda_1 \ge \lambda_2 \ge \dots \ge \lambda_d \ge 0$ be the 
eigenvalues of $M$
sorted in non-increasing order.
Since $M$ is of rank at most $k$, we have $\lambda_{k+1} = \cdots = \lambda_d = 0$. 
Let $r \le k$ be defined as the smallest index 
such that $\lambda_{r+1} < \delta$. 
Let $U$ be represent the orthogonal projector onto the top-$r$ eigenspace of $M$ (and hence of $A$ too), and $U^{\perp}=I-U$.
Notice that $\norm{U^{\perp} M U^{\perp}} = \lambda_{r+1} < \delta$. 
Then, using 
the positive semidefinite inequality $w_i U^\perp \mu_i \mu_i^T U^\perp \preceq U^\perp M U^\perp$, 
we obtain that
\[
  \norm{U^{\perp} \mu_i}_2^2 \le w_i^{-1} \bignorm{U^{\perp} M U^{\perp}}
	  =  \lambda_{r+1}/w_i < \delta / \wmin \; .
\]

Next, let $\Atilde$ represent the sample average $\Atilde=\frac{1}{N} X X^T$ 
(so $\Atilde$ converges to $A$ as $N \to \infty$).
Since $N \ge \poly(d)/\eta^2$, using standard concentration bounds (see Theorem 6.1.1 in \cite{Tropp2012} and related notes)
we have 
with high probability that $\norm{A - \Atilde} < \eta $. 
Let $\Utilde$ be the orthogonal projector onto the top-$k$ eigenspace of $\Atilde$, and let $\Vtilde=I - \Utilde$ be the orthogonal projector onto the bottom $(d-k)$ eigenspace of $\Atilde$. Hence for each $i \in [k]$, $\mu'_i = \Utilde \mu_i$. 
Notice that from Weyl's perturbation bounds for eigenvalues (see \cite{Bhatia}, Theorem III.2.1) $\lambda_{k+1}(\Atilde) \le \barsig^2+ \eta$. 
Therefore, the eigenvalues of $\Atilde$ corresponding to $\Vtilde$
(which are all at most $\barsig^2+ \eta$), 
and the eigenvalues of $A$ corresponding to $U$ 
(which are all at least $\barsig^2+ \delta$)
are separated by at least $\delta - \eta$. From standard perturbation bounds for eigenvectors (see \cite{Bhatia}, Theorem VII.3.1), we have
$$\norm{U \Vtilde } \le  \frac{\norm{A - \Atilde}}{\delta - \eta} \le \frac{\eta}{\delta - \eta} \le \delta .$$
Hence, for each $i \in [k]$, 
\begin{align*}
\norm{\mu_i - \mu'_i}_2^2 & = \norm{\Vtilde \mu_i}_2^2 = \iprod{\mu_i, \Vtilde \mu_i} \\
&= \iprod{U \mu_i, \Vtilde \mu_i} + \iprod{U^{\perp} \mu_i , \Vtilde \mu_i} = \iprod{\mu_i, U \Vtilde \mu_i} + \iprod{U^{\perp} \mu_i , \Vtilde \mu_i}\\
& \le \norm{U \Vtilde} \norm{\mu_i}_2^2 + \norm{U^{\perp} \mu_i}_2 \norm{\mu_i}_2\\
&\le \delta \norm{\mu_i}_2^2 + \sqrt{\frac{\delta}{\wmin}} \cdot \norm{\mu_i}_2  \le \eps^2
\end{align*}
by our choice of $\delta = \wmin \eps^4/(2\rho^2)$. 
\end{proof}

\section*{Acknowledgements}
The authors thank Santosh Vempala for suggesting the problem of
learning one-dimensional Gaussians, and other helpful discussions.
Part of this work was done when the second author was at the Courant
Institute and the Simons Collaboration on Algorithms and Geometry.

\bibliographystyle{plain}
\bibliography{aravind}

\end{document}